\documentclass[11pt]{article}
\usepackage[utf8]{inputenc}
\usepackage{style}
\usepackage{framed}
\usepackage{nicefrac}
\usepackage[parfill]{parskip}
\usepackage{setspace}
\usepackage{hyperref}
\usepackage{mdwlist}
\usepackage{mdframed}
\usepackage[ruled,linesnumbered]{algorithm2e}
\usepackage{relsize}
\usepackage[utf8]{inputenc}
\usepackage{style}
\usepackage{mdframed}
\usepackage{mdwlist}
\usepackage{nicefrac}
\usepackage[parfill]{parskip}
\usepackage{setspace}
\usepackage[colorlinks]{hyperref}
\usepackage[ruled,linesnumbered]{algorithm2e}
\usepackage{amsmath}
\usepackage{amssymb}
\usepackage{complexity}

\newcommand{\Ex}{\E}

\renewcommand{\cP}{{\cal P}}

\allowdisplaybreaks
\newcommand{\sym}{{\rm sym}}
\newcommand{\defeq}{\overset{\rm def}{=}}
\usepackage{microtype}
\newcommand{\vphi}{\varphi}

\usepackage{thm-restate}

\usepackage{boxedminipage}

\newcommand{\paren}[1]{\left(#1 \right )}

\newcommand{\abs}[1]{\left\lvert#1\right\rvert}

\newcommand{\Langle}{\left\langle}
\newcommand{\Rangle}{\right\rangle}

\newcommand{\smallsetexpansion}{{\sc SmallSetEdgeExpansion}}
\newcommand{\smallsetvertexexpansion}{{\sc SSVE}}
\newcommand{\stronguniquegames}{{\sc StrongUniqueGames}}
\newcommand{\uniquegames}{{\sc UniqueGames}}

\newcommand{\hypersse}{{\sc Hyper-SSE}}

\newcommand{\non}{\nonumber}

\newcommand{\cE}{{\mathcal{E}}}

\newcommand{\wh}[1]{\widehat{#1}}
\renewcommand{\vphi}{v_{\emptyset}}
\newcommand{\uphi}{u_{\emptyset}}

\newcommand{\vol}[1]{{\sf Vol}(#1)}
\newcommand{\ssee}{\textsc{SSE}}
\newcommand{\ssve}{\textsc{SSVE}}

\newcommand{\hypergraphsse}{{Hyper-SSE}}
\newcommand{\phiv}{\phi^{\sf V}}
\newcommand{\phie}{\phi^{\sf E}}

\newcommand{\hsse}{{\sc Hypergraph small-set expansion}}

\newcommand{\tPr}{\widetilde{\Pr}}
\newcommand{\tEx}{\widetilde{\Ex}}

\allowdisplaybreaks

\title{New Approximation Bounds for Small-Set Vertex Expansion}
\author{ Suprovat Ghoshal \\ Northwestern University \& TTIC\footnote{This work was done while SG was at University of Michigan.} \\ suprovat.ghoshal@northwestern.edu
	\and Anand Louis \\ Indian Institute of Science \\ anandl@iisc.ac.in }

\begin{document}

\begin{titlepage}
	\maketitle
	\begin{abstract}
		The vertex expansion of the graph is a fundamental graph parameter. Given a graph $G=(V,E)$ and a parameter $\delta \in (0,1/2]$, its $\delta$-Small-Set Vertex Expansion (\ssve) is defined as 
		\[ \min_{S : \abs{S} = \delta \abs{V}}  \frac{\abs{\partial^V(S)}}{ \min \{ \abs{S}, \abs{S^c} \}  } \]
		where $\partial^V(S)$ is the vertex boundary of a set $S$. The \ssve~problem, in addition to being of independent interest as a natural graph partitioning problem, is also of interest due to its connections to the \stronguniquegames~problem \cite{GL20}. We give a randomized algorithm running in time $n^{{\sf poly}(1/\delta)}$, which outputs a set $S$ of size $\Theta(\delta n)$, having vertex expansion at most
		\[
		 \max\left(O(\sqrt{\phi^* \log d \log (1/\delta)}) , \tilde{O}(d\log^2(1/\delta)) \cdot \phi^* \right),
		\] 
		where $d$ is the largest vertex degree of the graph, and $\phi^*$ is the optimal $\delta$-\ssve. The previous best known guarantees for this were the bi-criteria bounds of $\tilde{O}(1/\delta)\sqrt{\phi^* \log d}$ and $\tilde{O}(1/\delta)\phi^* \sqrt{\log n}$ due to Louis-Makarychev [TOC'16]. \\
		
		Our algorithm uses the basic SDP relaxation of the problem augmented with ${\rm poly}(1/\delta)$ rounds of the Lasserre/SoS hierarchy. Our rounding algorithm is a combination of rounding algorithms of \cite{RT12,ABG12}. A key component of our analysis is novel Gaussian rounding lemma for hyperedges which might be of independent interest.
		
	\end{abstract}
\end{titlepage}


\newpage

\section{Introduction}

Graph partitioning problems are a fundamental class of problems that are studied extensively in theory and practice. In theory, they have many connections to metric embeddings \cite{KV05,ALN08}, Markov chains~\cite{LK99,AndersenPeres09}, etc. in addition to connections to fundamental open questions such as the {\em Unique Games Conjecture}~\cite{Khot02a} and {\em Small-set Expansion Hypothesis}~\cite{RS10}. In practice, they are extensively used as inexpensive pre-processing steps for simplifying an optimization problem into isolated sub-problems of smaller size e.g., dynamic algorithms~\cite{SarWang19}, clustering algorithms~\cite{KVV04}, ranking~\cite{Fogel16}, etc. A commonly studied problem in this context is to design algorithms where the goal is to output a partition which minimizes the {\em edge expansion} of the set i.e,. given a $d$-regular graph $G = (V,E)$ on $n$-vertices, the goal is to minimize
\[
\phie_G(S) := \frac{|E(S,S^c)|}{d|S|}
\] 
over all sets of size at most $n/2$, where $\phie_G(S)$ is referred to as the edge expansion or ``conductance'' of the set $S$. This, and its several variants, have been studied extensively in the literature, on account of being natural optimization problems with deep connections to several areas in mathematics such as {\em Isoperimetric Inequalities}~\cite{Alon86} and {\em Metric Embedding}~\cite{KV05}. Of particular interest in this setting is the problem of minimizing the $\delta$-small set edge expansion of a graph -- denoted by $\phie_\delta$ -- where given a parameter $\delta \in (0,1/2]$, the goal is to compute a set with the minimum edge expansion over all sets of size $\delta n$. Introduced by Raghavendra and Steurer \cite{RS10}, the {\sc SmallSetEdgeExpansion}~problem (\ssee~in short) is central to the {\em Small-set Expansion Hypothesis} (SSEH) and is closely related to Khot's Unique Games Conjecture~\cite{Khot02a}.

The focus of this work is the related problem of minimizing {\em vertex expansion} $\phiv$ in graphs. Formally, given a graph $G = (V,E)$, the vertex expansion of a set $S \subseteq V$ is defined as
\[ \phiv_G(S) \defeq \frac{\abs{\partial^V_G(S)}}{ \abs{S}   } \]
where $\partial^V_G(S)$ is the vertex boundary\footnote{Formally, the vertex boundary of a set $S$ is defined as $\partial^V_G(S) := \{u \in S^c | \exists v \in S,  (u,v) \in E\}$.} of the set $S$. The question of finding sets with the minimum vertex expansion has been studied by several works in the past -- the best known upper bounds are by \cite{FLH08} who gave an $O(\sqrt{\log n})$-approximation algorithm, and by \cite{LRV13}, who gave a $O(\sqrt{\phiv \log d})$ bound for approximating vertex expansion on graphs of maximum degree $d$; they also gave a matching (up to constant factors) lower bound based on SSEH. 

Here we study the ``small-set'' variant of the above, namely the {\sc SmallSetVertexExpansion}~(\ssve)~problem. Formally, given a parameter $\delta \in (0,1/2]$, the $\delta$-\ssve~of $G$ is defined as
\[ \phiv_{\delta} \defeq \min_{S : \abs{S} = \delta \abs{V} } \phiv(S) . \]
The $\delta$-small set vertex expansion of a graph, denoted by $\phi^V_\delta$ henceforth, while being an independently interesting quantity on its own as a generalization of edge expansion, is also useful for studying hypergraph analogues of graph partitioning primitives such as Sparsest Cut and Cheeger's Inequality~\cite{CLTZ18,lm16}. Furthermore, minimizing the (small set) vertex expansion appears to be a key combinatorial bottleneck towards approximating CSPs with hard constraints such as the \stronguniquegames~problem\cite{KR08,GL20} -- similar to the well established connection between \uniquegames~and \ssee~\cite{RS09,RST12}.  

However, while $\phiv_\delta$ is a natural variant of edge expansion, as a measure of connectivity, its behavior can be quite different from the edge expansion counterpart $\phie_\delta$. In particular, unlike edge expansion, the vertex expansion of a set does not immediately admit a analogous re-interpretation in terms of random walks. In fact, note that it is entirely possible for a graph to have $\phiv_\delta$ much larger than $1$ -- for e.g., it is easy to show that for all constant $\delta$, $d$-regular random graphs can have $\phi^V_\delta$ as large as $\Omega(d)$. As a result, the techniques and results for approximating $\phie_\delta$ are unlikely to apply as is to the setting of $\phiv_\delta$. Furthermore, it is expected that the approximation curve\footnote{By approximation curve, we mean the mapping $\epsilon \to s_\delta(\epsilon)$, which maps $\epsilon$ to the best possible approximation guarantee one can hope to efficiently achieve over instances with $\phiv_\delta \leq \epsilon$.} of $\phiv_\delta$ also depends on the maximum degree $d$~\cite{LRV13} in addition to parameters $\delta$ and $\phiv_\delta$. To that end, correctly characterizing the three-way trade-off between the parameters $\phiv_{\delta},\delta$ and $d$ in the approximation curve has been challenging, leading us to question:
\begin{equation}				\label{eqn:ques}
\triangleright \textnormal{\it ~~ What is the right form of the approximation curve of \ssve~as a function of $\phiv_{\delta},\delta$ and $d$?}
\end{equation}  
Towards this, a couple of previous works combined with some additional observations paint a partial yet somewhat helpful picture. Firstly, there have been few works which make progress in terms of upper bounds: \cite{lm16} gives an algorithm which outputs $O(\delta n)$-sized sets with vertex expansion at most ${\tilde{O}((1/\delta))\cdot\sqrt{\phiv_{\delta} \log d}}$. Moreover, by combining a straightforward reduction from edge expansion to vertex expansion with the bounds from \cite{RST10a}, one can derive an algorithm which outputs a $O(\delta n)$-sized set with vertex expansion at most $O(d)\cdot\sqrt{\phi^V_\delta \log(1/\delta)}$. On the other hand, one can guess that for certain ranges of parameters $\phiv_\delta,d$ and $\delta$, the right approximation might be the (analytic) vertex expansion of $(1 - \epsilon)$-noisy Gaussian graphs\footnote{The $(1 -\epsilon)$-noisy Gaussian graph is the infinite graph on $\mathbbm{R}^d$ whose random walk is governed by the Ornstein-Uhlenbeck operator~\cite{OD14} with correlation parameter $(1 - \epsilon)$.} -- this is known (for e.g.,~\cite{GL20}) to be $\alpha_{\rm SSE}:= \sqrt{\epsilon \log(d) \log(1/\delta)}$. Finally, \cite{lm16} also showed that the basic SDP relaxation admits an integrality gap of $\alpha_{\rm Int}:=\Omega(\min\{d,1/\delta\})$.

\subsection{Our Results}	

This paper makes progress towards giving a more unified answer to \eqref{eqn:ques} by giving new approximation bounds, as well as strengthening the existing unconditional hardness, for the $\delta$-\smallsetvertexexpansion~problem. Our first main result is the following new approximation bound for the \smallsetvertexexpansion~problem:

\begin{restatable}{rethm}{hssemain}
	\label{thm:hsse-main}
	The following holds for any constant $\delta \in (0,1/2)$. There exists a randomized algorithm which on input a graph $G = (V,E)$ of maximum degree $d \in \mathbbm{N}$, outputs a set of size $[0.99 \delta n, 1.01 \delta n]$ having vertex expansion at most 
	\[
		O\paren{ \sqrt{{\phi^V_\delta} \log d \log(1/\delta)}} + {\phi^V_\delta} \cdot O(d)(\log(d)\log^2(1/\delta)).
	\]
	in time $n^{{\rm poly}(1/\delta)}$. 
\end{restatable}

The approximation bound from the above theorem can be interpreted to be as $\max\{\alpha_{\rm SSE},\tilde{O}(\alpha_{\rm Int} \cdot \phiv_\delta)\}$ where $\alpha_{\rm SSE} := \sqrt{\epsilon \log(d) \log(1/\delta)}$ is the (conjectured) SSEH based hardness (see Remark \ref{rem:sse-hardness}), and $\alpha_{\rm Int}$ is the integrality gap of Sum-of-Squares (SoS) lifting of the SDP relaxation (Lemma \ref{lem:int-gap}). We point out that there is a subtle difference between the setting of the above theorem -- where the non-expanding set is promised to be of relative size $\delta$, and the algorithm also outputs a set of size $\approx\delta$  -- and those of previous related works~\cite{RST10a,lm16} where the promise and the approximation guarantee are for sets of sizes at most $O(\delta n)$. Nevertheless, we still compare the guarantees of Theorem \ref{thm:hsse-main} with known bounds.

\begin{itemize}
	\item Louis and Makarychev~\cite{lm16} give algorithms which output sets of size at most $(1 + \eta)\delta$ with vertex expansion at most $O_\eta(1/\delta)\cdot\sqrt{\phiv_\delta \log d}$ for any $\eta$, with the $O_\eta(\cdot)$ hiding ${\rm poly}(1/\eta)$-terms. Then for $\delta^{0.99} \leq O(1/\sqrt{\phiv_\delta}d)$, Theorem \ref{thm:hsse-main} gives improved approximation guarantees.
	\item Reducing vertex expansion to edge expansion, and plugging in the guarantees of \cite{RST10a} for SSE yields an algorithm which outputs sets of size at most $\delta n$ with vertex expansion at most $O(d\sqrt{\phiv_{\delta} \log(1/\delta)})$. Again, when $\phi^V_\delta \leq 1/\log^4(1/\delta)$, Theorem \ref{thm:hsse-main} gives improved approximation guarantees in comparison. 
\end{itemize}

Combining the above immediately yields the following guarantee:

\begin{corollary}
	Given a graph $G$ with maximum degree $d$ such that $\phiv_\delta \leq 1/(d^3 \log^4(1/\delta))$, there exists a $n^{{\rm poly}(1/\delta)}$-time algorithm that returns a set of size $\Theta(\delta n)$ with vertex expansion at most $O(\sqrt{\phiv_\delta \log d \log(1/\delta)})$
\end{corollary}	
	
{\bf Lower bounds}. We also give conditional and unconditional hardness results which hold for more general settings of $\phiv_{\delta}$. The first result shows that it is SSEH hard to get a $f(d)$-approximation for SSVE for any increasing nonnegative function $f$.

\begin{restatable}{rethm}{ssedeg}
	\label{prop:deg-hardness}
	Assuming the SSEH, the following holds for any computable nonnegative increasing function $f:\mathbbm{N} \to \mathbbm{R}_{+}$. There exists an absolute constant $\epsilon_0 \in (0,1)$ such that for any $\epsilon \in (0,\epsilon_0)$, there exists $\delta = \delta(\epsilon)$ and a degree $d$ such that given a graph $G = (V,E)$ of maximum degree $d$ it is \NP-Hard to distinguish between
	\begin{itemize}
		\item {\bf YES}: $\phi^V_\delta(G) \leq \epsilon$
		\item {\bf NO}: $\phi^V_\delta(G) \geq f(d)\cdot \epsilon$
	\end{itemize}
\end{restatable}

\begin{remark}			\label{rem:sse-hardness}
	It is widely conjectured that for graphs with degree $d$, it should be hard to distinguish between whether $\phi^V_\delta \leq \epsilon$ and $\phi^V_\delta \geq \sqrt{\epsilon \log d \log(1/\delta)}$. This is consistent with known results for $\delta$-small set edge expansion~\cite{RST12}, for which it is known that it is hard to distinguish between $\phi^E_\delta \leq \epsilon $ and $\phi^E_\delta \geq \sqrt{\epsilon \log(1/\delta)}$. This is also consistent with known hardness for vertex expansion~\cite{LRV13} which states that it is hard to distinguish between the cases $\phi^V \leq \epsilon$ and $\phi^V \geq \sqrt{\epsilon\log d}$.
\end{remark}
	
Additionally, we also strengthen the integrality gap construction of \cite{lm16} to show a $\Omega(d)$-integrality gap for the $\tilde{O}(n)$-round SoS lifting of the basic SDP (Proposition \ref{lem:int-gap}).

{\bf Additional Algorithmic Applications}. Our algorithm for Theorem \ref{thm:hsse-hyper} is based on solving a SoS lifting of the strengthened SDP relaxation, followed by a conditioning + Gaussian rounding step. We use this framework to derive new approximation guarantees for the related problem of \hypergraphsse. Given a hypergraph $H=(V,E,w)$ with hyperedge weights $w: E \to \mathbbm{R}_+$, the expansion of a set $S \subset V$ is defined as
\begin{equation}				\label{eqn:hyper-exp} 
\phie_H(S) \defeq \frac{ w\paren{\partial^E_H(S)}}{\min \left\{\vol{S},\vol{S^c}\right\}},  
\end{equation}
where for a subset $S$, $\phie_H(S)$ denotes the hyperedge boundary\footnote{For a subset $S \subseteq V$, the hyperedge boundary $\partial^H_E(S)$ is the set of hyperedges crossing the cut $(S,S^c)$ i.e., $\partial^E_H(S) = \{e \in E | e \cap S, e \cap S^c \neq \emptyset\}$.} and $\vol{S}$ denotes its volume, defined as the sum of the degrees of the vertices in $S$. The following theorem states our guarantees for \hypersse.

\begin{restatable}{rethm}{hssehyper}
	\label{thm:hsse-hyper}
	There is a randomized algorithm which takes as input 
	a $r$-uniform hypergraph $H=(V,E,w)$ with hyperedge weights $w: E \to \mathbbm{R}^+$, max-degree $d_{\rm max}$,
	and containing a set of relative weight $\delta $ with expansion $\phie_\delta$,
	runs in time $n^{{\rm poly}(1/\delta)}$, and outputs a set of relative weight $\delta(1 \pm o_\delta(1))$
	having hyperedge expansion at most 
	\[
	O\paren{\sqrt{(d_{\rm max}/r) {\phie_\delta} \log(1/\delta)\log r }  + \tilde{O}(r){\phi^*}(\log(1/\delta))^{2}} 
	\]
\end{restatable}

Again, we compare this with the $\tilde{O}(1/\delta)\sqrt{(d_{\rm max}/r) \phie_\delta \log r}$-bound from \cite{lm16}; Theorem \ref{thm:hsse-hyper} again gives better guarantees in the small-volume + small-expansion setting i.e., $\delta^{0.99} \leq  d/\sqrt{\phie_\delta}$. 

\subsection{Related Works}

\paragraph{Graph/Hypergraph Partitioning Problems.} There is a long line of works which study graph/hypergraph partitioning while aiming to minimize various notions of expansion. The simplest variant where the objective is to minimize expansion i.e., {\sc Sparsest Cut}~has been studied extensively. Leighton and Rao \cite{LR99} first gave a $O(\log n)$-approximation for it. The breakthrough work of Arora, Rao and Vazirani \cite{ARV09} gave an $O(\sqrt{\log n})$-approximation for this problem. Feige, Hajhiyaghayi and Lee~\cite{FLH08} study the balanced vertex separator problem for which they gave a $O(\sqrt{\log n})$-approximation. Louis, Raghavendra and Vempala~\cite{LRV13} gave a $O(\sqrt{\phiv \log d})$ bound for minimizing vertex expansion over graphs with max-degree at most $d$; they also proved a matching (up to constant factors) lower bound based on SSEH. Finally, there have been several recent works which explore undirected and directed vertex and hypergraph expansion via reweighed eigenvalues~\cite{KLT22,LTW23}.

\paragraph{Small Set Expansion.}
The \ssee~problem was introduced by Raghavendra and Steurer~\cite{RS10} in the context of the Small Set Expansion Hypothesis, in an attempt to characterize the combinatorial structure underlying hard instances of CSPs such as \textsc{Unique Games}. Raghavendra, Stuerer and Tetali~\cite{RST10a} gave a $O(\sqrt{\phie_\delta \log (1/\delta)})$ bound for \ssee -- this was later matched (up to constant factors) by the work of Raghavendra, Steurer and Tulsiani \cite{RST12}. In terms of multiplicative approximation, Bansal et al. \cite{BFKM11} gave a bi-criteria $O(\sqrt{\log n \log (1/\delta)})$ approximation for \ssee. There have also been a sequence of works which show that \ssee~is tractable on graphs having ``low threshold rank'', for e.g, see \cite{Kol10}, \cite{GS11}. On other other hand, Arora, Barak and Steurer~\cite{ABS15} gave a subexponential time algorithm for \ssee~on general instances. \ssee~is also studied in the context of higher order variants of Cheeger's inequality~\cite{LOT14,LRTV12} which relate the multi-way expansion with the higher eigenvalues of the normalized Laplacian matrix of the graph. Louis and Makarychev ~\cite{lm16} gave $\tilde{O}(1/\delta)\sqrt{\log n}$-approximation algorithms for the $\delta$-\ssve~and $\delta$-\hypersse~problems respectively. Additionally, for graphs of maximum degree $d$, they also gave a $\tilde{O}(1/\delta)\sqrt{\phiv_\delta \log d}$-bound for \ssve.

\paragraph{CSPs with Cardinality Constraints.}
\ssee~and \hypersse~can also be interpreted as Constraint Satisfaction Problems (CSPs) with global cardinality constraints. CSPs with cardinality constraints model a wide array of combinatorial optimization problems such as {\sc Max-Bisection}, {\sc Min-Bisection}, {\sc Densest}-$k$-{\sc Subgraph}, {\sc Min-Max-Partition}-etc. The global constraint influences the tractability of underlying CSP. For e.g., Austrin and Stankovic~\cite{AS19} showed that adding cardinality constraints to {\sc Max-Cut} makes it strictly harder. Therefore, there have been several attempts to study this class of CSPs on its own. Raghavendra and Tan~\cite{RT12} proposed a general framework for approximating CSPs with cardinality constraints, and used it to give a $0.85$-approximation for {\sc Max Bisection}, which was later improved to $0.8776$ by Austrin, Bennabas and Georgiou~\cite{ABG12}. There have been several works \cite{GS11,GL14,BansalSticky20,GL22,GL23} which study such CSPs in more general contexts.

\section{Overview}		\label{sec:overview}

Our algorithmic result for \ssve~is achieved via a known reduction to \hypersse~ (Lemma \ref{lem:transform}), henceforth we shall focus our discussion on \hypersse. Recall that the setting of \hypersse~is as follows: given a hypergraph $H = (V,E)$, the objective is to find a set $S \subseteq V$ of relative volume $\delta$ which minimizes the hyperedge expansion $\phie_H(S)$ (as defined in \eqref{eqn:hyper-exp}). As is usual, our approach to designing an algorithm for \hypersse~relies on the ``SDP Relaxation + Gaussian Rounding'' framework. However, towards designing the algorithm, we will encounter and address several fundamental issues which will guide our final approach. The rest of the section will consist of the following components.
\begin{itemize}
	\item We first discuss the basic SDP relaxation and its integrality gap. 
	\item We shall then sketch a relatively straightforward $\sqrt{\phie_\delta \log(1/\delta)}$-algorithm for the ``equals'' version of small-set edge expansion (i.e., arity $2$), and discuss the key bottlenecks towards extending this approach to the setting of hypergraphs of higher arities.
	\item Finally, we discuss our approach towards addressing these challenges and conclude by presenting a simplified version of our algorithm and its analysis.
\end{itemize}

\subsection{The Basic SDP}					\label{sec:basic-sdp}

We begin by describing the basic SDP relaxation for the \hypersse~in SDP \ref{sdp:basic}.

\begin{SDP}\label{sdp:basic}
	\begin{eqnarray*}				
	\text{minimize}  &  \sum_{e \in E} \max_{i,j \in e} \| u_i - u_j\|^2 & \\
	\text{subject\space to}& \sum_{i \in V} \langle u_i, \uphi \rangle = \delta n \\
	& \langle u_i, \uphi\rangle = \|u_i\|^2 & \forall i \in V\\
	& \|\uphi\|^2 = 1.
	\end{eqnarray*}
\end{SDP}

In the above SDP, the objective is a convex vector relaxation of the hypergraph cut function. The variables i.e., the vectors $\{u_i\}_{i \in V}$ are intended to be $\{0,1\}$ indicators of a set in an integral solution, and $\uphi$ is the purported ``one'' vector which aligns the SDP solution. The first constraint in the SDP relaxation is intended to control the cardinality of the set represented by the solution, and the second constraint is the vectorization of the Booleanity constraint $x^2 = x$ for $x \in \{0,1\}$. This relaxation and its instantiations to the setting of graphs (often strengthened with $\ell^2_2$-triangle inequalities), have been used in most of the earlier work on small-set expansion problems \cite{RST10a,BFKM11,lm16}. While for constant arity (i.e. size of the largest hyperedge), the above SDP can yield optimal (up to constant factors and assuming SSEH) approximation results (e.g,. see \cite{RST10a} and \cite{RST12}), in general the integrality gap is at least linear in the arity of the hypergraph, as stated in the following observation.

\begin{observation}\cite{lm16} 			\label{obs:ig}
	For hypergraphs with arity $d$, the integrality gap $\alpha_{\rm Int}(\delta,d,\epsilon)$ of the SDP \ref{sdp:basic} is at least $\Omega(\min\{1/\delta, d\})$.
\end{observation}

The integrality gap instance is simple to describe; it is a single hyperedge on $d$ vertices. Here for $\delta = 1/d$ the optimal expansion is $1$, since the hyperedge will always get cut. On other other hand, consider the following vector assignment. Let $\uphi, \bar{z}_1,\ldots,\bar{z}_d$ be $(d + 1)$-orthonormal vectors. For every $i \in [d]$, assign $u_i = \delta  \uphi + \sqrt{\delta - \delta^2}\bar{z}_i$. It is easy to verify that (a) this assignment is feasible and (b) this yields a value of $1/d$ for the objective. Observation \ref{obs:ig} results in a situation where the integrality gap (denoted by $\alpha_{\rm Int}$)\footnote{Here $\alpha_{\rm Int} = \alpha_{\rm Int}(d,\delta,\epsilon)$ denotes the integrality gap of SDP \ref{sdp:basic} on $\delta$-\hypersse~instances with maximum degree $d$ and SDP value $\epsilon$} and SSEH based hardness factors (denoted by $\alpha_{\rm SSE}$) are {\em incomparable}. This leads us to ask the following which is another motivating question for this work: {\em Can we give an algorithm which achieves an approximation factor of $O(\max(\alpha_{\rm SSE},\alpha_{\rm Int} \cdot \phiv_\delta))$?}

\begin{remark}
	We point out that the integrality curve $\alpha_{\rm Int}(\cdot)$ corresponds to that of the standard SDP relaxation described in SDP \ref{sdp:basic}. In particular, $\alpha_{\rm Int}(\cdot)$ does not necessarily have to match the approximation curve $\alpha_{\rm SSE}(\cdot)$. On the other hand, it is clear that any rounding scheme for the above relaxation cannot have an approximation guarantee better than $\alpha_{\rm Int}$. Hence a restatement of the above question is that can we give find sets with vertex expansion at most $\alpha_{\rm SSE} = \alpha_{\rm SSE}(\phiv_\delta,d,\delta)$ when $\alpha_{\rm SSE} < \alpha_{\rm Int}\cdot \phiv_\delta$? 
\end{remark}

\subsection{A $\sqrt{\phie_\delta \log(1/\delta)}$-algorithm for SSE}


Here we sketch an algorithm for \ssee~which outputs $~\delta|V|$-sized sets with edge expansion at most $\sqrt{\phie_\delta \log(1/\delta)}$. Towards this, we begin by considering the instantiation of SDP \ref{sdp:basic} for the setting of graphs i.e,. $d = 2$:
 
\begin{SDP}\label{sdp:sse}
	\begin{eqnarray*}				
		\text{minimize}  &  \sum_{(i,j) \in E} \| u_i - u_j\|^2 & \\
		\text{subject\space to}& \sum_{i \in [n]} \langle u_i, \uphi \rangle = \delta n \\
		& \langle u_i,\uphi \rangle = \|u_i\|^2 	& \forall i \in V\\
		& \|\uphi\|^2 = 1.
	\end{eqnarray*}
\end{SDP}

{\bf Idealized Analysis and a Gaussian Rounding Question}. Given the SDP relaxation in SDP \ref{sdp:sse}, the natural next step is to design and analyze a rounding algorithm which will use the SDP solution to round off a feasible integral solution -- i.e., a $\delta n$-sized subset -- with the objective that hyperedge boundary of the rounded solution is not too large in comparison to the SDP objective. This is the central step of the process and is key to determining the overall approximation guarantee of the algorithm. For the setting of graphs i.e, when the arity $d = 2$, designing optimal rounding schemes is a well understood process and is based on the following geometric principle: vertices whose corresponding vectors are far apart should be far more likely to be cut by the rounded off solution and vice versa. This is easily achieved by projecting the vectors along a randomly sampled Gaussian vector. Furthermore, in order to control the size of the set output by the algorithm, the partition of the vertices is guided by an appropriately chosen threshold. These with some additional technical considerations taken together yield the rounding scheme described in Figure \ref{fig:round-basic}.

\begin{figure}[ht!]
	\begin{mdframed}
		\begin{center}
			\underline{\bf Rounding for SSE}
		\end{center}
	\vspace{5pt}
			\begin{itemize}
				\item For every $i \in V$, write $u_i = \mu_i \uphi + z_i$, where $z_i \perp \uphi$ is the component of $u_i$ orthogonal to the one vector $\uphi$.
				\item {\bf Random Projection}. Sample $g \sim N(0,1)^k$ (where $k$ is the dimension of the SDP solution) and for every $i \in V$, project $g_i := \Langle g,\frac{z_i}{\|z_i\|}\Rangle$.
				\item {\bf Thresholding}. Construct $S \subseteq V$ by including every $i \in V$ for which if $g_i \leq \Phi^{-1}(\delta)$.\footnote{Here $\Phi:\mathbbm{R} \to [0,1]$ is the Gaussian CDF function.} 
				\item Output set $S$.	
			\end{itemize}
		
	\end{mdframed}
\caption{Basic Rounding Scheme}
\label{fig:round-basic}
\end{figure}

The correctness of the rounding scheme can be shown by establishing the following two steps: (i) the set returned is of size $\approx \delta |V|$ with high probability and (ii) the expected expansion of the rounded off solution is small. For establishing (i), one can easily see that the choice of the threshold ensures that for $i \in V$, marginally $g_i = \langle g,z_i/\|z_i\|\rangle$ is a standard Gaussian random variable, and hence 
\begin{equation}					\label{eqn:del-bias}
\Pr_{(x_i)_{i \in V}}\big[x_i = 1\big] = \Pr_{g \sim N(0,1)}\Big[g_i \leq \Phi^{-1}(\delta)\Big] = \delta,
\end{equation}
which implies that the expected size of the set rounded by the algorithm is $\delta|V|$. The more interesting component here is to analyze the expansion of the set $S$ output by the algorithm (Figure \ref{fig:round-basic}). Equivalently, we want to bound the expected fraction of edges cut by $S$. It turns out that for a fixed edge $(i,j) \in E$, analyzing the probability of $(i,j)$ getting cut reduces to the following Gaussian stability type question which asks "{\em what is the probability that two $\rho$-correlated Gaussians get separated by a halfspace of volume $\delta$}"? This question and its variants have been studied in various contexts in the combinatorial optimization and graph partitioning problems in particular. Goemans and Williamson~\cite{GW94} first studied the above for the setting of $\delta = 0$; subsequent works such as \cite{CMM06a} on \uniquegames~and \ssee~study it for more general $\delta$'s. In particular, incorporating the bounds from \cite{CMM06a} into the setting of our rounding scheme yields the following.
\begin{fact}				\label{fact:edge-sep}
	Fix $i,j \in V$ and let $H_\delta = \{x \in \mathbbm{R} | x \leq \Phi^{-1}(\delta)\}$ be the halfspace of Gaussian volume $\delta$. Then,
	\[
	\Pr\Big[\mathbbm{1}_{H_\delta}(g_i) \neq \mathbbm{1}_{H_\delta}(g_j)\Big]
	\lesssim \|u_i-u_j\|\cdot\sqrt{\delta\log(1/\delta)}.
	\]
	where $\lesssim$ hides constant multiplicative factors. 
\end{fact}
We point out that earlier works which employ Gaussian thresholding rounding study it for the setting where the random projection is done along the $\{u_i\}_{i \in V}$ vectors as opposed to the $\{z_i\}_{i \in V}$ vectors. Consequently, even for edges, one still needs derive the above bound specifically for our rounding scheme, although it follows in a straightforward manner using the techniques of \cite{CMM06a}. Equipped with the above bound, (assuming $G$ is $\ell$-regular for simplicity) one can conclude the analysis of the approximation guarantee by bounding the expansion of the rounded set as:
\begin{equation}				\label{eqn:final}
\frac{\sum_{(i,j)} \|u_i - u_j \| \sqrt{\delta \log(1/\delta)}}{\delta \ell n}
\leq \sqrt{\frac{|E|}{\ell n}} \cdot \sqrt{\frac{\sum_{(i,j)} \|u_i - u_j\|^2 \log(1/\delta)}{\delta \ell n}} 
\leq \sqrt{\phie_\delta \log(1/\delta)}. 
\end{equation}
From the above discussion, it is easy to infer that a key bottleneck towards establishing analogous approximation guarantees for \hypergraphsse~would be to extend Fact \ref{fact:edge-sep} to $d$-sets of correlated Gaussians (where $d > 2$). In fact, looking at the form of $\alpha_{\rm SSE}$, one might even be tempted to suggest that the following holds\footnote{The intuition behind suggesting this bound is the following well-known Gaussian geometric fact (e.g., Proposition 8.6~\cite{GL20arxiv}): for $\epsilon$-small enough as a function of $\delta,d$, if $g_1,\ldots,g_d$ were $(1 - \epsilon)$-correlated copies of a standard Gaussian $g \sim N(0,1)$, then the RHS of \eqref{eqn:cut-bound} can be shown to be $\Theta(\sqrt{\epsilon \log d \log(1/\delta)})$.}:
\begin{question}					\label{ques:cut-bound}
	Given unit vectors $u_1,\ldots,u_d$, let $g_i := \Langle g, \frac{z_i}{\|z_i\|} \Rangle$ be the corresponding projected Gaussians (as in Figure \ref{fig:round-basic}). Furthermore, let $H_\delta \subset \mathbbm{R}$ be the halfspace of Gaussian volume $\delta$. Then
	\begin{equation} \label{eqn:cut-bound}
	\Pr_{(g_i)_{i \in [d]}}\Big[\exists i,j \mbox{ s.t. } \mathbbm{1}_{H_\delta}(g_i) \neq \mathbbm{1}_{H_\delta}(g_j) \Big] \overset{\mbox{\large{ \bf ?}}}{\leq} \left\{ \max_{i,j \in e} \|{u}_i - {u}_j\|\right\} \cdot \sqrt{\delta \log(1/\delta) \log(d)},
	\end{equation}
\end{question}
If the above bound i.e, \eqref{eqn:cut-bound}, held for arbitrary sets of vectors $u_1,\ldots,u_d$, then plugging in such a bound into the analysis described above would immediately yield an algorithm which outputs $~\delta$-weight sets with hyperedge expansion at most $\sqrt{\phi^{\sf H}_\delta \log (d) \log(1/\delta)}$ and we would be done. Unfortunately, it turns out that the above does not actually hold for arbitrary choices of $u_1,\ldots,u_d$, and the witness to that can again be constructed from the integrality gap example, as we discuss below. 

Let $(u_i)_{i \in [d]}$ be the vector assignment for the integrality gap assignment i.e,  $u_i = \delta\uphi + \sqrt{\delta - \delta^2} \cdot \bar{z}_i$ for every $i \in V$ where recall that $\{\bar{z}_i\}$ are orthonormal vectors which are all orthogonal to $\uphi$. Then observe that for $g_i = \langle g,\bar{z}_i \rangle$, the Gaussians $g_1,\ldots,g_d$ are {\em independent} standard Gaussian random variables and hence they separate the edge with probability at least $1 - (1 - \delta)^d - \delta^d \geq \Omega(1)$ when $d = \Theta(1/\delta)$. On the other hand, the RHS of \eqref{eqn:cut-bound} is at most  $\sqrt{\delta \cdot d^{-1} \cdot\log(1/\delta) \log(d)} \ll O(1)$ when $\delta$ is small, and hence the above assignment of vectors $\{v_i\}_{i \in [d]}$ violates \eqref{eqn:cut-bound}. Therefore, understanding under what conditions one can establish \eqref{eqn:cut-bound} and how can one analyze hyperedges which for which the local distribution violates \eqref{eqn:cut-bound} is one of the key bottlenecks towards designing and analyzing an approximation algorithm for \hypersse. 

\subsection{Our Approach: Nice edges and Gap edges.}

As described in the previous part, identifying robust characterizations of $(g_i)_{i \in [d]}$ from Question \ref{ques:cut-bound} under which one can expect to establish \eqref{eqn:cut-bound}, and ways to handle edges for which the local distribution does not allow for \eqref{eqn:cut-bound}, is one of the key challenges towards improving on existing bounds. Interestingly, it turns out that the fact that our example which violates \eqref{eqn:cut-bound} can be derived from the vector assignment of the integrality gap instance is not entirely coincidental. In fact, we can systematically show that collections of Gaussians (identified by the hyperedges in the hypergraph) that violate \eqref{eqn:cut-bound} have correlation structure which geometrically resemble the integrality gap vector assignment, in a somewhat approximate sense. This connection allows us to classify hyperedges as ``nice'' edges -- for which \eqref{eqn:cut-bound} holds, and consequently we can argue Gaussian stability type approximation bounds, and ``gap'' edges -- which resemble the integrality gap assignment, for which way pay a cost of $\tilde{O}(d \cdot{ \sf Cost}(e))$ -- where ${\sf Cost}(e)$ is the cost contributed by $e$ to the objective. 

Towards establishing the above, it is useful to adopt the probabilistic view point of pseudo-distributions\footnote{Informally, a degree-$r$ pseudo-distribution is a collection of distributions on subsets of variables of sizes at most $r$ that are locally consistent, see Section \ref{sec:lass} for a more formal description.} corresponding to the vector solution. Given a hyperedge $e = \{1,\ldots,d\}$, let $(u_i)_{i \in [d]}$ be the corresponding set of vectors, for which we can write $u_i = \mu_i\uphi + z_i$ for every $i \in [d]$, with $z_i$ being the component of $u_i$ orthogonal to the one-vector $\uphi$. This is a useful decomposition since the various quantities derived from it can also be interpreted as probabilistic quantities related to the corresponding degree-$2$ pseudo-distribution. In particular, denoting the local variables associated with the pseudo-distribution as $X_1,\ldots,X_d$, we have that $\mu_i$ represents the (pseudo) probability of setting the variable $X_i$ to $1$, and $\langle z_i,z_j \rangle/\|z_i\|\|z_j\|$ represents the (pseudo) correlation between variables $X_i$ and $X_j$. Consequently, note that the Gaussian ensemble $(g_i)_{i \in [d]}$ matches the correlation structure of the local variables $(X_i)_{i \in [d]}$ i.e, correlation between $g_i$ and $g_j$ is identical to that of $X_i$ and $X_j$, for every $i,j \in [d]$. Furthermore, for any $i,j$, we have that $\|u_i - u_j\|^2$ represents (up to multiplicative factors) the probability of event $\{X_i \neq X_j\}$. 

Now, from standard results in Gaussian noise stability\footnotemark[6] it is understood that one can establish \eqref{eqn:cut-bound} if the corresponding Gaussians are correlated enough. From a geometric perspective, this is expected to happen when the vectors $v_i$ are close enough as a function of the biases $\mu_i$. This intuition can be formalized using using the following probabilistic (and equivalently, geometric) inequality: for any pair of joinly distributed $\{0,1\}$-random variables $X_i,X_j$, we have 
\begin{equation}				\label{eqn:corr-inf}
\underbrace{\rho(X_i,X_j) \geq 1 - \frac{\tPr\big[X_i \neq X_j\big]}{\sigma(X_i)\sigma(X_j)}}_{{\it Probabilistic}}
\qquad\qquad\qquad\qquad
\underbrace{\langle \bar{z}_i,\bar{z}_j \rangle \geq 1 - \frac{\|u_i - u_j\|^2}{\|z_i\|\|z_j\|}}_{{\it Geometric}}
\end{equation}
where $\rho(\cdot,\cdot)$ denotes the correlation coefficient and $\sigma(\cdot)$ denotes the standard deviation. In particular, assuming $\mu_i = \delta \in (0,1)$ for every $i \in [d]$ for simplicity, one can use the above inequality to derive a relationship between the minimum correlation between the Gaussian random variables and the contribution of the hyperedge to SDP objective.
\begin{equation}				\label{eqn:rho-cut}
	\rho_e: =\min_{i,j \in e} \rho(g_i,g_j) \geq 1 - \frac{\max_{i,j \in [d]} \|u_i - u_j\|^2}{\delta} 
\end{equation}
The above inequality immediately suggests that one can potentially analyze the performance of Gaussian rounding by considering the following extreme (but not exhaustive) cases. For brevity, denote $\nu_e := \max_{i,j \in [d]} \|u_i - u_j\|^2$. Then,

{\bf Case (i)} ``{\sf Nice}$^*$ Edges'': Suppose $\nu_e \ll \delta$, then it follows that $\rho_e \sim 1$ i.e., we are in the almost correlated regime, where we can hope to show that Gaussian rounding yields a cut probability of $O(\sqrt{\nu_e} \cdot\sqrt{\delta \log(d) \log(1/\delta)})$ as suggested by \eqref{eqn:cut-bound}.

{\bf Case (ii)}``{\sf Gap}$^*$ Edges'': Suppose $\nu_e \geq \delta$. Then, we can upper bound the probability of $e$ getting cut by the probability of at least one of the variables being set to $1$, which using \eqref{eqn:del-bias} is at most $d \cdot \delta \leq d \cdot \nu_e$. Here, the contribution matches the multiplicative loss due to the integrality gap of the SDP.

The above classification into {\sf Nice}$^*$ and {\sf Gap}$^*$ edges via \eqref{eqn:rho-cut} forms the basis of our approach towards designing the algorithm for Theorem \ref{thm:hsse-main}. However, note that the above only works in the ``idealized setting'' with the following unreasonable assumptions hold (i) all hyperedges can be classified into the idealized definitions of {\sf Nice}$^*$ and {\sf Gap}$^*$ edges, and (ii) locally the biases $\mu_i$ of hyperedges should all be identical to $\delta$. Towards dealing with the above, we need to consider a modified rounding algorithm and additional pre-processing steps, and consider more relaxed notions of Nice and Gap edges which will allow us to smoothly interpolate between the two settings while covering all possible cases. Consequently, our modified rounding and relaxed notion of nice'ness imply that deriving \eqref{eqn:cut-bound} in the new setting becomes significantly more complicated. 

{\bf The Rounding Scheme, Gaussian Rounding Lemma, Additional Bottlenecks}. A first step towards addressing issues (i) and (ii) outlined above is to use a (standard) alternative rounding scheme~\cite{FG95,RT12} where instead of using a fixed bias for thresholding, we make the thresholds vertex sensitive. Formally, we replace the thresholding step in Figure \ref{fig:round-basic} with the following.
\[
\triangleright ~~ \textnormal{\it For every $i \in V$, include $i \in S$ if $g_i \leq \Phi^{-1}(\mu_i)$,}
\]
i.e., the threshold for vertex $i \in V$ is $\Phi^{-1}(\mu_i)$, where recall that $\Phi(\cdot)$ is the Gaussian CDF function. The varying thresholds allow us to define notions of {\sf Gap} and {\sf Nice} in the following more edge sensitive way. Fix a hyperedge $e = [d]$, and let $(v_i)_{i \in e}$ be the corresponding vectors from the vector solution of the SDP. As before, let us write $u_i = \mu_i\uphi + z_i$ where $z_i \perp \uphi$. Furthermore, for simplicity, we may assume $0 \leq \mu_d \leq \cdots \mu_1 \leq 1/2$ and as before, let $\nu_e$ denote the contribution of the hyperedge to the SDP objective. Then, 
\[
\triangleright ~~ \textnormal{\it In the above setting, we say that the edge $e$ is {\sf Nice} if $\mu_1 \geq \nu_e \cdot \log(d)\log(1/\delta)^2$, and {\sf Gap} otherwise.}
\]

Given this setup, we can now derive the intended bounds for the probability of any hyperedge getting cut under the rounding scheme. As before, for {\sf Gap} edges, we can bound the probability of the hyperedge getting cut by the probability of at least one of the variables being included in the set, which by a union bound and the definition of {\sf Gap} edges is at most $d \cdot \mu_1 \leq \tilde{O}(d \cdot \nu_e)$. On the other hand, if the edge is {\sf Nice}, then we can establish bounds on the cut probability using the following Gaussian rounding lemma which is the key technical contribution of this work.

\begin{lemma}[Informal version of Lemma \ref{lem:cut-bound1}]					\label{lem:cut-informal}
	Suppose the hyperedge $e = [d]$ is {\sf Nice}. Furthermore, suppose we have $\mu_i \geq \delta^{10}$ for every $i \in V$ and $|\mu_i - \mu_j| \leq \nu_e$ for every $i,j \in [d]$. Then, for $g \sim N(0,1)$, the ensemble of Gaussians $\left(g_i := \langle g, \bar{z}_i \rangle\right)_{i \in [d]}$ satisfy:
	\[
	\Pr_{g_1,\ldots,g_d} \Big[\exists \ i,j \in [d], \ g_i \leq \Phi^{-1}(\mu_i)~~ \& ~~ g_j > \Phi^{-1}(\mu_j)\Big] \lesssim \sqrt{\mu_1 \nu_e \log(d) \log(1/\delta)}.
	\]
\end{lemma}

Establishing the above involves most of the technical work -- this is mainly because, unlike the setting of Question \ref{ques:cut-bound}, the above lemma applies in a much broader sense where the biases $\mu_i$ are not fixed to $\delta$ but can vary from vertex to vertex. The proof of the lemma involves a carefully combined application of several ingredients including two sided perturbation bounds of the Gaussian CDF function $\Phi(\cdot)$ (Facts \ref{prop:half-error}, \ref{fact:incr}), large deviation inequalities for maximum of value of $d$-Gaussian random variables (Facts \ref{fact:gauss-max}, \ref{fact:gauss-max1}), and inequalities derived using the probabilistic (pseudo-distribution) interpretation of the vector solution (Lemma \ref{lem:ineq}). Additionally, we point out that the above lemma still doesn't apply unconditionally for {\sf Nice} edges, and in particular, requires them to additionally satisfy the following:

\begin{itemize}
	\item {\it $\delta$-bounded away'ness}: We need the biases $\mu_i$ to be at least $\delta^{O(1)}$ -- this is crucial in recovering the $\log(1/\delta)$-term in the approximation guarantee.
	\item {\it $\ell_1$-bounded'ness}. We need the absolute differences of the biases $\mu_i,\mu_j$ to be bounded in an $\ell_1$-sense by the SDP cost $\nu_e$. 
\end{itemize}

The above conditions are again needed due to additional complications that arise due to the fact that the thresholds in our rounding scheme now are vertex sensitive. The former is addressed by a delicate pre-processing step which ensures that the modified vector solution satisfies the bounded'away ness condition without affecting the overall geometry of the vector solution. The $\ell_1$-boundedness is ensured by introducing $\ell_1$-constraints in the SDP relaxation, and then pre-deleting edges which violate the constraint. We elaborate on these in the next section.


\subsection{The Final Algorithm and its Analysis} 

Our actual algorithm uses the following SoS lifting of the basic SDP described in Figure \ref{fig:round-basic}.

\begin{figure}[ht!]
	\begin{mdframed}
	\begin{center}
		\underline{\bf SDP Relaxation for \hypersse}
	\end{center}
		\begin{eqnarray}
		\min & \frac{1}{\delta n} \sum_{e \in E}w(e) \max_{i,j \in e}\tPr_{(X_i,X_j) \sim \mu_{ij}}\left[X_i \neq X_j\right] \\
		& \E_{i \sim V_G}\tPr_{X_i \sim \mu_i | X_S \gets \alpha} \Big[X_i = 1\Big] =  \delta & \forall S : |S| \leq R-1, \alpha \in \{0,1\}^S \label{eqn:card} \\
		& |\mu_{i} - \mu_{j}| \leq \tPr\left[X_i \neq X_j\right] & \forall i,j \in e, \forall e \in E 			\label{eqn:ell-1}
		\end{eqnarray}
	\end{mdframed}
	\caption{SoS SDP}
	\label{fig:sos-basic}
\end{figure}

For ease of exposition, the objective and constraints in the above SDP are expressed in terms of constraints on the pseudo-distribution. Denoting $(X_i)_{i \in V}$ as the set of local variables, the objective is a sum of terms corresponding to hyperedges $e \in E$, where for every hyperedge $e \in E$, the corresponding term in the objective measures the maximum probability of separation $\tPr_{(X_i,X_j) \sim \mu_{ij}}\left[X_i \neq X_j\right]$ among all pairs $i,j \in e$ under the pseudo-distribution. Note that this is equivalent to the corresponding vectorized term $\max_{i,j \in e}\|u_i - u_j\|^2$ (up to multiplicative constants), and hence the objective is identical to that of SDP \ref{sdp:basic} (up to scaling). In addition, as in SDP \ref{sdp:basic}, we introduce constraints on the biases \eqref{eqn:card} to control the size of the set rounded off by the algorithm. Finally, we also include $\ell_1$-constraints for every hyperedge $e$ which says that the biases $\mu_i,\mu_j$ for vertices $i,j$ inside the hyperedge $e$ cannot deviate by more than the probability of the event $X_i \neq X_j$. Now we describe the rounding scheme in Figure \ref{fig:hsse-round} which happens in four steps.

\begin{figure}[ht!]
	\begin{mdframed}
		\begin{center}{\underline{\bf Rounding for \hypersse}}\end{center}
		\vspace{5pt}
		{\bf Input}. Let $\{\mu'_e\}_{e \in E}$ be the $R$-round optimal pseudo-distribution from solving the SDP relaxation from Figure \ref{fig:sos-basic}.
		\begin{itemize}
				\item[I] {\bf Conditioning}. Choose a random subset $S$ of size $t$ and sample a labeling $X_S \gets \alpha \sim \mu'_S$ from the local distribution on $S$. Let $\mu := \mu'|_{X_S \gets \alpha}$ denote the resulting degree $2$-SoS pseudo-distribution conditioned on $X_S \gets \alpha$. 
				\item[II] {\bf {Edge Deletion}}. Denoting $\nu_e = \max_{i,j \in e} \Pr_{\mu}\left[X_i \neq X_j\right]$, we delete all the hyperedges with $\nu_e \geq 1/10$.
				\item[III] {\bf Pre-processing}. Let ${\bf V}:=\{v_i\}_{i \in V} \cup \{\uphi\} \subset \mathbbm{R}^{\ell}$ be the $\{\pm 1\}$-vector solution\footnote{For technical reasons, it is easier to work with the $\{\pm 1\}$-version of the vector solution instead of $\{0,1\}$. However, the two are related by the identity $v_i : \uphi - 2 u_i$. Consequently, we have the equivalence $u_i = \mu_i \uphi + z_i$ if and only if $v_i = (1 - 2\mu_i) \uphi - 2 z_i$.} corresponding to $\mu$. Using ${\bf V}$, we construct a new vector solution $\{v'_i\}_{i \in V} \cup \{\uphi\}$ as follows. Let $v_i = (1 - \mu_i) \uphi - 2z_i$ for every $i \in V$ as before, and let $\hat{z}$ be a unit vector orthogonal to ${\bf V}$. Then define new vectors $\{v'_i\}_{i \in V}$ as 
				\[
				v'_i = \frac{v_i - \theta \cdot \hat{z}}{\sqrt{1 + \theta^2}},
				\]
				where $\theta := \delta^{100}$.
				\item[IV] {\bf Gaussian Rounding}. Writing $v'_i = (1 - 2\mu'_i) \uphi - 2z'_i$ for every $i \in V$, we round off an integral solution as follows: sample Gaussian vector $g \sim N(0,1)^\ell$ where $\ell$ is the ambient dimension of ${\bf V'}$ and construct $S \subset V$ as 
				\[
				S:= \left\{ i \in V \Big| \Langle g, \frac{z'_i}{\|z'_i\|} \Rangle \leq \Phi^{-1}(\mu'_i)\right\}.
				\]
				Output $S$.
		\end{itemize}
	\end{mdframed}
\caption{HSSE Rounding}
\label{fig:hsse-round}
\end{figure}

We give a brief sketch of the analysis for the above rounding algorithm.

{\bf Conditioning}.  The first step ensures that with high probability, the average correlation between the random variables is small. This ensures that the Gaussian rounding in step IV returns a set which is close to the intended size with large probability. 

{\bf Edge Deletion}. This step along with the $\ell_1$-constraints ensures that for the surviving hyperedges, the ``biases'' (the bias of a vertex $i$ is defined as $\mu_{i}$) of vertices in a hyperedge are ``close'' to each other.  A simple application of Markov's inequality shows the cost paid due to the deleted edges at most $10 \cdot {\sf Sdp\mbox{-}Cost}$.

{\bf Pre-processing}. Here, given the degree-2 SoS solution $\{\mu_{S,\alpha}\}$, let $\{v_i\}_{i \in V} \cup \{\vphi\}$ be the corresponding unique $\{\pm 1\}$-vector solution. As before, we express $v_{i} = (1 - \mu_{i}) \uphi - 2 z_i$ where $\uphi$ is the one vector of the solution, and $\mu_{i} = \Pr_{X_i \sim \mu}\left[X_{i} = -1\right]$. From this, we construct a new vector solution $\{v'_i\}_{i \in V}$ given as 
\[
v'_i \defeq \frac{v_i - \theta\hat{z}}{\sqrt{1 + \theta^2}} = \frac{(1 - 2\mu_i) \uphi - 2z_i - {\theta}\hat{z}}{\sqrt{1 + \theta^2}} 
\]
where $\theta = \delta^{O(1)}$ is a constant depending only on $\delta$ and $\hat{z}$ is a unit vector that is orthogonal to $\uphi,\{z_i\}_{i \in V}$. Although the resulting set of vectors need not constitute a degree-$2$ SoS solution, (for e.g., the Booleanity condition $x^2 = x$ need not be satisfied), we can still pretend that they come from a degree-$2$ pseudo-distribution and write them as $v'_i = (1 - 2\mu'_i)\uphi - 2z'_i$, where $\uphi$ is the original one vector. A delicate analysis shows that they still {\em approximately} satisfy the properties of the original degree-$2$ solution, in addition to satisfying the $\delta$-bounded away'ness condition needed to instantiate Lemma \ref{lem:cut-informal}. We point the readers to Section \ref{sec:pre-proc} for the full suite of properties we re-establish for the new vector solution.  

{\bf Gaussian Rounding}. Finally, we are in the setting where every surviving hyperedge $e \in E$ satisfies the $\delta$-bounded away'ness condition, and the $\ell_1$-bounded'ness condition required for instantiating Lemma \ref{lem:cut-informal}. Now, as before, we shall now break the cut probability analysis into two cases. If $e$ is {\sf Nice}, then Lemma \ref{lem:cut-informal} applies and we can bound the probability of the hyperedge getting cut by $\sqrt{\nu_e \alpha_e \log d \log(1/\delta)}$ (where $\alpha_e$ is the bias farthest away from $\{0,1\}$). On the other hand, if the edge is a {\sf Gap} edge, then we know $\mu_i \leq \tilde{O}(\nu_e)$ for every $i \in e$ and hence by a union bound, the probability of the edge getting cut is bounded by $\tilde{O}(d \cdot \nu_e)$. Combining these cut probability bounds with steps analogous to \eqref{eqn:final} yields the claimed approximation guarantee.

\subsection{Proof Sketch of Gaussian Rounding Lemma}

We conclude our discussion by including a proof sketch of Lemma \ref{lem:cut-informal}. For simplicity, we will assume that $u_i = \delta \uphi + \sqrt{\delta(1 - \delta)} z_i$, and in particular all the rounding thresholds are identical to $t := \Phi^{-1}(\delta)$. Let $g_1,\ldots,g_d$ be the corresponding set of Gaussians. Let $g_1,\ldots,g_d$ be the corresponding set of Gaussians. For $j = 2,\ldots,d$, let $\rho_j$ denote the correlation between $g_1$ and $g_j$. Then for every $j \geq 2$, we can write $g_j = \rho_j g_1 + \sqrt{1 - \rho^2_j}\zeta_j$, where $\zeta_j$ is a $N(0,1)$-Gaussian random variable that is independent of $g_1$. Since $g_j = \langle g, \bar{z}_j \rangle$ by definition, using simple algebra, we can lower bound the correlation $\rho_j$ as a function of the cut-objective as (see \eqref{eqn:corr-inf}):
\[
\rho_j \geq 1 - \frac{\|u_i - u_j\|^2}{4\delta} \geq 1 - \frac{\nu_e}{\delta},
\]
which in turn implies that $\sqrt{1 - \rho^2_j} \leq 4\sqrt{\nu_e/\delta}$ (when $\nu_e \ll \delta$). Now, note that since $\zeta_1,\ldots,\zeta_d$ are $N(0,1)$ random variables, with high probability at least $1 - e^{-d}$ 
\[
\zeta_{\rm max} := \max_{j \geq 2} |\zeta_j|
\]
is bounded by $4\sqrt{\log d}$. Now conditioned on this\footnote{Note that this event is independent of $g_1$.}, for any fixing of $g_1$, we then have that for every $j \geq 2$, 
\[
|g_j - g_1| \leq (1 - \rho_j)|g_1| + \sqrt{1 - \rho^2_j}\cdot\zeta_{\rm max} \lesssim \frac{\nu_e |g_1|}{\delta} + \sqrt{\frac{\nu_e}{\delta}}\cdot\zeta_{\rm max} \overset{1}{\lesssim} \sqrt{\frac{\nu_e}{\delta}} \cdot \zeta_{\rm max},
\] 
where step $1$ holds w.h.p over the choice of $g_1$. Now the main observation in the argument that given that $g_j$'s are close to $g_1$, the only way the set of Gaussians $g_1,\ldots,g_d$ can be cut\footnote{Here, the Gaussians $g_1,\ldots,g_d$ is said to be ``cut'' if $g_i \leq \Phi^{-1}(\delta)$ and $g_j > \Phi^{-1}(\delta)$ for some $i,j \in [d]$} is if $g_1$ is close to threshold $t =\Phi^{-1}(\delta)$.
\begin{observation}[Informal version of Claims \ref{cl:lower-thr} and \ref{cl:upper-thr}]
	For any fixing of $\zeta_{\rm max}$, the edge $(g_1,\ldots,g_d)$ can be cut only if 
	\[
	g_1 \in \left[\frac{t}{1 - \frac{\nu_e}{\delta}} - 4\zeta_{\rm max}\sqrt{\nu_e/\delta}, t + 4\zeta_{\rm max}\sqrt{\nu_e/\delta}\right]
	\] 
\end{observation}
Therefore, conditioned on $\zeta_{\rm max} \leq 4\sqrt{\log d}$, we have that
\begin{align*}
	\Pr_{g_1,\ldots,g_d|\zeta_{\rm max}}\big[e~\textnormal{is cut}\big]
	&\leq \Pr_{g_1 \sim N(0,1)}\left[g_1 \in \left[\frac{t}{1 - \frac{\nu_e}{\delta}} - 4\zeta_{\rm max}\sqrt{\nu_e/\delta}, t + 
	4\zeta_{\rm max}\sqrt{\nu_e/\delta}\right] \right] \\
	&= \Phi\left(t + 4\zeta_{\rm max} \sqrt{\nu_e/\delta} \right) - \Phi\left(\frac{t}{1 - 4\nu_e/\delta} - 4\zeta_{\rm max} \sqrt{\nu_e/\delta}\right) \\
	&\overset{1}{\lesssim}  \left(t + 4\zeta_{\rm max} \sqrt{\nu_e/\delta} - \frac{t}{1 - 4\nu_e/\delta} - 4\zeta_{\rm max} \sqrt{\nu_e/\delta}\right) \delta \sqrt{\log \frac{1}{\delta}} \\
	&\lesssim \zeta_{\rm max}\sqrt{\nu_e/\delta} \cdot \delta\sqrt{\log\frac{1}{\delta}} \\
	&\overset{2}{\leq} \sqrt{\nu_e \delta \log d \log(1/\delta)}, 	 
\end{align*}
where in step $1$ we use the fact that $|\Phi(t \pm \Delta) - \Phi(t)| \lesssim \Delta \delta\sqrt{\log 1/\delta}$ whenever $\Delta \leq 1/|t|$ (Facts \ref{prop:half-error} and \ref{fact:incr})\footnote{This is exactly where we use that since $e$ is a nice edge, $\nu_e/\delta \leq 1/\log(d)\log(1/\delta)$ and hence $|t|\sqrt{\nu_e/\delta} \leq 1$.}, and step $2$ is due to the conditioning on the event $\zeta_{\rm max} \leq 4\sqrt{\log d}$.

\begin{remark}
	Note that in the argument outlined above, several places use bounds that hold with high probability, and hence a naive implementation of the above approach will incur additive error terms due to the small probability events. The actual argument carefully incorporates the small probability events into the analysis in a way so that the corresponding error terms result in multiplicative error instead of additive error.
\end{remark}

\section{Preliminaries}

We introduce some notation and technical preliminaries that will be used in the rest of the paper. Given a hypergraph $H = (V,E,w,W)$, $V$ denotes its vertex set, $E$ its hyperedge set, and $w:E \to \R_{\geq 0}$ and $W:V \to \mathbbm{R}_{\geq 0}$ are the hyperedge and vertex weight functions respectively. For a subset $S \subseteq V$, we shall use $H[S]$ to denote the sub-hypergraph of $H$ induced on $S$. In the case when $H$ is an unweighted graph, we shall use ${\rm deg}_H(i)$ to denote the degree of vertex $i$. We will often be partitioning a hyperedge $e$ as $e = \uplus_{i \in [k]} e_i$, where each $e_i$ is going to be a sub-hyperedge of $e$. Given a subset $S \subseteq V$, we shall say that a hyperedge $e$ is {\em cut} by $S$ if $e \not\subseteq S$ and $e \not\subseteq S^c$. 

We will be using the notation $\lesssim$ to hide leading multiplicative constants, and $\tilde{O}(\cdot)$ to hide logarithmic terms. For a vector $v \in \mathbbm{R}^d$, we shall use $\bar{v}$ to denote the unit vector along the direction of $v$. 

We will use $N(0,1)$ to denote the standard normal distribution i.e., the univariate Gaussian distribution with mean $0$ and variance $1$. We will use $\Phi(\cdot)$ to denote the Gaussian CDF function i.e, $\Phi(t) = \int^{t}_{-\infty} (1/\sqrt{2\pi})e^{-t^2/2} dt$. It is well known that the map $x \to \Phi(x)$ is increasing, continuous and therefore invertible $\forall x \in \R$. We will assume that  the inverse map can be computed in polynomial time up to polynomial precision. Throughout this paper, all the logarithms used will be base $2$.

\subsection{Reduction from \smallsetvertexexpansion~to \hypersse.}

\begin{lemma}[\cite{lm16}, see also Lemma 3.1~of \cite{GL20}]			\label{lem:transform}
	Given a graph $G = (V_G,E_G)$, one can construct an intermediate weighted graph $G_\sym(V_\sym,E_\sym,w_\sym)$ with vertex weights $w_\sym :V_\sym \to \mathbbm{R}_{\geq 0}$, and final weighted hypergraph $H(V,E,w,W)$ with edge weights $w:E \to \mathbbm{R}_{\geq 0}$ and vertex weights $W: V \to \mathbbm{R}_{\geq 0}$ such that the following properties hold:
	\begin{itemize}
		\item[(i)] The vertex set of $V_\sym$ and $V$ are both $V_G \cup E_G$.
		\item[(ii)] The intermediate graph satisfies $\Phi^{\sym}_\delta(G_\sym)  \leq \phiv_{\delta}(G)$ and $\sum_{v \in V_\sym} w_{\sym}(v) = |V_G|$, where $\Phi^\sym_\delta$ is the $\delta$-symmetric vertex expansion\footnote{Given a vertex weighted graph $G = (V,E,w)$, the symmetric vertex boundary of a set $S$, denoted by $\partial^{\rm sym}_G(S)$, is the set of vertices that straddle the cut $S,S^c$ i.e., vertices which have at least one incident edge crossing the cut. Analogously, we can define the symmetric vertex expansion of a set $S$ as $\Phi^{\rm sym}_G(S) = w(\partial^{\rm sym}_G(S)/w(S)$. The $\delta$-symmetric vertex expansion of a graph is the smallest symmetric vertex expansion in $G$ among all sets with relative vertex weight $\delta$.}.
		\item[(iii)] The final hypergraph $H$ satisfies $\phi^{E}_{\delta}(H) \leq \Phi^{\sym}_{\delta}(G_{\rm sym})$ and $\sum_{e \in E} w(e) = \sum_{v \in V} W(v)  = |V_G|$. In particular, there exists a subset $S \subset V$ such that $w(S) = \delta w(V)$, $|S \cap V_G| = \delta|V_G|$ and $\phie_H(S) \leq \phi^V_\delta$. 
		\item[(iv)] If $G$ has max-degree at most $d$, then the final hypergraph $H$ is of arity at most $d$.
		\item[(v)] There is a one-to-one mapping $\pi: E \to V$ between the vertex set $V$ and edge set $E$ in $H$ which satisfies $\pi(e)  \in e$ for every $e \in E$.
		\item[(vi)] Furthermore, if there exists a subset of vertices $S' \subset V$ in $H$ such that $w(S') = \delta'w(V)$ and $\phie_{\delta'}(H) \leq \epsilon'$ then there exists a subset $S \subset V_G$ such that $|S| \in [(1 - \epsilon')\delta'|V_G|, \delta'|V_G|]$ such that $\phiv_G(S) \leq 2\epsilon'$. 
	\end{itemize}
\end{lemma}

We remark that in \cite{lm16}, the above lemma is stated for the ``at most'' version of the small set expansion problems i.e., where the volume of the set is constrained to be at most $\delta$ instead of being exactly equal to. However, it is easy to verify that their transformations and analyses work as is for the equals variant as well. For completeness, we include a sketch of it below for immediate reference and a full proof in Appendix \ref{sec:redn}.  

\begin{proof}[Proof sketch of Lemma \ref{lem:transform}]
	The construction is as follows. 
	
	{\bf Step (i)} $G \mapsto G_{\rm sym}$. Given $G = (V_G,E_G)$ we construct a weighted bipartite graph $G_{\sym} = (V_\sym,E_\sym,w_\sym)$ where $V_\sym = V \cup E$ and we add an edge between $v \in V_G$ and $e \in E_G$ if edge $e$ is incident to vertex $v$. For every $v \in V_G$, we assign weight $w_{\sym}(v) = 1$, and for every edge $(u,v) \in E_G$ we assign the corresponding vertex weight $w_\sym((u,v)) =  0$. 
	
	{\bf Step (ii)} $G_{\rm sym} \mapsto H$. Given $G_\sym = (V_\sym,E_\sym)$ we construct the weighted hypergraph $H = (V,E)$ with $V = V_{\rm sym}$. Furthermore, we include the following hyperedges: for every vertex $v \in V$, we introduce the hyperedge $\{v\} \cup N_{G_\sym}(v)$ with weight $w\left(\{v\} \cup N_{G_\sym}(v)\right) = w_\sym(v)$. Furthermore, for every vertex $v \in V_\sym$, we retain the vertex weight i.e, $W(v) = w_\sym(v)$
	
	It can be verified the above construction of $G_{\sym}$ and $H$ satisfies items (i)-(vi). For completeness we include a proof in the Appendix \ref{sec:redn}.
\end{proof}

\subsection{Lasserre/SoS Hierarchy}			\label{sec:lass}	

The Lasserre/SoS hierarchy~\cite{Shor87,Lass01} is a sequence of strengthening of the basic SDP which enforce constraints on subsets of variables of increasing sizes. Formally, given a basic SDP with variable set $X_1,\ldots,X_n$, the $R$-round Lasserre relaxation of the SDP optimizes the objective over the set of variables $\{X_{S,\alpha}\}_{|S| \leq R,\alpha \in \{0,1\}^S}$ where $X_{S,\alpha}$ is meant to be the indicator of the event that the subset $S$ gets assigned label $\alpha$. The variables $X_{S,\alpha}$ are associated with the so called degree-$R$ pseudo-distribution $\mu := \{\mu_{S}\}_{S: |S| \leq R}$ where $\mu_S$ is a probability distribution over labelings for the variables in $S$. While the various local distributions $\{\mu_S\}$ may be not be globally consistent, the hierarchy ensures that they are consistent upto $R$-sized subset of variables by enforcing constraints among the $\mu_S$ variables. As is standard, we shall use the ``widetilde'' notation to denote various quantities defined with respect to the pseudo-distribution -- for e.g., we use $\tPr$ notation to denote the pseudo-probability and so on. Our analysis will not require the full strength of the hierarchy, and will essentially rely on a few well-known properties of degree-$2$ pseudo-distributions which we formally describe below. For interested readers, we refer them to \cite{LaurentSoS} for a more exhaustive treatment of the topic.

{\bf Degree-$2$ SoS distributions} The degree-$2$ SoS lifting of an SDP is the SDP itself with additional degree-$2$ Lasserre constraints. Hence, a degree-$2$ pseudo-distribution $\mu$ can be equivalently expressed as a vector solution comprising of a set of unit vectors $\{v_i\}_{i \in [n]} \cup \{\uphi\}$ where $\uphi$ is the ``one vector'' of the solution satisfying $\|\uphi\|^2 = 1$. In particular, in an integral assignment, the $v_i$'s are intended to be $\pm 1$ indicators. A useful feature of the vector representation is that due to the Lasserre constraints, linear algebraic functions of these vectors express several useful parameters of the pseudo-distribution. We list some of the key such properties below.

\begin{proposition}					\label{prop:lass}
	Every degree-$2$ pseudo-distribution can be equivalently identified with a vector solution $\{v_i\} \cup \{\uphi\}$ consisting of unit vectors such that the following holds. For every $i \in V$, let $v_i = (1 - 2\mu_i)\uphi - 2z_i$, where $-2z_i$ is the component of $v_i$ orthogonal to the direction $\uphi$. Then:
	\begin{itemize}
		\item[(i)] We can write the corresponding $\{0,1\}$-vectors as $u_{i} = (\uphi - v_i)/2$. Equivalently, every $i \in V$ we have $u_i = \mu_i \uphi + z_i$. 
		\item[(ii)] For every $i \in V$, we have $\mu_i = \tPr_{x_i \sim \mu} \left[x_i = 1\right]$ and $\|z_i\|^2 = \mu_i(1 - \mu_i)$.
		\item[(iii)] For every $i,j \in V$, we have $\langle z_i,z_j \rangle = \widetilde{{\rm Cov}}_\mu(x_i,x_j)$.
		\item[(iv)] For every $i,j \in V$, we have $(1/4)\|v_i - v_j\|^2 = \tPr_{x_i,x_j \sim \mu} \left[x_i \neq x_j\right]$.
	\end{itemize}
\end{proposition}
\begin{proof}
	Let $X_0 \equiv 1, X_1,\ldots,X_n$ be the $\pm 1$ valued pseudo-variables corresponding to the degree-$2$ pseudo-distribution. Let $M := \widetilde{\Ex}[XX^T] \in \mathbbm{R}^{\{0\} \cup [n] \times \{0\} \cup [n]}$ be the degree-$2$ pseudo-moment matrix corresponding to the degree-$2$ pseudo-distribution such that $M[i,j] = \tEx\left[X_iX_j\right]$, . In particular, we have that ${\rm diag}(M) = {\bf 1}^{n + 1}$. Since $M$ is the degree-$2$ pseudo-moment matrix, it must be PSD, and hence it can be decomposes as $M = V^\top V$, where $V = [v_0 \equiv\uphi~v_1 \cdots v_n]$ is a matrix with columns consisting of unit vectors. In particular, for every $i \neq j$, we have $\langle v_i,v_j \rangle = \tEx\left[X_iX_j\right]$.
	
	Now let $x_0 \equiv 1, x_1,\ldots,x_n$ be the corresponding $\{0,1\}$-valued pseudo-variables, where $x_i = \frac12(1 - X_i)$.  Due to the above decomposition, for every $i \in [n]$, we have 
	\begin{equation}				\label{eqn:prop}
	1 - 2\tPr\left[x_i = 1\right] = \tEx\left[X_i\right] = \tEx\left[X_iX_0\right] = \langle v_i,v_0 \rangle.
	\end{equation}
	Hence, denoting $\mu_i := \tPr\left[X_i = 1\right]$ for every $i \in [n]$, we can write $v_i = (1 - 2\mu_i)\uphi - 2z_i$ where $-2z_i$ is the component of $v_i$ in the subspace orthogonal to $\uphi$. Now we are ready to prove the above items one by one.
	
	{\bf Item (i)}: Let $u_i = \frac12(\uphi - v_i)$. Then $\uphi,u_1,\ldots,u_n$ are the corresponding $\{0,1\}$ vector solution. In particular, they satisfy
	\begin{align*}
	\langle u_i,u_j \rangle 
	&= \frac14\Big(\|\uphi\|^2 - \langle v_i + v_j,\uphi \rangle + \langle v_i,v_j \rangle\Big) \\
	&= \frac14\left(1 - (1 - 2\mu_i) - (1 - 2\mu_j) + \tEx\left[X_iX_j\right] \right)			\\
	&= \frac14\left(2\mu_i + 2\mu_j - 1 + \tEx\left[(1 - 2x_i)(1 - 2x_j)\right]\right) \\
	&= \tEx\left[x_ix_j\right] \\
	&= \tPr\left[x_i = 1,x_j = 1\right].
 	\end{align*}
	In particular $U = [\uphi u_1 \cdots u_n]$ satisfy $\tEx\left[xx^\top\right] = U^\top U$.
	
	{\bf Item (ii)}: The first part follows from \eqref{eqn:prop}. For the second part, observe that 
	\[
	\mu_i = \|u_i\|^2_2 = \mu^2_i + \|z_i\|^2,
	\]
	which on rearranging gives us $\|z_i\|^2_2 = \mu_i(1 - \mu_i) = \widetilde{{\rm Var}}(x_i)$.
	
	{\bf Item (iii)}: Again using the decomposition for $u_i$ vectors, we get that
	\[
	\tPr\left[X_i = 1,X_j = 1\right] = \langle u_i,u_j \rangle = \mu_i\mu_j + \langle z_i,z_j \rangle
	\]
	which on rearranging gives us
	\[
	\langle z_i,z_j \rangle = \tPr\left[x_i = 1,x_j = 1\right] - \tPr\left[x_i = 1\right]\tPr\left[x_j = 1\right] = \widetilde{\rm Cov}(x_i,x_j).
	\]
	{\bf Item (iv)}: Finally,
	\[
	\|u_i - u_j\|^2 = \|u_i\|^2 + \|u_j\|^2 - 2\langle u_i,u_j \rangle = \tPr\left[x_i = 1\right] + \tPr\left[x_j = 1\right] - 2\tPr\left[x_i = 1,x_j = 1\right] = \tPr\left[x_i \neq x_j\right].
	\]	
\end{proof}
	
In addition, we need the following lemma which will determine the running time of our algorithm.

\begin{lemma}~\cite{Shor87,Lass01}				\label{lem:sos-time}
	The $R$-round Lasserre lifting of the basic SDP on $n$ variables can be solved in time $n^{O(R)}$ up to polynomial precision.
\end{lemma}

\subsection{Gaussian CDF Facts}				\label{sec:guass}

The following is a basic difference bound for $\Phi$.

\begin{fact}[Folklore]
	\label{fact:cdf1}
$|\Phi(x + z) - \Phi(x)| \leq \frac{|z|}{\sqrt{2\pi}}$ for any $x,z \in \R$. 
\end{fact}
\begin{proof}
	This follows from the fact that $\phi(x) \leq 1/\sqrt{2\pi}$ for all $x \in \mathbbm{R}$.
\end{proof}

Now we recall some well-known facts which follow from the asymptotic approximation of Gaussian CDF function $\Phi(\cdot)$.

\begin{fact}[Eq. 7.1.13~\cite{AS65}]				\label{fact:cdf-estim}
	For any $t \in (-\infty,0)$ we have 
	\[
	\frac{1}{\sqrt{2 + t^2} + |t|} \frac{1}{\sqrt{2\pi}}e^{-t^2/2} \leq \Phi(t) \leq \frac{1}{|t|}\frac{1}{\sqrt{2\pi}}e^{-t^2/2}.
	\]
\end{fact}
This directly leads to the following bounds on $\Phi^{-1}(\cdot)$.

\begin{fact}[Folklore]				\label{fact:cdf}
	For any $\mu \in (0,1/2)$, we have $|\Phi^{-1}(\mu)| \leq \max\left(2\sqrt{\log\frac1\mu},O(1)\right)$.
\end{fact}

\begin{fact}[Folklore]				
\label{prop:half-error}
	There exists $\epsilon_0 \in (0,1)$ such that for every $\epsilon \leq \epsilon_0$, and every $t \in (-\infty,0]$ we have 
	\[
	\Phi(t) - \Phi(t - \epsilon) \leq 4\epsilon\cdot\Phi(t)\sqrt{\log\frac{1}{\Phi(t)}}.
	\]
\end{fact}
\begin{proof}
	We first observe that the map $\ell \mapsto \Phi(\ell)$ is convex in $(-\infty,0)$, which can be argued as follows. By definition, we have $\Phi(\ell) = \int^\ell_\infty \phi(x) dx $ and therefore using Leibniz rule we have $\Phi'(\ell) = \phi(\ell)$. This implies that $\Phi''(\ell) = -\frac{\ell}{\sqrt{2\pi}}\cdot e^{-\ell^2/2}$ which is nonnegative in $(-\infty,0)$, thus establishing the convexity of $\Phi(\cdot)$ in $\mathbbm{R}_{\leq 0}$. Therefore using the convexity of $\Phi(\cdot)$ we get that 
	\begin{eqnarray*}
	\Phi(t) - \Phi(t- \epsilon) \leq \epsilon\cdot\Phi'(t)  &\leq& \epsilon \cdot \frac{d}{dt}\int^t_{-\infty} \phi(x)dx \\
	&\leq& \epsilon\cdot\phi(t) \\
	&\overset{1}{\leq}& 2\epsilon\cdot\Phi(t) |t| \\
	&\leq& 4\epsilon\cdot \Phi(t)\sqrt{\log \frac{1}{\Phi(t)}}, 
	\end{eqnarray*}
	where step $1$ is due to Fact \ref{fact:cdf-estim}.
\end{proof}

 \begin{fact}[Folklore]		
	 \label{fact:incr}
 	Let $t \in (-\infty,0)$ and $\Delta \in (0,1)$ be such that $\Delta|t| \leq 1$. Then,
 	\[
 	\Phi(t + \Delta) - \Phi(t) \leq C\Phi(t)\Delta\sqrt{\log\frac{1}{\Phi(t)}}
 	\]
 	for some absolute constant $C > 0$. In particular, $C$ can be taken to be $24$.
 \end{fact}
 \begin{proof}
	 Let $\delta := \Phi(t)$. Now suppose $t \geq - 4$. Then, $\delta \geq \Phi(-4) = : \delta_0$ and since $ x \to x\sqrt{\log(1/x)}$ is a increasing function for $x \in (0,1)$ we have $\delta\sqrt{\log(1/\delta)} \geq \delta_0\sqrt{\log(1/\delta_0)} := C$. Hence using this we can bound:
	 \[
	 \Phi(t + \Delta) - \Phi(t) \leq \Delta \leq \frac1C \Delta \cdot \delta \sqrt{\log(1/\delta)},
	 \] 
	 which establishes the inequality for $t \geq -4$. Now we proceed to address the case where $t \leq -4$. Since $\Delta \in (0,1)$ in the setting of the lemma, we have $t + \Delta \leq -3$. Now, using the definition of $\Phi$ we proceed as follows:
 	\begin{eqnarray}
 	\Phi(t + \Delta) - \Phi(t) = \int^{t + \Delta}_t \frac{1}{\sqrt{2\pi}}\cdot e^{-t^2/2}dt 
 	&\leq& \Delta \cdot \frac{e^{-(t+ \Delta)^2/2}}{\sqrt{2\pi}}\cdot  		\non\\
 	&\leq& \Delta \cdot \Phi(t + \delta)(|t + \Delta|) 			\non\\
 	&\leq& 2\Delta \cdot \Phi(t + \Delta)\sqrt{\log\frac{1}{\Phi(t + \Delta)}} 		\label{eqn:fact-eq} 
 	\end{eqnarray}
 	Therefore, all that remains is to bound $\Phi(t + \Delta)$ with $\Phi(t)$. From Fact \ref{fact:cdf-estim} and the fact that $t,t+ \Delta \leq -2$ we have   
 	\[
 	 \frac{1}{2|a|} e^{-a^2/2} \leq \Phi(a) \leq \frac{1}{|a|} e^{-a^2/2}
 	\]
 	for $a = t,t+ \Delta$. Hence,
 	\[
 	\frac{\Phi(t + \Delta)}{\Phi(t)} \leq 4\left(1 + \frac{\Delta}{|t|}\right)^{-1}\cdot e^{-(|t| - \Delta)^2/2 + t^2/2} \leq 4e^{-\Delta^2/2 + \Delta|t|} \leq 12
 	\]
 	where in the last step we used $\Delta|t| \leq 1$. Plugging in the above bound into \eqref{eqn:fact-eq} give us 
 	\[
 	\Phi(t + \Delta) - \Phi(t) \leq 8\delta\Delta\cdot\sqrt{\log \frac1\delta}.
 	\]
 \end{proof}

\begin{fact}		\label{fact:gauss-max}
	Let $(g_1,\ldots,g_d)$ be {\em jointly distributed} where marginally $g_i$ is distributed $N(0,1)$ for every $i \in [d]$. Then for any constant $C \geq 4$ we have 
	\[
	\Pr_{g_1,\ldots,g_d} \left[\max_{j \in [d]} |g_j| \geq \sqrt{C\log d}\right] \leq e^{-C\log d/4}.
	\]
\end{fact}
\begin{proof}
	Again using the estimates from Fact \ref{fact:cdf-estim} we get that 
	\begin{align*}
		\Pr_{g_1,\ldots,g_d} \left[\max_{j \in [d]} |g_j| \geq \sqrt{C\log d}\right]
		 \leq \sum_{j \in [d]} \Pr_{g_j \sim N(0,1)} \left[|g_j| \geq \sqrt{C \log d}\right] 
		 \leq 2d \cdot e^{-C\log d/2} 			 
		 \leq e^{-C\log d/4}. 		
	\end{align*}
\end{proof}

\begin{fact}		\label{fact:gauss-max1}
	Let $(g_1,\ldots,g_d)$ be {\em jointly distributed} Gaussians which are marginally $N(0,1)$. Then for some absolute constant $C > 0$ we have
	\[
	\Ex_{g_1,\ldots,g_d \sim N(0,1)} \left[\max_{i \in [d]} g^2_j\right] \leq C \log (d)
	\]
\end{fact}
\begin{proof}
	Let $\zeta_{\rm max}$ denote the random variable $\max_{i \in [d]} g_i$. By definition,
	\begin{align*}
	\Ex_{g_1,\ldots,g_d \sim N(0,1)} \left[\max_{i \in [d]} g^2_j\right] 
	&= \int^{\infty}_0 \Pr\left[ \zeta^2_{\rm max} \geq t\right] dt \\
	&\leq  C \log d + \int^{\infty}_{C \log d} \Pr\left[ \zeta^2_{\rm max} \geq t\right] dt \\
	&\leq  C \log d + \int^{\infty}_{\alpha = C } \log (d) \Pr\left[ \zeta_{\rm max} \geq \sqrt{\alpha \log (d)}\right] d\alpha  \tag{Change of variables $t = \alpha \log (d)$}\\
	&\leq  C \log d + \int^{\infty}_{\alpha = C } \log (d) e^{-\alpha \log (d)/2} d\alpha \\
	&\leq  C \log d + o(1). 
	\end{align*}
\end{proof}

\begin{fact}		\label{fact:gauss-max0}
	Let $(g_1,\ldots,g_d)$ be {\em jointly distributed} Gaussians whose marginal distributions are $N(0,1)$. Then,
	\[
	\Ex_{g_1,\ldots,g_d} \left[\max_{i \in [d]} |g_i| \right] \leq \sqrt{C\log d},
	\]
	for some absolute constant $C>0$.
\end{fact}
\begin{proof}
	Using Jensen's inequality and Fact \ref{fact:gauss-max1} we have
	\[
	\Ex_{g_1,\ldots,g_d} \left[\max_{i \in [d]} |g_i| \right]
	\leq \Ex_{g_1,\ldots,g_d} \left[\max_{i \in [d]} g^2_i \right]^{1/2}
	\leq \sqrt{C \log d}.
	\]
\end{proof}

\section{Approximation Algorithm for \ssve}

In this section, we prove Theorem \ref{thm:hsse-main}. The algorithm for Theorem \ref{thm:hsse-main} is described in Algorithm \ref{alg:approx-pred}.

	\begin{algorithm}[ht!]
		\SetAlgoLined
		\KwIn{A graph $G = (V_G,E_G)$ of maximum degree $d$, volume parameter $\delta \in (0,1/2]$} 
		Set parameter $t \gets (1/\delta)^{24}$\;
		Construct the hypergraph $H = (V,E,w,W)$ from $G$ using the algorithm from Lemma \ref{lem:transform}\;
		Solve the following $R \defeq (t+2)$-round Lasserre relaxation of the following SDP:  
		\begin{align}
			\textnormal{Minimize}~& \frac{1}{\delta W(V)} \sum_{e \in E}w(e) \max_{i,j \in e}\tPr_{(X_i,X_j) \sim \mu_{ij}}\left[X_i \neq X_j\right] 
			\label{eqn:sdp-1}\\
		& \E_{i \sim V_G}\tPr_{X_i \sim \mu_i } \Big[X_i = 1\Big] =  \delta 		\label{eqn:spd-2}\\
		& \left|\mu_{X_i\gets 1} - \mu_{X_j\gets 1} \right| \leq \tPr_{\mu_{i,j}}\left[X_i \neq X_j\right] & \forall i,j \in e, \forall e \in E 
		\label{eqn:sdp-3}
		\end{align}
		where $\mu_{X_i \gets 1}$ is the probability of the local variable $X_i$ being set to $1$ under $\mu$. \label{step:solve}\;
		Let $\mu^*:= \{\mu^*_{T,\alpha}\}$ be the degree-$R$ pseudo-distribution corresponding to the optimal solution \label{step:condition} \;
		Sample a uniformly random subset $T \subseteq V$ of size at most $t$ and an assignment from the local distribution $X_T \gets \alpha \sim \mu^*_T$\;
		Let $\mu$ be the degree-$2$ pseudo-distribution obtained by conditioning $\mu^*$ on $X_T \gets \alpha$ \;
		Let $\Delta_e = \max_{i,j \in e} \tPr_{\mu_{ij}| X_S \gets \alpha} \left[X_i \neq X_j\right]$ be the contribution from $e \in E$\;
		Delete all the hyperedges with $\Delta_e \geq 1/10$ and consider them as cut \label{step:e-del}   \linebreak
		
		{\bf New Vector Solution}.  \label{step:newvec}\\
		Construct a new vector solution $\{v'_i\}_{i \in V}$ as follows. 
		Let $\{v_i\}_{i \in V} \cup \{\uphi\}$ be the vector solution corresponding to $\mu$ from Proposition \ref{prop:lass}\;  
		For every $i \in V$, write $v_i = (1 - 2\mu_i) \uphi - 2 z_i$ as in Proposition \ref{prop:lass}\;
		Let $\hat{z}$ be a unit vector orthogonal\footnotemark[11]  to $\uphi$ and $\{z_i\}_{i \in V}$,
		and let $\theta \gets \delta^{12}$.
		For every $i \in V$, define \label{step:vi'} 
		\begin{equation}			\label{eqn:vi-def}
		v'_i \defeq \frac{v_i - \theta \hat{z}}{\sqrt{1 + \theta^2}}.
		\end{equation}
		{\bf Shifted Hyperplane Rounding}. \label{step:gaussproj} \\ 
		Write $v'_i = (1 - 2\mu'_i) \uphi - 2z'_i$ where $\langle \uphi,z'_i \rangle = 0$\;
		Sample Gaussian $g \sim N(0,1)^r$ and construct $S \subset V$ as 
		\[
		S:= \left\{i \in S| \Langle g, \frac{z'_i}{\|z'_i\|} \Rangle \leq \Phi^{-1}(\mu'_i) \right\}.
		\]
		where $r$ is the ambient dimension of the vector solution $\{v'_i\}_{i \in V} \cup \{\uphi\}$.
		\label{step:round}\;
		{\bf Rolling Back to $G$}. Use item (vi) of Lemma \ref{lem:transform} to find a subset $S' \subset V_G$ in $G$ such that 
		$|S'| \in [0.99\delta|V_G|,1.01\delta|V_G|]$ $\phiv_G(S') \leq 2\phie_H(S)$. \label{step:final}\;
		Output subset $S'$\;
		\caption{Approximation algorithm for \ssve} 
		\label{alg:approx-pred}
	\end{algorithm}	
	\footnotetext{Suppose the ambient dimension of the vector solution is $r$, and the dimension of the span of the vectors $\{\uphi\} \cup \{z_i\}_{i \in V}$ is less than $r$, then clearly we can find such a $\hat{z}$ in $\mathbbm{R}^r$. Otherwise, we treat the vectors to be in $\mathbbm{R}^{r + 1}$ (by padding each vector by a $0$), and then we can find such a $\hat{z}$ in $\mathbbm{R}^{r + 1}$.}

\paragraph{The objective function.} 
Note that the objective is a weighted linear combination of terms $\max_{i,j} \tPr\left[X_i \neq X_j\right]$ terms, which is convex. To see this, observe that for any $e \in E$, $c_{e}$ and include the constraints 
\[
\forall e \in E, i,j \in e : \qquad \tPr_{\mu_{ij}} \left[X_i = 1,X_j = 0\right] + 
\tPr_{\mu_{ij}} \left[X_i = 0,X_j = 1\right] \leq c_{e}
\] 
and minimize the objective
\[
\frac{1}{\delta W(V)} \Ex_S \Ex_{X_S \sim \mu_S} \sum_{e \in E} w_e c_{e}
\]
which is weighted linear combination of variables, and therefore convex. Hence, it follows that we can solve the program in Step \ref{step:solve} up to polynomial precision in polynomial time (Lemma \ref{lem:sos-time}),  

\paragraph{The $\ell_1$-constraints}

Note that for any $i,j \in V$, the constraint $|\mu_{X_i \gets 1} - \mu_{X_j \gets 1}| \leq \tPr_{(X_{i},X_j) \sim \mu_{ij}}\left[ X_i \neq X_j\right]$ is just a linear constraint on the variables of the relaxation i.e.,
\[
\left| \tPr_{X_i \sim \mu_i} \Big[X_i = 1\Big] - \tPr_{X_j \sim \mu_j}\Big[X_j = 1\Big] \right| \leq 
\tPr_{(X_i,X_j) \sim \mu_{ij}}\left[X_i = 1,X_j = 0\right] + \tPr_{(X_i,X_j) \sim \mu_{ij}}\Big[X_i = 0, X_j = 1\Big].
\]
Furthermore, we point out that any SoS feasible solution where the number of rounds is at least $2$ automatically satisfies these constraints.
\begin{claim}
	Let $\mu$ be a pseudo-distribution of degree at least $2$. Then $\mu$ immediately satisfies the constraints from \eqref{eqn:sdp-3}.
\end{claim}
\begin{proof}
	Let $\{u_i\}_{i \in V} \cup \{\uphi\}$ be the vector solution corresponding to the pseudo-distribution. Note that since $\mu$ is a pseudo-distribution, the corresponding vector solution satisfies the $\ell^2_2$ triangle inequality constraints. Hence, for any $i,j \in V$, we have
	\begin{align*}
		\left|\tPr_{X_i}\left[X_i = 1\right] - \tPr_{X_j}\left[X_j = 1\right]\right|
		&= \left|\|u_i\|^2_2 - \|u_j\|^2_2\right|			\\
		&\leq \|u_i - u_j \|^2_2 \\
		&= \tPr_{X_i,X_j}\left[X_i \neq X_j\right]. 
	\end{align*}
\end{proof}

\subsection{Analysis of Algorithm \ref{alg:approx-pred}}

Our analysis begins with the following lemma which says that for the \hsse~instance $H(V,E,w)$ constructed from the graph $G$, the optimal value of the SDP in Step \ref{step:solve} is at most $\phiv_\delta$.

\begin{lemma}					\label{lem:alg-feas}
	Let $H = (V,E,w,W)$ be the corresponding hypergraph obtained by instantiating Lemma \ref{lem:transform} with $G$. Then the SDP in step \ref{step:solve} of Algorithm \ref{alg:approx-pred} is feasible. Moreover, its optimal value is at most $\phiv_\delta := \phiv_{\delta}(G)$.
\end{lemma}

\begin{proof}
	We first observe that combining items (ii) and (iii) of Lemma \ref{lem:transform}, we have 
	\begin{equation}
		\phie_\delta(H) \leq \Phi^{\sym}_{\delta}(G_{\rm sym}) \leq \phiv_\delta(G) = \phiv_\delta.
	\end{equation}
	Let $S \subset V$ be the subset guaranteed by the item (iii) of Lemma \ref{lem:transform} satisfying the conditions:
	\[
	(a)~~~w(S) = \delta \cdot w(V) \qquad\qquad (b)~~~|S \cap V_G| = \delta |V_G| \qquad\qquad (c)~~~w\left(\partial^E_H(S)\right) \leq \phiv_\delta \cdot w(S).
	\]
	Now let $\mu$ be the degree-$R$ psuedo-distribution corresponding to the integral assignment $\mathbbm{1}_S$, i.e., for any subset $T \subset V$ and any assignment $\alpha \in \{0,1\}^T$ we have
	\[
	\tPr_{X_T \sim \mu_T} \left[X_T = \alpha\right] = \prod_{i \in T, \alpha(i) = 1}\mathbbm{1}_S(i).
	\]
	As is standard, we can show that $\mu$ is a feasible pseudo-distribution assignment to the convex program in step 1. Clearly, since $\mu$ corresponds to an integral assignment, it satisfies all the degree-$R$ SoS constraints. Furthermore, observe that 
	\[
	\Ex_{i \sim V_G} \tPr\left[X_i = 1\right] = \frac{|S \cap V_G|}{|V_G|} = \frac{\delta |V_G| }{|V_G|} = \delta,
	\]
	where the second equality follows from condition (b) from above. Furthermore, since $\mu$ corresponds to an integral assignment, the above arguments hold for any conditioning $X_A \gets \alpha$. Finally, fix any hyperedge $e \in E$, and $i,j \in e$. Then,
	\[
	|\mu_{X_i \gets 1} - \mu_{X_j \gets 1}| = \left|\mathbbm{1}_S(i) - \mathbbm{1}_S(j)\right| \leq \max_{i',j' \in e} \left|\mathbbm{1}_S(i') - \mathbbm{1}_S(j')\right|
	= \max_{i,j \in e}\tPr_{X_i,X_j \sim \mu}\left[X_i \neq X_j\right],
	\]
	and again since $\mu$ corresponds to an integral assignment, the above holds for all conditionings $X_T \gets \alpha$. These establishes that $\mu$ is a feasible solution to the SoS relaxation in Step \ref{step:solve}.
	
	Finally, since $\mu$ is a feasible solution, the value attained by it is an upper bound on the optimal value of the relaxation. Hence, the optimal value is at most
	\begin{align*}
		&\frac{1}{\delta W(V)}\sum_{e \in E}w(e) \max_{i,j \in e}\tPr_{(X_i,X_j) \sim \mu_{ij}}\left[X_i \neq X_j\right]  \\
		& \leq \frac{1}{\delta W(V)} \sum_{e \in E}w(e) \max_{i,j \in e} |\mathbb{1}_S(i) - \mathbbm{1}_S(j)| \\
		& = \frac{w\left(\partial^E_H(S)\right)}{\delta W(V)}  \\
		& \leq \phiv_\delta \cdot \frac{w(S)}{\delta W(V)} 					\tag{Item (c)}\\
		& = \phiv_\delta \cdot \frac{\delta w(V)}{\delta W(V)} 				\tag{Item (b)}\\[3pt]
		& = \phiv_\delta.    					 \tag{Since $w(V) = W(V) = |V_G|$ from Lemma \ref{lem:transform}(iii)} 
	\end{align*}
	
\end{proof}

{\bf Proof Map}. The rest of the proof of Theorem \ref{thm:hsse-main} consists of several parts. We outline the overall structure of the proof here for accessibility.

\begin{enumerate}
	\item In Section \ref{sec:pre-proc}, we analyze the pre-processing step and show that the new vector solution $\{v'_i\}_{i \in V}$ approximately preserves the properties of the degree-$2$ solution $\{v_i\}_{i \in V}$ in addition to satisfying some new ones.
	\item In Section \ref{sec:cut-bound1}, we analyze the hyperplane rounding step and prove our main lemma which formally establishes that \eqref{eqn:cut-bound} holds for nice distributions (Lemma \ref{lem:cut-bound1}) - this is the key technical ingredient in the analysis of the approximation guarantee of our algorithm.
	\item In Section \ref{sec:conc}, we establish the concentration bound (Lemma \ref{lem:conc}) on the size of the set returned by the algorithm.
	\item Finally in Section \ref{sec:ssve-main} we combine all the ingredients to prove Theorem \ref{thm:hsse-main}.
\end{enumerate}

\subsection{Preprocessing Step}				\label{sec:pre-proc}

Recall that Step \ref{step:vi'} of Algorithm \ref{alg:approx-pred} defines 
\begin{equation} \label{eqn:vi'2}
v'_i \defeq \frac{(1 - 2\mu_i) \uphi - 2z_i - \theta\hat{z}}{\sqrt{1 + \theta^2}} = (1 - 2\mu'_i)\uphi - 2z'_i . 
\end{equation}

It is easy to see that the new set of vectors also have unit norm i.e, 
\begin{eqnarray*}
\|v'_i\|^2 &=& (1 + \theta^2)^{-1} \left\| (1 - 2\mu_i) \uphi - 2z_i  - \theta \hat{z} \right\|^2 \\
&=& (1 + \theta^2)^{-1} \left(\|(1 - 2\mu_i)\uphi - 2z_i\|^2 + \theta^2\|\hat{z}\|^2\right) \\
&=& (1 + \theta^2)^{-1} \left(\|v_i\|^2 + \theta^2\right) = 1
\end{eqnarray*}
where the second equality uses that $\hat{z}$ is orthogonal to every other vector, and the third inequality uses that $\hat{z}$ is unit norm. Furthermore, now writing $v'_i = (1 - 2\mu'_i) \uphi - 2z'_i$ and comparing the $\uphi$ and $\uphi^{\perp}$ components, we get that 
\begin{equation}				\label{eqn:mu-w-defn}
\mu'_i = \frac{\mu_i}{\sqrt{1 + \theta^2}} + \frac12\left(1 - \frac{1}{\sqrt{1 + \theta^2}}\right)
 \ \ \ \ \textnormal{and} \ \ \ \ 
 z'_i = \frac{z_i + (\theta/2)\hat{z}}{\sqrt{1 + \theta^2}}.
\end{equation}
In  particular, rearranging the above expression for $\mu'_i$ gives us the following useful identity:
\begin{equation} 			\label{eqn:mu-shift}
\mu'_i - \mu_i = \left(\frac12 - \mu_i\right)\left(1 - \frac{1}{\sqrt{1 + \theta^2}}\right).
\end{equation}
Using the above, we establish the following.

\begin{claim}				\label{cl:mu-bounds}
For $\theta^2 \leq 1/100$, and for every $i \in V$, we have 
	(i) $|\mu_i - \mu'_i| \leq \theta^2$ and 
	(ii) $\mu'_i \in [\theta^2/10, 1 - \theta^2/10]$.
\end{claim}
\begin{proof}
	 Fix a $i \in V$. From \eqref{eqn:mu-shift} we have 
	 \begin{eqnarray*}
		 |\mu_i - \mu'_i| &=& \left|\frac12 - \mu_i\right|\cdot\left(1 - \frac{1}{\sqrt{1 + \theta^2}}\right) 
	 \leq \frac12\left(1 - \frac{1}{\sqrt{1 + \theta^2}}\right) \\
	 &=& \frac{\sqrt{1 + \theta^2} - 1}{2\sqrt{1 + \theta^2}}
	 \leq \sqrt{1 + \theta^2} - 1 \leq \theta^2,
	 \end{eqnarray*}
 	where the last inequality is due to the fact $\sqrt{1 + \theta^2} \leq 1 + \theta^2$.
	For establishing (ii), suppose $\mu_i \in [1/4,3/4]$. Then, we get that 
	$\mu'_i \in [1/4 - \theta^2, 3/4 + \theta^2]$. For our range of $\theta^2$, we get that
	$[1/4 - \theta^2, 3/4 + \theta^2] \subseteq [\theta^2/10,1 - \theta^2/10]$.
	Now, let us assume that $\mu_i \in [0,1/4]$. From \eqref{eqn:mu-shift} we have 
	\begin{align*} 
	\mu'_i - \mu_i & = \paren{\frac{1}{2} - \mu_i} \cdot\left(1 - \frac{1}{\sqrt{1 + \theta^2}} \right)
	\geq \frac{1}{4} \left(\frac{\sqrt{1 + \theta^2} - 1}{\sqrt{1 + \theta^2}}\right) \\
	& = \frac{1}{4} \left(\frac{\theta^2}{\sqrt{1 + \theta^2} \paren{\sqrt{1 + \theta^2} + 1 } }\right) 
	\geq  \frac{\theta^2}{10}.
	\end{align*}
	Therefore, we get that $\mu'_i \geq \mu_i + \theta^2/10 \geq \theta^2/10$. We also have that 
	$\mu'_i \leq \mu_i + \theta^2 \leq 1/4 + \theta^2 \leq 1 - \theta^2/10$.

	The case $\mu_i \in [3/4,1]$ can be proved similar to the previous case.

\end{proof}
We shall also need the following which follows directly from the definition.
\begin{claim}					\label{cl:vprime-equal}
	For every $i,j \in V$ we have 
	\[
	\|v'_i - v'_j\|^2 = \frac{1}{1 + \theta^2}\|v_i - v_j\|^2.
	\]
\end{claim}
\begin{proof}
	Fix a $i,j \in V$. Then using \eqref{eqn:vi'2} we have:
	\begin{eqnarray*}
	\|v'_i - v'_j\|^2 
	&=& \left\|\frac{(1 - 2\mu_i)\uphi - 2z_i - \theta\hat{z}}{\sqrt{1 + \theta^2}} - \frac{(1 - 2\mu_j)\uphi - 2z_j - \theta\hat{z}}{\sqrt{1 + \theta^2}}\right\|^2\\
	&=& \frac{1}{1 + \theta^2} 
	\left\|\left((1 - 2\mu_i)\uphi - 2z_i\right) - \left((1 - 2\mu_j)\uphi - 2z_j\right)\right\|^2\\
	&=& \frac{1}{1 + \theta^2} \|v_i - v_j\|^2 .
	\end{eqnarray*}
\end{proof}

Lastly we shall need the following.
\begin{claim}				\label{cl:w-shift}
	For every $i,j \in V$, we have $|\langle z'_i,z'_j \rangle| \leq | \langle z_i,z_j \rangle | + \theta^2/4$.
\end{claim}
\begin{proof}
	By definition of the $z'_i$'s we can write 
	\begin{align*}
	|\langle z'_i, z'_j \rangle| 
	& = \frac{1}{1 + \theta^2}\left| \Langle \Big(z_i + (\theta/2) \hat{z}\Big),\Big( z_j + (\theta/2)\hat{z}\Big) \Rangle \right|\\
	& \leq \left| \Langle \Big(z_i + (\theta/2) \hat{z}\Big),\Big( z_j + (\theta/2)\hat{z}\Big) \Rangle \right|\\
	& = | \langle z_i,z_j \rangle + (\theta^2/4) \|\hat{z}\|^2 | & (\hat{z} \perp z_i,z_j) \\
	& = | \langle z_i,z_j \rangle + (\theta^2/4)  | & \textrm{(using $\|\hat{z}\|^2 = 1$)}\\
	& \leq | \langle z_i,z_j \rangle| + \theta^2/4 .
	\end{align*}
\end{proof}

\subsection{Simulating Degree-$2$ SoS Properties}
\label{sec:proxy}
In general, the vectors $\{v'_i\}_{i \in V}$ needn't form a degree-$2$ SoS solution. In this section, 
we show that they still satisfy some useful properties which we will use in our analysis of
Algorithm \ref{alg:approx-pred}.

\begin{claim}			\label{cl:mu-ell1}
	For any $i,j \in V$ we have 
	\[
	|\mu'_i - \mu'_j| \leq  2\|v'_{i} - v'_j\|^2.
	\]
\end{claim}
\begin{proof}
	For any $i,j \in V$,
	\begin{align*}
	|\mu'_i - \mu'_j | 
	& = \frac{|\mu_i - \mu_j|}{\sqrt{1 + \theta^2}}  \tag{Using \eqref{eqn:mu-w-defn}} \\
	& \leq \Pr_{X_{ij} \sim \mu}\left[X_i \neq X_j\right] 		\tag{Using \eqref{eqn:sdp-3}} \\
	& \leq \|v_i - v_j\|^2 										\tag{Using Proposition \ref{prop:lass}} \\
	& = (1 + \theta^2) \|v'_i - v'_j\|^2   \tag{using Claim \ref{cl:vprime-equal}} \\
	& \leq 2 \|v'_i - v'_j\|^2 .
	\end{align*}
\end{proof}

\begin{claim}[Cut-Correlation Bound]			\label{cl:corr1}
	For any $i,j \in V$
	\[
	1 - \langle \bar{z'_i},\bar{z'_j} \rangle \leq \frac{\|v'_i - v'_j\|^2}{8\|z'_i\|\|z'_j\|}.
	\]
\end{claim}
\begin{proof}
	For any $i,j \in V$,
	\begin{align*}
	\|v'_i - v'_j\|^2 
	& = 4(\mu'_i - \mu'_j)^2\|\uphi\|^2 + 4\|z'_i - z'_j\|^2 & \textrm{(using \eqref{eqn:vi'2} and $z_i,z_j,\hat{z} \perp \uphi$)} \\
	&\geq 4\|z'_i - z'_j\|^2 \\
	&= 4 (\|z'_i\|^2 + \|z'_j\|^2 - 2 \langle z'_i,z'_j \rangle) \\
		&\geq 4 (2\|z'_i\|\|z'_j\| - 2 \langle z'_i,z'_j \rangle) & \textrm{(AM-GM inequality)} \\
	&= 8 \|z'_i\|\|z'_j\|(1 -  \langle \bar{z'}_i, \bar{z'}_j \rangle). 
	\end{align*}
\end{proof}

\begin{claim}				\label{cl:mu-w-bound}
	For any $i \in V$ such that $\mu'_i \leq 1/2$, we have $\mu'_i \leq 4\|z'_i\|^2$. 
\end{claim}
	\begin{proof}
		Fix an $i \in V$ such that $\mu'_i \leq 1/2$.
		Using the definition from \eqref{eqn:mu-w-defn} we have 
		\[
		\|z'_i\|^2 = \left\| \frac{z_i + (\theta/2)\hat{z}}{\sqrt{1 + \theta^2}} \right\|^2 = \frac{\|z_i\|^2 + \theta^2/4}{1 + \theta^2} .
		\]
		Note that from \eqref{eqn:mu-shift} we get that $\mu'_i \leq 1/2$ if and only if $\mu_i \leq 1/2$. 	
		Therefore, here we have $\mu_i \leq 1/2$. Since $\|z_i\|^2 = \mu_i(1 - \mu_i)$ it follows that 
		\[
		\frac{\mu_i}{\sqrt{1 + \theta^2}} \leq \mu_i  \leq 2\|z_i\|^2 \leq \frac{4\|z_i\|^2}{1 + \theta^2}
			\qquad \forall \theta \in (0,1) .	\]
		 Similarly
		\[
		\frac12\left(1 - \frac{1}{\sqrt{1 + \theta^2}}\right)
		\leq \frac12\left(1 - \frac{1}{1 + \theta^2}\right)
		\leq \frac{\theta^2}{1 + \theta^2}
		\]
		where in the first step we use $\sqrt{1 + a} \leq 1 + a$ for every $a \geq 0$. 
		Using \eqref{eqn:mu-w-defn} and combining the bounds from above,
		\[
		\mu'_i = \frac{\mu_i}{\sqrt{1+ \theta^2}} + \frac12\left(1 - \frac{1}{\sqrt{1 + \theta^2}}\right)
		\leq 4\cdot\frac{\|z_i\|^2}{1 + \theta^2} + 4\cdot\frac{\theta^2/4}{1 + \theta^2} = 4 \|z'_i\|^2 .
		\]
	\end{proof}

\section{Shifted Hyperplane Rounding}				\label{sec:cut-bound1}

In this section, we prove our main technical ingredient (Lemma \ref{lem:cut-bound1}) which gives the probability of a hyperedge getting separated using Gaussian rounding (Step \ref{step:gaussproj}). 

\begin{restatable}{lem}{cutbound}
	\label{lem:cut-bound1}
	The following holds for every $d \geq 2$. Fix a hyperedge $e = (1,2,\ldots,d)$ and let $(\mu'_1,\ldots,\mu'_d)$ be the corresponding set of biases which satisfy $0 \leq \mu'_d \leq \mu'_{d-1}\leq \cdots \leq \mu'_1\leq 1/2$. Furthermore, suppose $\mu'_1 \geq \max\{A_{d,\delta} \nu_e,\delta^{C^*_0}\}$ where (i) $C^*_0 = 24$, (ii) $C^*_1 = 32C^*_0\log(1/\delta)$, and (iii) $A_{d,\delta} = 64 C^*_0 C^*_1 \log(1/\delta) \log(d)$. Then
	\[
	\Pr_S\left[e \textrm{ is cut } \right] \lesssim C^*_0\sqrt{\alpha_e \nu_e \log(1/\delta)\log (d)  },  
	\]
	where $\alpha_e = \max_{i \in e} \min\{\mu'_i,1 - \mu'_i\}$ and $\nu_e:= \max_{i,j \in e} \|v'_i - v'_j\|^2$.
\end{restatable}

Towards proving Lemma \ref{lem:cut-bound1}, we establish some basic preliminaries and conventions that will be used in this section. Throughout the rest of the section, we will be working in the setting of Lemma \ref{lem:cut-bound1}. Let $\cG = (g_1,\ldots,g_d)$ be the collection of Gaussians in the setting of the Lemma \ref{lem:cut-bound1} defined as 
\[
g_i \defeq \Langle g, \frac{z'_i}{\|z'_i\|} \Rangle = \langle g, \bar{z'_i} \rangle,
\] 
Note that the event ``$e$ is cut'' translates to the following event in terms of the Gaussians:
\[
\Big\{\exists i,j \in e : g_i \leq t_i \wedge g_j > t_j\Big\}.
\]
The biases for the rounding, denoted by $\mu'_1,\ldots,\mu'_d$, are identified with thresholds $t_1,\ldots,t_d$ such that $t_i = \Phi^{-1}(\mu'_i)$. We use $\rho_{ij}$ to denote the correlation between $g_i,g_j$ and note that $\rho_{ij} = \langle \bar{z'_i},\bar{z'_j} \rangle$ using the definition of $(g_i)_{i \in e}$. For brevity, we shall use $\rho_j$ to denote $\rho_{1,j}$ for every $j \geq 2$. Define $\tilde{\theta}_m$ as  
\begin{equation}   				\label{eqn:theta-def}
\tilde{\theta}_m \defeq \max_{2 \leq j \leq d} 1 - \rho^2_j 
\end{equation}
and for the specific setting of the hyperedge $e = [d]$, let it be attained by some $j^* \in [d] \setminus \{1\}$. The following lemma states some inequalities used in the rest of the proofs. 
\setcounter{theorem}{1}
\begin{lemma}				\label{lem:ineq}
	Suppose the conditions in the setting of Lemma \ref{lem:cut-bound1} hold for the hyperedge $e = [d]$. Denote $\nu_e = \max_{i,j \in e} \|v'_i - v'_j\|^2$. Then the parameters $\mu'_i,t_i$ and the correlations $\rho_j = \rho_{1j}$ defined as above satisfy the following properties.
	\begin{itemize}
		\item[(a)] $\max_{i,j \in e} | \mu'_i - \mu'_j| \leq 2\nu_e$.
		\item[(b)] $\max_{i,j \in e} \mu'_i/\mu'_j \leq 2$.
		\item[(c)] $\tilde{\theta}_m \leq 2\frac{\nu_e}{\sqrt{\mu'_1\mu'_j}}$ for every $j \geq 2$.
		\item[(d)] $\tilde{\theta}_m \leq 2\nu_e/\mu'_i$ for every $i \in e$.
		\item[(e)] For every $i \in e$, we have $\mu'_i \geq \delta^{C^*_0}$.
		\item[(f)] For every $i \in e$, $|t_i| \leq 2\sqrt{C^*_0 \log(2/\delta)}$.
		\item[(g)] For every $j \geq 2$  we have $\rho_j \geq 0.9$.
	\end{itemize}
\end{lemma}
\begin{proof}
	For (a), for any $i,j \in V$, using Claim \ref{cl:mu-ell1} we have $|\mu'_i - \mu'_j| \leq 2\|v'_i - v'_j\|^2 \leq 2\nu_e$. 

	Recall that in the setting of the Lemma \ref{lem:cut-bound1} we have $\mu'_1 \geq A_{d,\delta} \nu_e \geq 100 \nu_e$. Hence for any $j,j' \in e$, using  item (a) we get that $\mu'_j \geq \mu'_1 - 2\nu_e \geq 0.98 \mu'_1 \geq 0.98 \mu'_{j'}$, which establishes (b). 

	For (c), using the definition from \eqref{eqn:theta-def} we see that
	\begin{align*}
	\tilde{\theta}_m = \max_{j \geq 2} (1 - \rho_j)(1 + \rho_j) 
	&\leq 2 \max_{j \geq 2}(1 - \rho_j) 			\\
	& = 2 \max_{j \geq 2} \paren{ 1 - \langle \bar{z'_1},\bar{z'_j} \rangle	}	\tag{Definition of $\rho_j$} \\
	&\leq \frac14\max_{j \geq 2}\frac{\|v'_1 - v'_j\|^2}{\|z'_1\|\|z'_j\|}		\tag{Claim \ref{cl:corr1}} \\
	&\leq \frac14\max_{j \geq 2}\frac{\nu_e}{\|z'_1\|\|z'_j\|}					\tag{Definition of $\nu_e$} \\
	&\leq \max_{j \geq 2}\frac{\nu_e}{\sqrt{\mu'_1\mu'_j}} 		\tag{Since $\mu'_1,\mu'_j \leq 1/2$ and Claim \ref{cl:mu-w-bound}}\\
		& \leq 2 \frac{\nu_e}{\sqrt{\mu'_1\mu'_i}} \qquad \forall i \in e  \tag{Using item (b)}	
	\end{align*}

	For (d), using $\mu'_1 \geq \mu'_i,\ \forall i \in e$, we get
	\[ \tilde{\theta}_m \leq 2 \frac{\nu_e}{\sqrt{\mu'_1\mu'_i}} 
		\leq \frac{2 \nu_e}{\mu'_i} \qquad \forall i \in e . \]

Item (e) follows from Claim \ref{cl:mu-bounds} and our setting of parameter $\theta^2 = \delta^{C^*_0}$. For item (f), we use item (e), Fact \ref{fact:cdf}, and the fact that the map $x \mapsto |\Phi^{-1}(x)|$ is decreasing in $x \in (0,1/2]$ to show
\[
|t_i| = \left|\Phi^{-1}(\mu'_i)\right| \leq \left| \Phi^{-1}(\delta^{C^*_0}) \right| \leq 2 \sqrt{C^*_0\log\frac1\delta} . 
\]
For item (g), we observe that for any $j \in \{2,\ldots,d\}$ we have
\begin{align*}
\rho_j & \geq  1 - \frac{\|v'_1 - v'_j\|^2}{8\|z'_1\|\|z'_j\|} 		\tag{Claim \ref{cl:corr1}}\\
					  & \geq 1 - \frac{\nu_e}{\sqrt{\mu'_1\mu'_j}} 		\tag{Claim \ref{cl:mu-w-bound}}\\
					  & \geq 1 - \frac{1}{A_{d,\delta}} 			\tag{Since $\mu'_1 \geq A_{d,\delta} \cdot \nu_e$}\\
					  &\geq 0.9.	
\end{align*}
\end{proof}

\newcommand{\tpo}{t_d}
\newcommand{\tdo}{t'_d}

\subsection{Proof of Lemma \ref{lem:cut-bound1}}		

\begin{proof}
	Throughout the proof, we shall assume that all the conditions of Lemma \ref{lem:cut-bound1} and their consequences (Lemma \ref{lem:ineq}) hold. Recall that $g_1,\ldots,g_d$ is the collection of correlated Gaussians associated with the vertices of the hyperedge $e$. We begin with the following claim.
	
	\begin{claim}				\label{cl:gauss-decomp}
		For every $i \in \{2,\ldots,d\}$, we can write $g_i = \rho_{i} \cdot g_1 + \sqrt{1 - \rho^2_{i}}\cdot \zeta_i$, where $(\zeta_2,\ldots,\zeta_d)$ are jointly distributed Gaussian random variables that are independent of $g_1$. In particular, $\zeta_i$ is marginally distributed as $N(0,1)$ for every $i \in \{2,\ldots,d\}$.
	\end{claim}	
	\begin{proof}
		For simplicity, let us denote $b_i := z'_i/\|z'_i\|$. Then by definition, we have $g_i = \langle g, b_i \rangle$ for every $i \in [d]$ where $g$ is distributed as $N(0,1)^r$. Furthermore, we also have $\langle b_1,b_i \rangle = \rho_{i}$. Hence for any $i \in \{2,\ldots,d\}$, we can write $b_i$ as a sum of its $b_i$ and $b^\perp_i$ components as 
		\[
		b_i = \rho_i \cdot b_1 + \sqrt{1 - \rho^2_i} \cdot b'_i,
		\] 
		where $b'_i$ is the unit vector which is the component of $b_i$ that is orthogonal to $b_1$. Hence, writing $\zeta_i := \langle g,b'_i \rangle$ for every $i \in \{2,\ldots,d\}$ we get that 
		\[
		g_i = \langle g, b_i \rangle = \rho_i \cdot\langle g,b_1 \rangle + \sqrt{1 - \rho^2_i} \cdot \langle g,b'_i \rangle = \rho_i \cdot g_1 + \sqrt{1 - \rho^2_i} \cdot \zeta_i.
		\]
		Note that since $b'_i$ is a unit vector, the random variable $\zeta_i = \langle g, b'_i \rangle$ is marginally distributed as $N(0,1)$. Finally, since $b'_i \perp b_1$ for every $i \geq 2$, it follows that $\zeta_2,\ldots,\zeta_d$ are independent of $g_1$.
	\end{proof}
	Let $\zeta_{\rm max} := \max_{j \geq 2} |\zeta_j|$, and let $\varphi: \mathbbm{R}^{\geq 0} \to \mathbbm{R}^{\geq 0}$ be the probability mass function corresponding to the distribution of $\zeta_{\rm max}$. The following is a straightforward consequence of Claim \ref{cl:gauss-decomp}.
	\begin{observation} 		\label{obs:zeta-indep}
		The random variable $\zeta_{\rm max}$ is independent of $g_1$.
	\end{observation}  
	Note that since $\zeta_{\rm max}$ is the maximum absolute value of $(d-1)$ (non-independent) Gaussian random variables, we can use Facts \ref{fact:gauss-max}, \ref{fact:gauss-max1}, and \ref{fact:gauss-max0} to derive useful bounds on the expectation and tail probabilities of $\zeta_{\rm max}$.
	\begin{lemma}				\label{lem:zeta}
		The random variable $\zeta_{\rm max} = \max_{j \geq 2} |\zeta_j|$ satisfies the following bounds for some absolute constant $C>0$.
		\begin{itemize}
			\item[(i)] $\Ex_{\zeta_2,\ldots,\zeta_d}\left[\zeta_{\rm max}\right] \leq \sqrt{C\log d}$. \\[-6pt]
			\item[(ii)] $\Ex_{\zeta_2,\ldots,\zeta_d}\left[\zeta^2_{\rm max}\right] \leq C\log d$.
			\item[(iii)] $\Pr_{\zeta_2,\ldots,\zeta_d}\left[\zeta_{\rm max} \geq \sqrt{C^*_1\log d}\right] \leq e^{-\frac{C^*_1\log d}{4}}$.
		\end{itemize}
	\end{lemma}
	\begin{proof}
		Since $\zeta_{\rm max}$ is the maximum absolute value of $\zeta_2,\ldots,\zeta_d$, each of which is marginally distributed as $N(0,1)$, items (i),(ii), and (iii) follow directly using Facts \ref{fact:gauss-max0}, \ref{fact:gauss-max1}, and \ref{fact:gauss-max} respectively.
	\end{proof}
	
	Recall that $t_i = \Phi^{-1}(\mu'_i)$ for every $i \in [d]$. Define
	\[
	\tdo \defeq \frac{\tpo}{\sqrt{1 - \tilde{\theta}_m}} = \frac{\tpo}{\rho_{j^*}}
	\]
	where recall that $\tilde{\theta}_m = 1 - \rho^2_{j^*}$ from \eqref{eqn:theta-def}.  We begin by making a couple of basic observations.
	\begin{claim}			\label{cl:lower-thr}
		Fix a realization of $\zeta_{\rm max} = a$ independent of $g_1$ (Observation \ref{obs:zeta-indep}) such that $a > 0$. Suppose $g_1 \leq \tdo - 2a \sqrt{\tilde{\theta}_m}$. Then for any other $j \geq 2$ we have $g_j < t_j$.
	\end{claim}
	\begin{proof}
	Fix any $j \geq 2$. Recall that $\rho_j \geq 0$ using Lemma \ref{lem:ineq}(g). 
	Using the upper bound on $g_1$, the fact that $1 - \rho_j^2 \leq \tilde{\theta}_m$, and  the fact that $|\zeta_j| \leq \zeta_{\rm max} = a$ we get that
	\begin{align*}
	g_j = \rho_j g_1 + \sqrt{1 - \rho^2_j} \cdot \zeta_j 
	&\leq \rho_j \paren{\frac{\tpo}{\sqrt{1 - \tilde{\theta}_m}} - 2a \sqrt{\tilde{\theta}_m}} + \sqrt{\tilde{\theta}_m} \cdot a \\
	& = \rho_j \frac{\tpo}{\rho_{j^*}} + (1 - 2\rho_j) a \sqrt{\tilde{\theta}_m} \\
	& \overset{1}{\leq} \tpo - 0.8 a \sqrt{\tilde{\theta}_m} < \tpo \leq t_j
	\end{align*}
	where inequality $1$ follows from the following observations. By definition of $j^*$, we have 
	$\rho_j \geq \rho_{j^*}$ for every $j \geq 2$, and $\tpo < 0$. Therefore, $\rho_j \tpo/\rho_{j^*} \leq \tpo$.
	The bound on the second term follows using $\rho_j \geq 0.9$ (Lemma \ref{lem:ineq}(g)).
	\end{proof}
	\begin{claim} 			\label{cl:upper-thr}
		Fix a realization of $\zeta_{\rm max} = a$ independent of $g_1$ (Observation \ref{obs:zeta-indep}) such that $a > 0$. Suppose $g_1 \geq t_1 + 2a\sqrt{\tilde{\theta}_m}$. Then for any $j \geq 2$ we have $g_j > t_j$.
	\end{claim}
	\begin{proof}
	Fix any $j \geq 2$, Then using the lower bound on $g_1$, and the bounds $\rho_j \geq 0$ and $|\zeta_j| \leq a$ we have
	\begin{eqnarray*}
	g_j = \rho_j g_1 + \sqrt{1 - \rho^2_j} \cdot \zeta_j 
	&\geq& \rho_j \left(t_1 + 2a\sqrt{\tilde{\theta}_m}\right) - a\sqrt{\tilde{\theta}_m} \\
	& = & \rho_j t_1 + (2\rho_j - 1) a \sqrt{\tilde{\theta}_m}  \\
	&\overset{1}{\geq} & t_j + 0.8 a\sqrt{\tilde{\theta}_m} > t_j \\
	\end{eqnarray*}
	where in inequality $1$, we use $t_j \leq t_1 \leq 0$ (since $\mu'_j \leq \mu'_1 \leq 1/2$) and $1 \geq \rho_j \geq 0.9$ (Lemma \ref{lem:ineq}(g)).
	\end{proof}
	\begin{claim}				\label{cl:t-bound}
		$|\Phi(\tpo) - \Phi(\tdo)| \leq |\tpo - \tdo| \leq 2 |\tpo|\tilde{\theta}_m$.
	\end{claim}
	\begin{proof}  	
		Using Fact \ref{fact:cdf1}, and $\tilde{\theta}_m = 1 - \rho^2_{j^*}$ we have
		\[
		|\Phi(\tpo) - \Phi(\tdo)| \leq |\tpo - \tdo| = \left|\frac{\tpo}{\rho_{j^*}} - \tpo \right|
		= |\tpo| \frac{1 - \rho_{j^*}}{\rho_{j^*}}
		= |\tpo| \frac{1 - \rho^2_{j^*}}{\rho_{j^*}(1 + \rho_{j^*})}
		\leq 2|\tpo|\tilde{\theta}_m
		\]
		where in the last term we use $\tilde{\theta}_m = 1 - \rho^2_{j^*}$ and $\rho_{j^*} \geq 0.9$ from Lemma \ref{lem:ineq}(g).
	\end{proof}

	Therefore, from Claims \ref{cl:lower-thr} and \ref{cl:upper-thr} it follows that fixing $\zeta_{\rm max} = a$ independent of $g_1$, the hyperedge can get cut only when $g_1 \in [\tdo - 2a\sqrt{\tilde{\theta}_m},t_1 + 2a\sqrt{\tilde{\theta}_m}]$. For ease of notation, we shall denote $\Delta(a,t_1,\tdo) = \Phi(t_1 + 2a \sqrt{\tilde{\theta}_m}) - \Phi(\tdo - 2a\sqrt{\tilde{\theta}_m})$. Furthermore, recall that $\varphi(\cdot)$ is the mass function of the random variable $\zeta_{\rm max}$. Then,
	\begin{align}
	\Pr\left[ e \mbox{ is cut} \right]
	&= \int^{\infty}_0 \Pr\left[e \mbox{ is cut} \Big| \zeta_{\rm max} = a \right] \varphi(a)da \nonumber \\
	&\leq \int^{\infty}_0 \Pr\left[ g_1 \in [\tdo -2 a \sqrt{\tilde{\theta}_m}, t_1 + 2a \sqrt{\tilde{\theta}_m} ] \Big| \zeta_{\rm max} = a  \right] \varphi(a)da \nonumber\\ 
	&= \int^{\infty}_0 \Pr\left[ g_1 \in [\tdo -2 a \sqrt{\tilde{\theta}_m}, t_1 + 2a \sqrt{\tilde{\theta}_m} ]\right] \varphi(a)da \tag{Observation \ref{obs:zeta-indep}} \\
	&= \int^{\infty}_0 \left(\Phi(t_1 + 2a\sqrt{\tilde{\theta}_m}) - \Phi(\tdo - 2a\sqrt{\tilde{\theta}_m})\right) \varphi(a)da \nonumber \\
	&= \int^{\infty}_0 \Delta(a,t_1,\tdo) \varphi(a)da \nonumber \\
	&= \int^{\infty}_{\sqrt{C^*_1 \log d}} \Delta(a,t_1,\tdo) \varphi(a)da + \int^{\sqrt{C^*_1 \log d}}_0 \Delta(a,t_1,\tdo) \varphi(a)da, 	\label{eqn:main}
	\end{align}
	where recall that $C^*_1 = 32 C^*_0 \log(1/\delta)$ and $C^*_0 = 24$ is a constant.

	\paragraph{Bounding the first term of \eqref{eqn:main}.} We begin by breaking the integral into intervals.
	\begin{align}
	&\int^{\infty}_{\sqrt{C^*_1 \log d}} \Delta(a,t_1,\tdo) \varphi(a)da 			\non\\
	&= \int^{\infty}_{\sqrt{C^*_1 \log d}} \paren{\Phi\paren{t_1 + 2a\sqrt{\tilde{\theta}_m}} - \Phi\paren{\tdo - 2a\sqrt{\tilde{\theta}_m}}} \varphi(a)da	\non \\
	&= \int^{\infty}_{\sqrt{C^*_1 \log d}}\left( \paren{\Phi\paren{t_1 + 2a\sqrt{\tilde{\theta}_m}} - \Phi(t_1)\right) + \left(\Phi(t_1) - \Phi(\tdo)\right) + \left(\Phi(\tdo) - \Phi\paren{\tdo - 2a\sqrt{\tilde{\theta}_m}}} \right)\varphi(a)da. 	\label{eqn:terms}
	\end{align}
	For the middle term of \eqref{eqn:terms} we have 
	\begin{align}
	& \int^{\infty}_{\sqrt{C^*_1 \log d}} \left(\Phi(t_1) - \Phi(\tdo)\right) \varphi(a)da	\non \\
	&= \int^{\infty}_{\sqrt{C^*_1 \log d}} \left(\Phi(t_1) - \Phi(\tpo) + \Phi(\tpo)- \Phi(\tdo)\right) \varphi(a)da 			\non\\
	&= \int^{\infty}_{\sqrt{C^*_1 \log d}} \left(|\mu'_1 - \mu'_d| + \Phi(\tpo)- \Phi(\tdo) \right) \varphi(a)da 		\tag{Since $\Phi(t_i) = \mu'_i$}	\non\\
	&\lesssim \int^{\infty}_{\sqrt{C^*_1 \log d}} \left(\nu_e + {|\tpo|\tilde{\theta}_m} \right) \varphi(a)da \tag{Lemma \ref{lem:ineq}(a), Claim \ref{cl:t-bound}}\\
	&= \left(\nu_e + |\tpo|\tilde{\theta}_m \right) \int^{\infty}_{\sqrt{C^*_1 \log d}} \varphi(a)da	\non \\
	&\leq \left(\nu_e + |\tpo|\tilde{\theta}_m \right) e^{-C^*_1\log(d)/4} 		\tag{Lemma \ref{lem:zeta}(iii)} 		\non\\
	&\leq \nu_e + \left(\frac{2\nu_e}{\mu'_1}\cdot2\sqrt{C^*_0 \log\frac2\delta} \right)\cdot e^{-C^*_1 (\log d)/4}  \tag{Lemma \ref{lem:ineq}(d), Lemma \ref{lem:ineq}(f)} \non \\
	&\leq \nu_e + \left(\frac{2\nu_e}{\delta^{C^*_0}}\cdot2\sqrt{C^*_0 \log\frac2\delta}\right)\cdot \delta^{8C^*_0 \log d} \tag{using $C^*_1 \geq 32C^*_0 \log(1/\delta)$ and $\mu'_1 \geq \delta^{C^*_0}$} \non \\ 
	&\lesssim \nu_e.								\label{eqn:mid-term}
	\end{align}

	For the remaining two terms in \eqref{eqn:terms}~we again observe that $|\Phi(x + z) - \Phi(x)| \leq |z|$ and therefore,
	\begin{align}
	&\int^{\infty}_{\sqrt{C^*_1 \log d}} \left(\Phi(t_1 + 2 a \sqrt{\tilde{\theta}_m}) - \Phi(t_1) + \Phi(\tdo) - \Phi(\tdo - 2 a \sqrt{\tilde{\theta}_m})\right) \varphi(a)da 		\non\\
	&\lesssim \sqrt{\tilde{\theta}_m}\int^{\infty}_{\sqrt{C^*_1 \log d}} a  \varphi(a) da   \non\\
	&= \sqrt{\tilde{\theta}_m} \Ex\left[\zeta_{\rm max} \mathbbm{1}_{\zeta_{\rm max} \geq \sqrt{C^*_1 \log d}}\right] 		\non\\
	&\leq \sqrt{\tilde{\theta}_m} \Ex\left[\zeta^2_{\rm max}\right]^{1/2} \Ex\left[\mathbbm{1}_{\zeta_{\rm max} \geq \sqrt{C^*_1 \log d}}\right]^{1/2} 
	\tag{Cauchy-Schwarz}		\non \\ 
	&\lesssim \sqrt{\tilde{\theta}_m \log d} \cdot \Pr\left[ \zeta_{\rm max} \geq \sqrt{C^*_1 \log d}\right]^{1/2} \tag{Lemma \ref{lem:zeta}(ii)}   \non \\
	&\leq \sqrt{\tilde{\theta}_m \log d} \cdot e^{-C^*_1 (\log d)/8} 		\tag{Lemma \ref{lem:zeta}(iii)}			\non \\
	&\lesssim \sqrt{\frac{\nu_e}{\mu'_1} \log d} \cdot e^{-C^*_1(\log d)/8}	\tag{Lemma \ref{lem:ineq}(d)} 			\non \\
	& \leq \sqrt{\frac{\nu_e}{\delta^{C^*_0}} \log d} \cdot \delta^{2 C^*_0 \log d}	\tag{using $C^*_1 \geq 32 C^*_0 \log(1/\delta)$ and $\mu'_1 \geq \delta^{C^*_0}$.} \non \\
	&\lesssim \sqrt{\mu'_1 \nu_e} = \sqrt{\alpha_e\nu_e} 				\label{eqn:rem-term}
	\end{align}
	where the penultimate step uses $\mu'_1 \geq \delta^{C^*_0}$ in the setting of the lemma (Lemma \ref{lem:ineq} (e)) , and the last step is due to $\alpha_e = \mu'_1$ in the setting of the lemma. In summary, combining the bounds from \eqref{eqn:mid-term} and \eqref{eqn:rem-term} we get that 
	\begin{equation}					\label{eqn:first-bound}
	\int^{\infty}_{\sqrt{C^*_1 \log d}} \Delta(a,t_1,\tdo) \varphi(a)da 
	\lesssim \nu_e + \sqrt{\alpha_e \nu_e}
	\end{equation} 
	
	\paragraph{Bounding the second term of \eqref{eqn:main}.}
	We begin by observing that whenever $\zeta_{\rm max} = a \leq \sqrt{C^*_1 \log(d)}$, by the setting of parameters in Lemma \ref{lem:cut-bound1} and using the bounds on $\tilde{\theta}_m,|t_1|$ from Lemma \ref{lem:ineq}(d),(f) we have 
	\begin{align}					
	2 a \cdot \sqrt{\tilde{\theta}_m} \cdot |t_1| 
	& \leq 2 \sqrt{C^*_1 \log(d)} \cdot \sqrt{\frac{2\nu_e}{\mu'_1}} \cdot 2\sqrt{2 C^*_0 \log(1/\delta)} \non \\
	& = \sqrt{\frac{64 C^*_1 C^*_0\log(d)\log(1/\delta)}{A_{d,\delta}}} \tag{Since $\mu'_1 \geq A_{d,\delta} \nu_e$} \non \\
	& = 1 ,  \label{eqn:cond1}
	\end{align} 
where the last step is due to $A_{d,\delta} = 64C^*_0C^*_1\log d\log(1/\delta)$ in the setting of Lemma \ref{lem:cut-bound1}. Then instantiating $\Delta = 2 a \sqrt{\tilde{\theta}_m}$ and $t = t_1$, we observe that the parameters $\Delta$ and $t$ satisfy the premise of Fact \ref{fact:incr} (from \eqref{eqn:cond1}), and therefore using Fact \ref{fact:incr} we get 
\begin{equation}
	\label{eq:t1a}
\Phi\paren{t_1 + 2a\sqrt{\tilde{\theta}_m}} - \Phi(t_1) \lesssim 
	\Phi(t_1) 2a\sqrt{\tilde{\theta}_m} \sqrt{\log\frac{1}{\Phi(t_1)}} .
\end{equation}
Next, observe that 
\[ \abs{\tpo - \paren{\tdo - 2a \sqrt{\tilde{\theta}_m}}} \leq |\tpo - \tdo| + 2a\sqrt{\tilde{\theta}_m} \ 
	\overset{\textnormal{Claim \ref{cl:t-bound}}}{\leq}  \ 2|\tpo| \tilde{\theta}_m + 2a\sqrt{\tilde{\theta}_m} .
\] 
Using this and Fact \ref{prop:half-error}, we get
\begin{equation}
	\label{eq:t1b}
\Phi(\tpo) - \Phi\paren{\tdo - 2a\sqrt{\tilde{\theta}_m}} \lesssim
2\paren{|\tpo|\tilde{\theta}_m + a \sqrt{\tilde{\theta}_m}}\Phi(\tpo) \sqrt{\log\frac{1}{\Phi(\tpo)}} .
\end{equation}

	Now we can bound
	\begin{align}
	& \Phi\paren{t_1 + 2 a \sqrt{\tilde{\theta}_m}} - \Phi\paren{\tdo - 2 a \sqrt{\tilde{\theta}_m}} \non \\
	&= \Phi\paren{t_1 + 2 a \sqrt{\tilde{\theta}_m}} - \Phi(t_1) + \Phi(t_1) - \Phi(\tpo) + \Phi(\tpo) - \Phi\paren{\tdo - 2a\sqrt{\tilde{\theta}_m}} \non \\
	&= \mu'_1 - \mu'_d + \Phi\paren{t_1 + 2 a \sqrt{\tilde{\theta}_m}} - \Phi(t_1) + \Phi(\tpo) - \Phi\paren{\tdo - 2a\sqrt{\tilde{\theta}_m}} 		\tag{Since $t_i = \Phi^{-1}(\mu'_i)$}  \non \\
	&\leq  2\nu_e + \left(\Phi\paren{t_1 + 2a\sqrt{\tilde{\theta}_m}} - \Phi(t_1)\right) + \left(\Phi(\tpo) - \Phi\paren{\tdo - 2a\sqrt{\tilde{\theta}_m}} \right) 
	\tag{Claim \ref{cl:mu-ell1}} \non \\
	&\lesssim \nu_e + \Phi(t_1) 2a\sqrt{\tilde{\theta}_m} \sqrt{\log\frac{1}{\Phi(t_1)}} 
	+ 2\paren{|\tpo|\tilde{\theta}_m + a \sqrt{\tilde{\theta}_m}}\Phi(\tpo) \sqrt{\log\frac{1}{\Phi(\tpo)}} \tag{using \eqref{eq:t1a}, \eqref{eq:t1b}} \non \\
	&= \nu_e + \Phi(t_1) 2a\sqrt{\tilde{\theta}_m} \sqrt{\log\frac{1}{\Phi(t_1)}} 
	+ \Phi(t_d) 2a \sqrt{\tilde{\theta}_m}\sqrt{\log\frac{1}{\Phi(\tpo)}}  + |\tpo|\tilde{\theta}_m \mu'_d \sqrt{\log\frac{1}{\mu'_d}}\non \\
	&\overset{1}{\lesssim} \nu_e + \mu'_1 a\sqrt{\tilde{\theta}_m} \sqrt{\log\frac{1}{\mu'_1}} 
	+ |\tpo|\tilde{\theta}_m \mu'_d \sqrt{\log\frac{1}{\mu'_d}} 		\non \\	
	&\overset{2}{\leq}  \nu_e + \mu'_1 a\sqrt{\tilde{\theta}_m} \sqrt{\log\frac{1}{\mu'_1}} 
	+ 2C^*_0\nu_e \log\frac{1}{\delta} \non \\
	&\lesssim \mu'_1 a\sqrt{\frac{\nu_e}{\mu'_1}} \sqrt{C^*_0 \log \frac{1}{\delta}} + C^*_0\nu_e \log\frac{1}{\delta} \tag{Lemma \ref{lem:ineq}(d),(f)} \non \\
	& \leq a\sqrt{C^*_0 \alpha_e\nu_e \log \frac{1}{\delta}} + C^*_0\nu_e \log\frac{1}{\delta}, 			\label{eqn:sec-term}			
	\end{align} 
	where the last step uses $\alpha_e = \mu'_1$ in our setting. Further, we can argue steps $1$ and $2$ as follows:
	\begin{itemize}
		\item For step $1$, we use that (i) $\Phi(t_d) = \mu'_d$ and (ii) the map $x \mapsto x\sqrt{\log(1/x)}$ is increasing in $x \to [0,1/2]$.
		\item Step $2$ can be obtained by combining the observations (i) $|\tpo| \leq 2\sqrt{2 C^*_0\log(1/\delta)}$ (Lemma \ref{lem:ineq} (f)) (ii) $\tilde{\theta}_m \leq 2\nu_e/\mu'_d$ (Lemma \ref{lem:ineq}(d)) (iii) $\mu'_d \geq \delta^{C^*_0}/2$ (Lemma \ref{lem:ineq}(f)) to get
		\[
		|\tpo|\tilde{\theta}_m \mu'_d \sqrt{\log\frac{1}{\mu'_d}}
		\leq \paren{2\sqrt{2 C^*_0\log(1/\delta)}} \cdot \left(\frac{2 \nu_e}{\mu'_d}\right) \cdot \mu'_d \cdot
		\sqrt{\log \frac{1}{\delta^{C^*_0}}}
		\lesssim  C^*_0 \nu_e \log\frac1\delta .
		\]
	\end{itemize}
	Therefore, plugging in the above bound we get that 
	\begin{align}
	&\int^{\sqrt{C^*_1 \log d}}_0 \Delta(a,t_1,\tdo) \varphi(a)da				\non \\
	&= \int^{\sqrt{C^*_1 \log d}}_0 \left(\Phi\paren{t_1 + \sqrt{\tilde{\theta}_m}a} - \Phi\paren{t'_d - \sqrt{\tilde{\theta}_m}a}\right) \varphi(a)da 		\non\\
	&\lesssim \int^{\sqrt{C^*_1 \log d}}_0 \left(a\sqrt{C^*_0 \alpha_e\nu_e \log \frac{1}{\delta}} + C^*_0\nu_e \log\frac{1}{\delta} \right) \varphi(a) da 	\tag{Using \eqref{eqn:sec-term}}\non\\
	&= \sqrt{C^*_0 \alpha_e\nu_e \log \frac{1}{\delta}} \int^{\sqrt{C^*_1 \log d}}_0 a \cdot \varphi(a) da 
		+ C^*_0\nu_e \log\frac{1}{\delta} \int^{\sqrt{C^*_1 \log d}}_0 \varphi(a) da 	\non\\
	&\leq \sqrt{C^*_0 \alpha_e\nu_e \log \frac{1}{\delta}}\cdot \Ex\left[\zeta_{\rm max}\right] + C^*_0\nu_e \log\frac{1}{\delta} \non \\ 
	&\lesssim \sqrt{C^*_0 \alpha_e\nu_e \log d \log \frac{1}{\delta}} + C^*_0\nu_e \log\frac{1}{\delta} \label{eqn:sec-bound}
	\end{align}
	where in the penultimate inequality, we bound the first integral by $\Ex[\zeta_{\rm max}]$ and the second integral can be bounded by $1$ since $\varphi(a)$ is the probability mass function of $\zeta_{\rm max}$. For the last step, we use the bound $\Ex\left[\zeta_{\rm max}\right] \lesssim \sqrt{\log d}$ from Lemma \ref{lem:zeta}(i).
	
	\paragraph{Finishing the proof.} 
	Combining the two bounds from \eqref{eqn:first-bound} and \eqref{eqn:sec-bound} and plugging them into \eqref{eqn:main}, we get that 
	\begin{align*}
	\Pr\Big[e \mbox{ is cut }\Big] 
	& = \int^{\infty}_{\sqrt{C^*_1 \log d}} \Delta(a,t_1,\tdo) \varphi(a)da + \int^{\sqrt{C^*_1 \log d}}_0 \Delta(a,t_1,\tdo) \varphi(a)da \\
	& \lesssim C^*_0 \nu_e \log(1/\delta) + C^*_0\sqrt{\alpha_e\nu_e} \sqrt{\log(d)\log\frac1\delta} \\
	& \lesssim C^*_0 \sqrt{\nu_e \mu'_1/A_{d,\delta}} \log(1/\delta) + C^*_0\sqrt{\alpha_e\nu_e} \sqrt{\log(d)\log\frac1\delta} \\
	& \lesssim  C^*_0\sqrt{\alpha_e\nu_e} \sqrt{\log(d)\log\frac1\delta},
	\end{align*}
	where the penultimate inequality is due $\nu_e \leq \mu'_1/A_{d,\delta}$ in the setting of the lemma, and the last inequality is using $\mu'_1  = \max(\mu'_1, 1- \mu'_1) = \alpha_e$.
\end{proof}
	
\subsection{When the biases are greater than $1/2$}

Using Lemma \ref{lem:cut-bound1} as a black box, we can prove an identical variant of the lemma for the case where all the biases are greater than $1/2$. We state and prove the lemma formally below.

\begin{restatable}{lem}{cutsym}			
	\label{lem:cut-bound-sym}
	Fix a hyperedge $e = (1,2,\ldots,d)$ and let $(\mu'_1,\ldots,\mu'_d)$ be the corresponding set of biases such that $1 \geq \mu'_d \geq \mu'_{d-1}\geq \cdots \geq \mu'_1\geq 1/2$ and $1 - \mu'_1 \geq \max(A_{d,\delta} \nu_e,\delta^{C^*_0})$, where $A_{d,\delta}:= 16C^*_0C^*_1 \log d \log (1/\delta)$. Then
	\[
	\Pr_{e \sim E}\left[e \mbox{ is cut}\right] \lesssim C^*_0\sqrt{\alpha_e\nu_e \log(1/\delta) \log (d)},
	\]
	where $C^*_1 = O(C^*_0\log(1/\delta))$.
\end{restatable}

\begin{proof}
	Let $g_1,\ldots,g_d$ be the ensemble of Gaussians associated with vertices $1,\ldots,d \in e$. Let us consider a new family of Gaussian random variables $\wh{\cG} = (\hat{g}_1,\ldots,\hat{g}_d)$ such that $\hat{g}_i := - g_i$ for every $i \in [d]$. Furthermore, for every $i \in [d]$, define $\hat{\mu}_i = 1 - \mu'_i$ and $\hat{t}_i = \Phi^{-1}(\hat{\mu}_i)$. Since $\phi(\cdot)$ is symmetric around the origin, we have $\hat{t}_i = - t_i$ for every $i \in [d]$. Consequently, this gives us a one-to-one correspondence between cut events with respect to $\cG$ and $\wh{\cG}$ respectively i.e.,
	\[
	\left\{\exists i,j \in e : g_i \leq t_i  \ \ \&  \ \  g_j > t_j \right\} \iff \left\{\exists i,j \in e : \hat{g}_i \leq \hat{t}_i  \ \ \&  \ \ \hat{g}_j > \hat{t}_j \right\} 
	\]
	Furthermore, we can conclude the following. 
	\begin{itemize}
		\item We have $0 \leq \hat{\mu}_d \leq \cdots \leq \hat{\mu}_1 \leq 1/2$ and $\hat{\mu}_1 \geq \max\{\delta^{C^*_0}, A_{d,\delta} \nu_e\}$.
		\item Furthermore $\max_{i \in e} \min( \hat{\mu}_i,1 -\hat{\mu}_i) = \max_{i \in e} \min(\mu'_i, 1 - \mu'_i) = \alpha_e$.
	\end{itemize}
	Therefore the Gaussians $(\hat{g}_i)_{i \in e}$, and the biases $(\hat{\mu}_i)_{i \in e}$ together satisfy all the requirements for Lemma \ref{lem:cut-bound1}. Hence, using the one-to-one correspondence established above and invoking Lemma \ref{lem:cut-bound1} we get that  
	\begin{eqnarray*}
		\Pr_{(g_i)_{i \in e}} \left[ \left\{\exists i,j \in e : g_i \leq t_i  \ \ \&  \ \ g_j > t_j \right\} \right]
		&=& \Pr_{(\hat{g}_i)_{i \in e}}\left[\left\{\exists i,j \in e : \hat{g}_i \leq t'_i \ \  \&  \ \ \hat{g}_j > \hat{t}_j \right\} \right] \\
		&\lesssim& C^*_0\sqrt{\alpha_e\nu_e \log(1/\delta) \log (d)}.
	\end{eqnarray*}
\end{proof}

\section{Concentration of Size of the Set}				\label{sec:conc}

The proofs in this sections frequently rely on some basic information theoretic quantities; we refer to the readers to Appendix \ref{sec:inf-theory} for a self-contained overview of the definitions and observations needed there.

Let us define the set of variables $y_1,\ldots,y_n$ as: 
\begin{enumerate}
	\item Sample a Gaussian $g \sim N(0,1)^r$ and for every $i \in [n]$ define $\zeta_i = \langle g, \bar{z'_i} \rangle$. 
	\item Define thresholds $t_i = \Phi^{-1}(\mu'_i)$.
	\item Set $y_i \defeq \mathbbm{1}(\zeta_i \leq \Phi^{-1}(\mu'_i))$ for every $i \in V$. 
\end{enumerate}
Note that the above steps are precisely the Gaussian rounding steps from Algorithm \ref{alg:approx-pred} (Line \ref{step:round}). In particular, the random variable $y_i = \mathbbm{1}_S(i)$ just indicates whether vertex $i$ is included in $S$. The main lemma of this section is the following lemma which gives concentration bounds for the size of $S$. 
\begin{lemma}
\label{lem:conc}
Setting $t = (1/\delta)^{100}$ and $\theta = \delta^{12}$ in Algorithm \ref{alg:approx-pred} we have: 
\[ 
\Pr_S\Big[ \left|\Ex_{i \sim V} \left[y_i\right] - \Ex_{i \sim V}[\mu'_i]\right| \geq \delta^2\Big] \leq \delta . 
\]
\end{lemma}
The above lemma was established for the setting $\theta = 0$, here we extend it to the setting $\theta := \delta^{O(1)}$. The proof of the above lemma mostly along the lines of Corollary 5.7~\cite{RT12arXiv} with minor changes to account for the fact that we are working with the perturbed vector solution. Before proving the above, we need some preparation. To begin with, the following lemma is an extension of Theorem 5.6 from \cite{RT12} which gives useful variance bounds. 
\begin{lemma}				\label{lem:conc1}
Let $x_1,\ldots,x_n$ be the variables corresponding to the deg-$2$ SoS solution constructed in Step \ref{step:condition} of Algorithm \ref{alg:approx-pred}.
Then,
\[ 
{\rm Var}\left(\Ex_{i \sim V}\left[y_i\right]\right) \leq O(1)\cdot \Ex_{i,j \sim V} \left[|I(x_i,x_j)\right]^{1/12} + O(\theta^{1/3}),
\]
where $I(\cdot,\cdot)$ denotes the mutual information between a pair of variables.
\end{lemma}
Towards proving Lemma \ref{lem:conc1}, we shall need the following lemma which is an adaptation of Lemma 5.6 from \cite{RT12}.
\begin{lemma}				\label{lem:corr-bound}
	For any $i,j \in V$, we have 
	\[
	I(y_i;y_j) \leq O(|\langle z'_i,z'_j \rangle|^{1/3} + \theta ) 
	\]
\end{lemma}
\begin{proof}
	Fix $i,j \in V$, and let $|\langle z'_i,z'_j \rangle| = c$. Without loss of generality, we can assume that $c \leq 1/2$ (otherwise the inequality is trivially true, since $I(y_i;y_j) \leq 1$ (Fact \ref{fact:I-bound})). Since $\langle z'_i,z'_j \rangle = \|z'_i\|\|z'_j\|\langle \bar{z'_i},\bar{z'_j} \rangle$,  we must have (i) $\|z'_i\| \leq c^{1/3}$ or (ii) $\|z'_j\| \leq c^{1/3}$ or (iii) $|\langle \bar{z'_i},\bar{z'_j} \rangle| \leq c^{1/3}$.
	Suppose (i) is true (the case (ii) can be handled similarly). Using \eqref{eqn:mu-w-defn}, 
	\[
	c^{2/3} \geq \|z'_i\|^2 = \frac{\|z_i\|^2 + \theta^2/4}{1 + \theta^2} \geq \frac{\|z_i\|^2}{2} = \frac{\mu_i(1 - \mu_i)}{2}.
	\]
	Therefore, $\min(\mu_i, 1 - \mu_i) \leq 4 c^{2/3}$. Without loss of generality we may assume that $\mu_i \leq 4 c^{2/3}$. Using Claim \ref{cl:mu-bounds}, we get that $\mu'_i \leq \mu_i + \theta^2 \leq 4c^{2/3} + \theta^2$. Since $\Ex\left[y_i\right] = \mu'_i$, it follows that
	\begin{align*}
	I(y_i;y_j) & \leq H(y_i) =  \mu'_i \log \frac{1}{\mu'_i} + (1 - \mu'_i) \log \frac{1}{1- \mu'_i} 
			\leq 2 \mu'_i \log \frac{1}{\mu'_i} 		\tag{Fact \ref{fact:h-bound}}\\
		& \leq 4 \sqrt{\mu'_i} \tag{using Fact \ref{fact:fact1}} \\
	& = O\paren{\sqrt{c^{2/3} + \theta^2}} = O \paren{ c^{1/3} + \theta} \tag{$\sqrt{a^2 + b^2} \leq \abs{a} + \abs{b}$} 
	\end{align*}
	Otherwise, we have $|\langle \bar{z'_i},\bar{z'_j} \rangle | \leq c^{1/3}$. However, note that the $\rho: = \langle \bar{z'_i}, \bar{z'_j} \rangle$ is exactly the correlation between gaussians $g_i $ and $g_j$ used for rounding. Note that covariance matrix $\Sigma$ of the joint distribution is exactly 
	\[
	\Sigma = 
	\begin{pmatrix}
	1 & \rho \\ \rho & 1
	\end{pmatrix} .
	\]
	Therefore, 
	\[
		I(y_i;y_j) \leq I(g_i;g_j) = -\frac12 \log {\rm det}(\Sigma) =  \log \frac{1}{\sqrt{1 - c^{2/3}}}  \leq O(c^{1/3})
	\]
	where the first step is due to the data processing inequality (Lemma \ref{lem:data-proc}) and the second step is due to Fact \ref{fact:gauss-mi}. For the last step, we use the assumption $c \leq 1/2$ and get
	\[
	\log \frac{1}{\sqrt{1 - c^{2/3}}} = \frac12\log\left(1 + \frac{c^{2/3}}{1 - c^{2/3}} \right) \leq \frac{0.5 c^{2/3}}{1 - c^{2/3}} \leq \frac{0.5 c^{2/3}}{1 - 2^{-2/3}}
	\]
	which completes the proof.
	
\end{proof}

\begin{fact}
	\label{fact:fact1}
	For $x \in (0,1]$, we have $x \log (1/x) \leq 2 \sqrt{x}$.
\end{fact}
\begin{proof}
First, we observe $f(y) \defeq 2\sqrt{y} - \log y \geq 0\ \forall y \geq 1$. This follows from
	$f(1) = 2 \geq 0$ and $f'(y) = 1/ \sqrt{y} - 1/y \geq 0\ \forall y \geq 1$. Therefore,
	$\forall x \in (0,1]$ we have $\log (1/x) \leq 2 \sqrt{1/x}$. Multiplying both sides by $x$
	finishes the proof.
\end{proof}

We shall also need the following.
\begin{fact}[Fact B.5 \cite{RT12}] 			
\label{fact:b5}
For any pair of $\{0,1\}$-valued random variables $X,Y$ we have ${\rm Cov}(X,Y) \leq \sqrt{I(X;Y)}$.
\end{fact}

\paragraph{Proof of Lemma \ref{lem:conc1}.}
The proof of Lemma \ref{lem:conc1} follows along the lines of the proof of Theorem 5.6 from \cite{RT12}.

\begin{proof}
Following the proof of Theorem 5.6 from \cite{RT12} we have 
\begin{align}
{\rm Var}\left(\Ex_{i \sim V} \left[y_i\right]\right) 
&= \Ex_{i,j \sim V} \left[{\rm Cov}(y_i,y_j)\right] 			\non\\
&\leq \Ex_{i,j \sim V} \left[\left(I(y_i;y_j)\right)^{1/2}\right]		 		\tag{Fact \ref{fact:b5}}	\non\\
&\leq O(1)\cdot\Ex_{i,j \sim V} \left[\left(|\langle z'_i,z'_j \rangle|^{1/3} + \theta\right)^{1/2}\right] 		\tag{Lemma \ref{lem:corr-bound}} \non \\
&\leq O(1) \cdot\Ex_{i,j \sim V} \left[\left(|\langle z'_i,z'_j \rangle|\right)^{1/6}\right] + O(\theta^{1/2}).		\label{eqn:w-bound0}
\end{align}
Using Claim \ref{cl:w-shift}, we have 
\begin{align}
\Ex_{i,j \sim V} \left[|\langle z'_i,z'_j \rangle|^{1/6}\right]	
&\leq \Ex_{i,j \sim V} \left[\left(|\langle z_i,z_j \rangle| + \theta^2/4\right)^{1/6}\right] \non \\
&\leq \Ex_{i,j \sim V}\left[\left(|\langle z_i,z_j \rangle|^{1/6} + \theta^{1/3}\right)\right] \non \tag{$(a + b)^{1/6} \leq \abs{a}^{1/6} + \abs{b}^{1/6}$} \\
&= \Ex_{i,j \sim V}\left[|\langle z_i,z_j \rangle|^{1/6}\right] + \theta^{1/3}. 	\label{eqn:w-bound2}
\end{align}
Now finally, note that the $\{z_i\}$-vectors come from {\em the original conditioned} degree-$2$ SoS solution and hence satisfy $\langle z_i,z_j \rangle = {\rm Cov}(x_i,x_j)$, where $x_i,x_j$ are the {\em local} variables of the pseudo-distribution (Proposition \ref{prop:lass}). Hence, 

\begin{align}				
	\Ex_{i,j \sim V}\left[|\langle z_i,z_j \rangle|^{1/6}\right] & = 
	\Ex_{i,j \sim V}\left[|{\rm Cov}(x_i,x_j)|^{1/6}\right] 
	\leq \Ex_{i,j \sim V}\left[|I(x_i;x_j)|^{1/12}\right] & \textrm{(Fact \ref{fact:b5})} \non \\
	& \leq \left(\Ex_{i,j \sim V}\left[|I(x_i;x_j)|\right]\right)^{1/12}. \qquad \qquad \qquad \textrm{(using Jensen's inequality)}\label{eqn:conc2}
\end{align}

Therefore, combining the above bounds we get that 
\begin{align*}
{\rm Var}\left(\Ex_{i \sim V}\left[y_i\right]\right) 
&\lesssim  \Ex_{i,j \sim V} \left[|\langle z'_i,z'_j \rangle|^{1/6}\right] + O(\theta^{1/2}) 		\tag{Using \eqref{eqn:w-bound0}}\\
&\leq \Ex_{i,j \sim V} \left[|\langle z_i,z_j \rangle|^{1/6}\right] +  O(\theta^{1/3}) + O(\theta^{1/2})		\tag{Using \eqref{eqn:w-bound2}} \\
&\leq \Ex_{i,j \sim V} \left[I(x_i;x_j)\right]^{1/12} +  O(\theta^{1/3}),				\tag{Using \eqref{eqn:conc2}}
\end{align*}
which finishes the proof of the lemma.
\end{proof}

\subsection{Proof of Lemma \ref{lem:conc}}

We shall need the following lemma from \cite{RT12} which says that conditioning on large number of rounds reduces the average mutual information with high probability.

\begin{lemma}[\cite{RT12}]
	There exists a constant $C \geq 1$ such that the following holds. Fix $\eta,\delta \in (0,1)$, and let $t = 1/\eta\delta$. Then with probability $1 - \delta$ over the choice of random conditionings of at most $t$-sized subsets $X_T \gets \alpha$ from Step \ref{step:condition} we have
	\[
	\Ex_{i,j \sim V} \left[I(x_i;x_j)\right] \leq C\eta.
	\]
\end{lemma}
Using our choice of $\eta$ and invoking the above lemma, we get that with probability at least $1 - \delta$, we have that $\Ex_{i,j \sim V} \left[I(x_i;x_j)\right]^{1/12} \leq \delta^4/2$. Furthermore, using our choice of $\theta = (\delta/2)^{12}$ and invoking Lemma \ref{lem:conc1} we get that ${\rm Var}(|S|) \leq \delta^4$. Furthermore, we have $\Ex_{i \sim V} y_i = \delta$. Therefore, using Chebyshev's inequality, 
	\[
\Pr_{y_1,\ldots,y_n}\left[ \Big|\Ex_{i \sim V}\left[y_i\right] - \Ex_{i \sim V} \left[\mu'_i\right] \Big| > \delta^2\right] \leq \frac{\delta^5}{\delta^4} \leq \delta.
	\]

\section{Proof of Theorem \ref{thm:hsse-main}}				\label{sec:ssve-main}

In this section, we finally prove Theorem \ref{thm:hsse-main}. First we combine the Gaussian rounding Lemmas (Lemma \ref{lem:cut-bound1}, Corollary \ref{lem:cut-bound-sym}) to give a lemma which bounds the probability of any edge getting cut (Lemma \ref{lem:edge-cut}). Then we use this in Section \ref{sec:proof} to bound the expansion guarantee of the sets rounded by Algorithm \ref{alg:approx-pred} to complete the proof of Theorem \ref{thm:hsse-main}. 

\begin{lemma}				\label{lem:edge-cut}
	Fix a hyperedge $e \in E$.
	\[ \Pr\left[{e \textrm{ is cut}} \right]
	\lesssim \sqrt{\alpha_e\nu_e\log d \log(1/\delta)} + \tilde{O}(d)\cdot \nu_e(\log(1/\delta))^2  \]
\end{lemma}

Before we begin the proof of the lemma, we remark that since $C^*_0$ is fixed to be an absolute constant independent of all other parameters, in this proof we shall often treat $C^*_0$ as a constant and absorb it into the $O(\cdot)$ notation. 

\begin{proof}
	By reordering, we may assume $e = [d]$. We first prove the statement for the case when all the biases of the vertices in the hyperedge lie in $[\delta^{C^*_0},1/2]$, following which we extend the proof to the setting of arbitrary biases.
	
	\paragraph{Bounding the cut value when $\mu'_i \leq 1/2\ \forall i \in e$.} 
	In this case, $\min (\mu'_i,1 - \mu'_i) = \mu'_i$ for each $i \in e$. 
	Therefore, $\alpha_e$ defined as $\alpha_e \defeq \max_{i \in e} \min(\mu'_i,1 - \mu'_i)$ satisfies $\max_{i \in e} \mu'_i = \alpha_e $.
	Furthermore, we will assume the ordering satisfies $\mu'_d \leq \mu'_{d - 1} \leq \cdots \leq \mu'_1$. Now recall that 
	\begin{equation}					\label{eqn:cut-condn}
	\nu_e =\max_{i,j \in e}  \|v'_i - v'_j\|^2  \qquad \textnormal{and} \qquad \max_{i,j \in e} |\mu'_i - \mu'_j| \leq 2\nu_e . 
	\end{equation}
	where the inequality is due to item (a) of Lemma \ref{lem:ineq}. Next, we bound the probability of a hyperedge being cut in the rounded solution using two cases. Recall $A_{d,\delta} = 16 C^*_0 C^*_1\log(1/\delta) \log(d)$, where $C^*_1 \lesssim C^*_0 \log(1/\delta)$. 
	
	{\bf Case 1.1}: Suppose $\alpha_e \leq A_{d,\delta} \nu_e$. Here, by a union bound,
	\begin{align*}
		\Pr_{(g_i)_{i \in e}}\Big[e \mbox{ is cut }\Big] 
		&\leq \Pr_{(g_i)_{i \in e}}\left[\exists i \in e : y(i) \neq 0\right] \\
		&\leq \sum_{i \in e} \mu'_i \leq d A_{d,\delta} \nu_e& \textrm{(Using $\max_{i \in e} \mu'_i = \alpha_e $)} \\
		&\leq \tilde{O}(d) \cdot C^*_1C^*_0 \nu_e\log(1/\delta) \\
		&\leq \tilde{O}(d) \cdot (C^*_0)^2 \nu_e \log^2 (1/\delta) . 
	\end{align*}
	
	{\bf Case 1.2} Suppose $\alpha_e \geq A_{d,\delta} \nu_e$. Note that from our choice of $\theta$ in Step \ref{step:vi'} in Algorithm \ref{alg:approx-pred}, and using Claim \ref{cl:mu-bounds} we have $\mu'_1 \geq \delta^{C^*_0}$. There our setting of parameters here satisfies the premise of Lemma \ref{lem:cut-bound1}. Hence invoking Lemma \ref{lem:cut-bound1} we get that 
	\begin{align*}
		\Pr_{(g_i)_{i \in e}}\Big[e \mbox{ is cut }\Big] 
		&\lesssim C^*_0\sqrt{\alpha_e\nu_e \log(1/\delta) \log (d)}.
	\end{align*} 
	Combining the bounds from the two cases gives us the bound here.
	
	\paragraph{Bounding the cut value when $\mu'_i \geq 1/2 \ \forall i \in e$.}
	We claim that this case is symmetric to the above case and and an identical bound can be recovered using similar arguments. By reordering, we may assume that $1/2 \leq \mu'_1 \leq \mu'_2 \leq \cdots \leq \mu'_d$ and using the guarantee of Claim \ref{cl:mu-bounds} and our choice of $\theta$, we have $\mu'_d \leq 1 - \delta^{C^*_0}$. Furthermore the guarantees of \eqref{eqn:cut-condn} hold as is. Here we shall again consider two cases depending on the value of $\alpha_e = 1 - \mu'_1$. 
	
	{\bf Case 2.1}: Suppose $\alpha_e \leq A_{d,\delta} \nu_e$. Here, by a union bound,
	\begin{align*}
		\Pr_{(g_i)_{i \in e}}\Big[e \mbox{ is cut }\Big] 
		&\leq \Pr_{(g_i)_{i \in e}}\left[\exists i \in e : y(i) \neq 1\right] \\
		&\leq \sum_{i \in e}(1 - \mu'_i) \leq d A_{d,\delta} \nu_e & \textrm{(Using $\max_{i \in e} 1 - \mu'_i = \alpha_e $)} \\
		&\leq \tilde{O}(d) \cdot C^*_1C^*_0 \nu_e\log(1/\delta) \\
		&\leq \tilde{O}(d) \cdot (C^*_0)^2 \nu_e \log^2 (1/\delta) . 
	\end{align*}

	{\bf Case 2.2} Suppose $\alpha_e \geq A_{d,\delta} \nu_e$. Similar to Case 1.2, we observe that the setting of parameters here satisfy the conditions required for Lemma \ref{lem:cut-bound-sym}. Therefore, as before, we can invoke Lemma \ref{lem:cut-bound-sym} to bound
	\begin{align*}
		\Pr_{(g_i)_{i \in e}}\Big[e \mbox{ is cut }\Big] 
		&\lesssim C^*_0\sqrt{\alpha_e\nu_e \log(1/\delta) \log (d)}.
	\end{align*} 
	Hence, combining the two cases gives us 
	\begin{align*}
		\Pr_{(g_i)_{i \in e}}\Big[e \mbox{ is cut }\Big]  
			&\leq C^*_0\sqrt{\alpha_e\nu_e \log(1/\delta) \log (d)} + \tilde{O}(d)\cdot (C^*_0)^2 \nu_e \left(\log(1/\delta)\right)^{2}.
	\end{align*} 
		
	\paragraph{Extension to arbitrary biases.} Now let $e \in E$ be a hyperedge with bias values that crossover $1/2$. We partition the hyperedge as $e = e^+ \sqcup e^{-}$, where $e^+ := \left\{i \in e| \mu'_i \leq 1/2 \right\}$ and $e^{-} = \left\{i \in e| \mu'_i > 1/2 \right\}$. Then note that if $e^- = \emptyset$ or  $e^{+} = \emptyset$, then the analysis from cases 1 and 2 gives us the bound. Therefore, we consider the case $e^+ \neq \emptyset$ and $e^{-} \neq \emptyset$. Recall that from the edge deletion step (Step \ref{step:e-del} of Algorithm \ref{alg:approx-pred}), we have $\nu_e \leq 0.1$ for every surviving edge hyperedge $e$. Note that since the biases cross over $1/2$, we must have $\min_{i \in e} \mu'_i \geq 1/2 - 2\nu_e \geq 0.3$ and $\max_{i \in e} \mu'_i \leq 1/2 + 2\nu_e \leq 0.7$, and therefore in particular, $\alpha_e \geq 0.3$.
	Fix arbitrary vertices $i^+ \in e^{+}$ and $i^- \in e^{-}$. Then, we can bound the probability of a hyperedge $e$ getting cut as 
	\begin{align*}
		\Pr\Big[e \textrm{ is cut}\Big] 
		&\leq \Pr\Big[\left\{e^+ \mbox{ is cut} \right\} \vee \left\{e^{-} \mbox{ is cut} \right\}
		\vee \{(i^+,i^-) \mbox{ is cut} \} \Big] \\
		&\leq \Pr\Big[e^+ \mbox{ is cut} \Big] + \Pr\Big[e^{-} \mbox{ is cut} \Big]
		+ \Pr\left[ (i^+,i^-) \mbox{ is cut} \right] \\
		&\overset{1}{\lesssim} \left( \sqrt{\alpha_{e^+}\nu_{e} (\log d \log(1/\delta)) } + \sqrt{\alpha_{e^-}\nu_{e} (\log d \log(1/\delta)) }  \right) + \tilde{O}(d)\cdot \nu_e (\log (1/\delta))^{2}  +  \sqrt{\nu_e} \\
		&\overset{2}{\lesssim} \sqrt{\alpha_e\nu_e \log(1/\delta) \log (d)} + \tilde{O}(d)\cdot \nu_e (\log (1/\delta))^{2} + O(\sqrt{\alpha_e \nu_e})  
	\end{align*}
	Here step $1$ can be argued as follows. 
	\begin{itemize}
		\item[a.] For the first term, note that the hyperedge $e^+$ consists of vertices with biases within $[\delta^{C^*_0},1/2]$. Furthermore, $\nu_{e^{+}} = \max_{i,j \in e^+} \|v'_i - v'_j\|^2 \leq \max_{i,j \in e} \|v'_i - v'_j \|^2 \leq \nu_e $. Therefore, we can use case 1  to bound the probability of $e^+$ getting cut.
	\item[b.] For the second term, note that $e^-$ consists of vertices with biases within $[1/2,1 - \delta^{C^*_0}]$. Similar to the previous case, we get $\nu_{e^{-}} \leq \nu_e $. Therefore, we can use case 2 to bound the probability of $e^-$ getting cut.
		\item [c.] The third term constitutes the analysis of an edge getting cut for which we use Lemma A.5 from \cite{RT12arXiv}. Additionally we observe that $\sqrt{\nu_e} \lesssim \sqrt{\nu_e \alpha_e}$ since $\alpha_e \geq 0.4$ from the above discussion.
	\end{itemize}
	Step $2$ follows from the observations that
	\[
	\max\left(\alpha_{e^+},\alpha_{e^-} \right) \leq \frac12 \leq \frac{5\alpha_e}{4} \ \ \ \ \textnormal{ and } \ \ \ \ \max(\nu_{e^+}, \nu_{e^-}) \leq \nu_e .
	\]
\end{proof}

\subsection{Proof of Theorem \ref{thm:hsse-main}}			\label{sec:proof}

We begin by recalling some notation we will use in this proof. Recall that $H = (V,E,w,W)$ is the weighted hypergraph constructed using Lemma \ref{lem:transform}, where $w: E \to \mathbbm{R}_{\geq 0}$ is the edge weight function and $W: V \to \mathbbm{R}_{\geq 0}$ is the vertex weight function. The proof of the theorem consists of two parts: first we will establish that the weight of the subset $S$ rounded by Algorithm \ref{alg:approx-pred} in Step \ref{step:round} is $\Theta(\delta\cdot W(V))$. Conditioned on this, we will show that the subset $S$ has the intended hypergraph expansion guarantee. Finally using Lemma \ref{lem:transform} (vi), we can conclude the set $S$ can be used to identify a set $S' \subset V_G$ in $G$ of size $\sim \delta |V_G|$ with matching (up to constant factors) vertex expansion guarantee. 

\paragraph{Bounding vertex weight $W(S)$}
 
Recall that $S \subset V$ is the set rounded off in Step \ref{step:round} of Algorithm \ref{alg:approx-pred}. Let $i \sim V$ denote the random draw of a vertex according to vertex weight function $W$. To begin with, observe that the conditioned degree-$2$ pseudo-distribution $\mu$ obtained in Step \ref{step:condition} must satisfy:
\[
\Ex_{i \sim V} \big[\mu_i \big] = \E_{i \sim V}\Pr_{X_i \sim \mu}\left[X_i = 1\right]  = \delta.
\]
Furthermore, using Lemma \ref{cl:mu-bounds}, we can relate the analogous expression on the shifted biases to the above as 
\[
\left| \Ex_{i \sim V} \big[\mu'_i \big] - \Ex_{i \sim V} \big[\mu_i \big] \right| \leq \delta^{10}
\]
For a subset $T \subseteq V$, let $\delta_{\rm rel}(T)$ denote the {\em relative} weight of the set with respect to vertex weight function $W$.  We have the following bounds on the expected relative vertex weight of the set rounded off by the algorithm in Line \ref{step:round}:
\begin{align}
\Ex_S\left[\delta_{\rm rel}(S)\right] := \Ex_{S} \left[\frac{W(S)}{W(V)}\right] 
= \Ex_{S} \Ex_{i \sim V} \left[\mathbbm{1}_S(i)\right] 		
= \Ex_{i \sim V} \left[\mu'_i\right] 		
\in \delta(1 \pm o_\delta(1)). 	\label{eqn:s-step1}
\end{align}
On the other hand, from Lemma \ref{lem:conc1} we have 
\begin{equation}			\label{eqn:s-step2}
\Pr_S\Big[\delta_{\rm rel}(S) \notin [0.99,1.01] \Ex_S \left[\delta_{\rm rel}(S)\right]\Big] \leq 0.1 .
\end{equation}
Therefore combining the bounds from \eqref{eqn:s-step1} and \eqref{eqn:s-step2} we have 
\begin{equation}			\label{eqn:s-step3}
\Pr_S\left[\frac{W(S)}{W(V)} \notin [0.9,1.1] \delta \right] \leq 0.1.
\end{equation}

\paragraph{Bounding the expansion.}

To begin with, recall that from Lemma \ref{lem:alg-feas}  
\[
 \frac{1}{\delta W(V)} \sum_{e \in E} w(e)\max_{i,j \in e}\Pr_{X_{ij} \sim \mu^*}\left[X_i \neq X_j\right] \leq \phiv_\delta,
\]
Denoting $\Delta_e = \max_{i,j \in e}\Pr_{X_{ij} \sim \mu}\left[X_i \neq X_j\right]$, with probability at least $0.9$ over the choice of $\mu$ obtained by the random conditioning $X_t \gets \alpha$, using Markov's inequality we have 
\begin{equation}			\label{eqn:obj-1}
\frac{1}{\delta W(V)}\sum_{e \in E} w(e)\Delta_e \leq  10 \cdot\phiv_\delta,	
\end{equation}
Further, conditioning on the above, we observe that 
\begin{align}
\sum_{e \in E} w(e)\nu_e &= \sum_{e \in E} \max_{i,j \in e} \|v'_i - v'_j\|^2 \\
&\leq \sum_{e \in E} w(e)\max_{i,j \in e} \|v_i - v_j\|^2 		\tag{Claim \ref{cl:vprime-equal}}	\non\\ 
&= 4 \sum_{e \in E} w(e) \Delta_e			\tag{Proposition \ref{prop:lass}}	\non \\
&\leq 40 \phiv_\delta \left(\delta W(V)\right),	\label{eqn:obj-2}	
\end{align}
where the last step is using \eqref{eqn:obj-1}. Next, we account for the weight of edges removed in the deletion step.

\begin{claim}			\label{cl:del-cost}
	The deletion step (Step \ref{step:e-del}) removes at most $10 \sum_{e \in E} w(e) \nu_e$ weight of edges.
\end{claim}
\begin{proof}
We can bound the weight of edges deleted in Step \ref{step:e-del} as follows.
\[
\sum_{e \in E} w(e) \nu_e \geq \sum_{\substack{e \in E \\ \nu_e \geq 0.1}} w(e) \nu_e \geq 0.1 \sum_{\substack{e \in E \\ \nu_e \geq 0.1}}w(e)
\]
which on rearranging gives us that the weight of edges is deleted is at most $ 10\sum_e \nu_e w(e)$.
\end{proof}

Finally, the following claim accounts for the cost of edges cut from the rounding step.

\begin{claim}
For every edge $e \in E$ such that $\nu_e \leq 0.1$ we have 
\[
\Pr\left[e \textrm{ is cut} \right] \lesssim \sqrt{\mu'_{v(e)}\nu_e\log d \log(1/\delta)} + \tilde{O}(d)\cdot \nu_e(\log(1/\delta))^{2}
\]
where recall that $v(e) := \pi(e) \in e$ is the unique vertex identified by edge $e$ from Lemma \ref{lem:transform} (v).
\end{claim}
\begin{proof}
From Lemma \ref{lem:transform} (v), recall that every hyperedge $e \in E$ is uniquely identified with a vertex $v(e) \in V $ such that $v(e) \in e$. We shall need the following observation.

\begin{observation}				\label{obs:alpha-bd}
	For every hyperedge $e$, we have $\alpha_e \leq \mu'_{v(e)} + 2\nu_e$.
\end{observation}
\begin{proof}
	Fix a hyperedge $e$, and let $i_0 \in e$ be such that $\alpha_e = \min(\mu'_{i_0},1 - \mu'_{i_0})$. Hence,
	\[
	\alpha_e = \min(\mu'_{i_0}~,1 - \mu'_{i_0}) \leq \mu'_{i_0} \leq \mu'_{v(e)} + 2\nu_e
	\] 
	where the last inequality follows from Claim \ref{cl:mu-ell1}.
\end{proof}
Combining the above with Lemma \ref{lem:edge-cut}, we can establish
\begin{eqnarray}
\Pr\left[e \mbox{ is cut }\right] 
&\leq& \sqrt{\alpha_e\nu_e\log d \log(1/\delta)} + \tilde{O}(d)\cdot \nu_e(\log(1/\delta))^2  	\label{eqn:step-1} \non \\
&\overset{1}{\leq}& \sqrt{(\mu'_{v(e)} + 2\nu_e)\nu_e\log d \log(1/\delta)} + \tilde{O}(d)\cdot \nu_e(\log(1/\delta))^{2}    \non \\
&\overset{2}{\leq}& \sqrt{\mu'_{v(e)}\nu_e\log d \log(1/\delta)} + 2\nu_e\sqrt{\log d \log(1/\delta)} +
\tilde{O}(d)\cdot \nu_e \log(1/\delta)^{2}   \nonumber\\
&\overset{3}{\leq}& \sqrt{\mu'_{v(e)}\nu_e\log d \log(1/\delta)} + \tilde{O}(d)\cdot \nu_e(\log(1/\delta))^{2}	\label{eqn:step-2}
\end{eqnarray}
where step $1$ uses Observation \ref{obs:alpha-bd}, step 2 uses $\sqrt{a + b} \leq \sqrt{a} + \sqrt{b}$ for nonnegative $a,b$, and step $3$ follows by observing that the second term in the previous expression is dominated by the third term. 
\end{proof}
{\bf Cleaning up}. Note that with probability at least $0.8$, we have $W(S) \in [0.9\delta n,1.1 \delta n]$ and \eqref{eqn:obj-2} holds. Using this we can derive the following bound on the weighted sum of shifted biases: 
\begin{align}
\sum_{e \in E}w(e) \mu'_{v(e)} & = \sum_{e \in E} W(v(e)) \mu'_{v(e)} \tag{$w({v(e)}) = W(v(e))$ by construction of $H$} \non \\
& \overset{1}{\leq} \sum_{e \in E} W(v(e)) \mu_{v(e)} + \delta^{C^*_0} \sum_{e \in E} W(v(e))  \non \\
&= \sum_{i \in V} W(i) \mu_{i} + \delta^{C^*_0} \sum_{i \in V} W(i) \tag{one-to-one correspondence between $V$ and $E$.} \non \\
&\overset{2}{\leq} \delta n + \delta^{C^*_0} n \leq 2\delta n,			\label{eqn:weight-bound}
\end{align}
where in step $1$ above, we apply the first point of Claim \ref{cl:mu-bounds} point wise to bound every $\mu'_{v(e)}$. In step $2$, we bound the first term using the feasibility of $\mu_i$'s from the SDP solution, and the second term is bounded using the weight bound from Lemma \ref{lem:transform}. Now we are ready to bound the hyperedge expansion of $S$ in $H$ as:  
\begin{align*}
\phie_H(S) \lesssim
&\frac{ \sum_{e \in E}w(e) \paren{\sqrt{\mu'_{v(e)}\nu_e\log d \log(1/\delta)} + \tilde{O}(d)\cdot \nu_e(\log(1/\delta))^{2}}}{W(S)}  \\
& \lesssim \sum_{e \in E}w(e)\frac{\sqrt{\mu'_{v(e)}\nu_e\log d \log(1/\delta)}}{\delta n} + \frac{\sum_{e \in E}\tilde{O}(d)\cdot w(e)\nu_e(\log(1/\delta))^{2}}{\delta n}  	\tag{Since $W(S) \geq 0.9 \delta n$} \\
& \leq \sqrt{\frac{\sum_{e \in E}w(e) \mu'_{v(e)}}{\delta n}}\sqrt{\frac{\sum_{e \in E}w(e)\nu_e\log d \log(1/\delta) }{\delta n}} + \tilde{O}(d)\cdot \phiv_\delta (\log(1/\delta))^{2}  		\tag{Using \eqref{eqn:obj-2}}  \\
& \lesssim \sqrt{\phiv_\delta \log d \log(1/\delta)} + \tilde{O}(d)\cdot \phiv_\delta (\log(1/\delta))^2.  
\end{align*}
where in the last step we bound the first term using the bound from \eqref{eqn:weight-bound} and the second term using \eqref{eqn:obj-2}. 

{\bf Putting things together}. The above arguments imply that with probability at least $0.9$, the set $S \subset V$ constructed in line \ref{step:round} satisfies $|S|\in (1 \pm o_\delta(1))\cdot|V_G|$ and 
\[
\phi^E_H(S) \lesssim \sqrt{\phiv_\delta \log d \log(1/\delta)} + \tilde{O}(d)\cdot \phiv_\delta (\log(1/\delta))^2.
\]
Then in final step of the Algorithm \ref{alg:approx-pred} (Line \ref{step:final}), using item (vi) of Lemma \ref{lem:transform}, the algorithm can find a subset $S' \subseteq V_G$ such that $|S'| \in [0.8,1.2] \delta n$ and 
\[
\phiv_G(S') \leq \phie_H(S) \lesssim \sqrt{\phiv_\delta \log d \log(1/\delta)} + \tilde{O}(d)\cdot \phiv_\delta (\log(1/\delta))^2,
\]
which completes the proof of Theorem \ref{thm:hsse-main}.

\section{$f(d)$-inapproximability for \ssve}
\label{sec:deg-red}
We prove the following SSEH based hardness result for \ssve~which rules out any $f(d)$-approximation for \ssve.

\ssedeg*

Our starting point is the following alternative but equivalent form of the SSEH.
\begin{conjecture}[\cite{RS10,RST12}]	
\label{conj:sseh}
	Given an $\epsilon \in (0,1), M \geq 1$, there exists $\delta \in (0,1)$ such that the following holds. Given a regular graph $G = (V,E)$, it is NP-Hard to distinguish between the following two cases.
	\begin{itemize}
		\item {\bf YES}: There exists $S \subseteq V$ such that $\mu(S) = \delta$ and $\phi(S) \leq \epsilon$.
		\item {\bf NO}: For every $S \subseteq V$ such that $\mu(S) \in [\delta/M,M\delta]$, we have $\phi(S) \geq 1- \epsilon$.
	\end{itemize}
\end{conjecture}

For a regular graph $G=(V,E)$ and $S \subseteq V$, we use $\mu_G(S)$ to denote $|S|/|V|$.
We first describe a degree reduction procedure which approximately maintains the completeness and soundness guarantees of SSEH. 
\begin{lemma}[Folklore]			\label{lem:deg-red}
There exists constants $g \in \mathbbm{N}$ and $\alpha_0 \in (0,1)$ and a polynomial time reduction from any regular $(\epsilon,1-\epsilon)$-instance with parameters $M$ and $\delta$ (as defined in Conjecture \ref{conj:sseh}) to a $(g+1)$-regular instance $G'=(V',E')$ with the following properties.
	\begin{itemize}
		\item {\bf Completeness}. If $G$ is a YES instance, then there exists $S \subseteq V'$ such that $\mu_{G'}(S) = \delta$ and $\phi_{G'}(S) \leq \epsilon/(g+1)$.
		\item {\bf Soundness}. If $G$ is a NO instance, then for every $S \subseteq V'$ such that $\mu_{G'}(S) = \delta$, we have $\phi_{G'}(S) \geq \alpha_0$.
	\end{itemize} 
\end{lemma}
The reduction for the above lemma is folklore and uses the {\em replacement product trick} introduced in \cite{PY91}; for e.g., \cite{RT13} also use an identical construction for degree reduction for a different range of parameters. We include a proof of it here for our setting for completeness.

\begin{proof}

{\bf Construction}. Given a $d$-regular graph $G = (V,E)$, we construct a graph $G' = (V',E')$ as follows. Let $H = (V_H,E_H)$ be a $g$-regular $\alpha$-expander on $d$-vertices. For every $i \in V$, we replace it with a copy of $H$ in $G'$, which we call $H_i$ on vertex set $\cC_i$. (Note that by construction $|\cC_v| = d$). For every $i \in V$, fix an arbitrary ordering $i(1),i(2),\ldots,i(d)$ on the vertices in $\cC_i$. For every edge $(i,j) \in E$, we identify a {\em unique} vertex $i(a) \in \cC_i$ and $j(b) \in \cC_j$ and put an edge between $(i(a),j(b))$ in $G'$. Note that this can be done consistently such that in the resultant $G'$, every vertex has degree $g+1$.

We fix $M= 2$ and argue completeness and soundness of the reduction.

{\bf Completeness}. Suppose there exists a set $S \subseteq V$ such that $\mu_G(S) = \delta$ and $\phi(S) \leq \epsilon$. Then consider the corresponding set $S' = \cup_{i \in S} \cC_i$. It is easy to see that $\mu_{G'}(S) = \delta$ and $\phi_{G'}(S) = |\partial_G(S)|/((g+1)(d \delta n)) \leq \epsilon/(g+1)$.

{\bf Soundness}. Suppose there exists a set $S' \subset V'$ of measure $\delta$ such that $\phi_{G'}(S) \leq \epsilon'/(g+1)$ where $\epsilon' = \alpha/100$. Then $\left|\partial_{G'}(S') \right| \leq \epsilon' d \delta n$.

For every $i \in V$, let $\cC'_i \defeq \cC_i \cap S'$ and $\gamma_i \defeq |\cC'_i|/|\cC_i| = |\cC'_i|/d$. Define 
\[
T_{\rm Good} \defeq \left\{i \in V | \gamma_i \geq 1/2\right\} \qquad \textnormal{and} \qquad T_{\rm bad} \defeq \left\{i \in V | \gamma_i < 1/2 \right\} . \]
Note that by definition of $H$ we have 
\begin{align*}
\epsilon' d \delta n \geq \left|\partial_{G'}(S')\right| 
&\geq \sum_{i \in T_{\rm Bad}} |E_{G'}(\cC'_i ,\cC_i \setminus \cC'_i)| \\
&\geq \alpha \sum_{i \in {T_{\rm bad}} } \min\{\gamma_i,1 - \gamma_i\}\cdot g\cdot d 	\tag{Expansion of $H_i$'s}\\
&\geq \alpha \sum_{i \in {T_{\rm bad}} } \gamma_i  g d 					\tag{Definition of $T_{\rm bad}$}
\end{align*}
or rearranging we get that
\begin{equation}				\label{eqn:bad-set-bd}
\sum_{i \in T_{\rm bad}} \mu_{G'}(\cC'_i) = \frac{1}{n} \sum_{i \in {T_{\rm bad}} } \gamma_i 
	\leq \frac{1}{n} \cdot \frac{\epsilon' d \delta n }{\alpha g d} =  \frac{\epsilon' \delta }{\alpha g} .
\end{equation}
	Therefore we have $\sum_{i \in T_{\rm good}} \mu(\cC'_i) \geq \delta - \epsilon' \delta/(\alpha g) \geq 0.99 \delta$. Now note that every cloud $\cC'_i$ corresponding to each $i \in T_{\rm good}$ can contribute at most $1/n$ mass (recall that $T_{\rm good} \subseteq V$), and by definition of $T_{\rm good}$, must contribute at least $1/(2n)$ mass. Therefore, using the upper and lower bound we get that $|T_{\rm good}| \in[ 0.99  \delta n,2\delta n]$. Now using the fact that $\cC_i$'s are disjoint, we note that 
\begin{align}					
\sum_{\substack{i,j \in T_{\rm good} \\ i \neq j}}\left|E_{G'}(\cC'_i,\cC'_j)\right| &\leq \sum_{\substack{i,j \in T_{\rm good} \\ i \neq j}} \left|E_{G}(i,j)\right|
	= {\sf vol}_G(T_{\rm good}) - \phi_G(T_{\rm good}){\sf vol}_G(T_{\rm good}) \non \\
	&\leq {\sf vol}_G(T_{\rm good}) - (1 - \epsilon){\sf vol}_G(T_{\rm good}) 
	= \epsilon {\sf vol}_G(T_{\rm good})  \non \\ 
	& \leq 2\epsilon d \delta n \label{eqn:in-set-bd}
\end{align}
	where in the second inequality we use that $\phi_G(A) \geq (1 - \epsilon)$ for every $A$ such that ${\sf vol}_G(A) \in [\delta dn/2,2\delta dn]$. 

	Note that every $i(a) \in \cup_{j \in T_{\rm good}} \cC'_{j}$ has exactly one $E$-edge incident on it. Therefore the total number of $E$-edge incident on $\cup_{j \in T_{\rm good}} \cC'_{j}$  is at least 
\[ |\cup_{j \in T_{\rm good}} \cC'_j| = \sum_{j \in T_{\rm good}} |\cC'_j| 
	= \sum_{j \in T_{\rm good}} \gamma_i d 
	\geq |T_{\rm good}|d/2 \geq \delta n d/4 \] 
	which uses $\gamma_i \geq 1/2\ \forall i \in T_{\rm good}$.  

Out of these, at most $(\epsilon'/\alpha g) \delta d n$ are incident on $T_{\rm bad}$ (from \eqref{eqn:bad-set-bd}) and at most $2 \epsilon d \delta n$ are incident on itself (using \eqref{eqn:in-set-bd}). Therefore, at least 
\[
	d \delta n/4 - (\epsilon'/\alpha g) \delta d n - 2 \epsilon d \delta n \geq d \delta n/8
\]
number of $E$-edges leave the set $S'$ from $\cup_{i \in T_{\rm good}} \cC'_i$. Therefore, $|\partial_{G'}(S')| \geq d\delta n/8$, which gives us a contradiction.
\end{proof}

\begin{fact}			\label{fact:exp-bounds}
	For any graph $G = (V,E)$ with maximum degree $d$ we have 
	\[
	 \phi^V_\delta(G) \leq \phi^E_\delta(G) \leq d\cdot \phi^V_{\delta}(G). 
 	\]
\end{fact}

We now prove Proposition \ref{prop:deg-hardness}.
\begin{proof}[Proof of Proposition \ref{prop:deg-hardness}]
	Fixing $\epsilon$, let $G = (V,E)$ be an instance of \smallsetexpansion~as guaranteed by Conjecture \ref{conj:sseh}. Let $G'= (V',E')$ be the graph obtained by instantiating Lemma \ref{lem:deg-red} on $G = (V,E)$. Note that from the guarantee of Lemma \ref{lem:deg-red}, we have that the maximum degree of $G'$ is $g+1$. Denote $\epsilon' := \epsilon/(g+1)$. Then, if $G$ is a YES instance, using Proposition \ref{fact:exp-bounds} we have 
	\[
	\phi^{V'}_\delta(G') \leq \phi^{E'}_\delta(G') \leq \frac{\epsilon}{g+1}.
	\] 
	On the other hand, if $G$ is a NO instance, then using the lower bound from Proposition \ref{fact:exp-bounds} we get that
	\[
	\phi^{V'}_\delta(G') \geq \frac{\phi^{E'}_{\delta}(G')}{g+1} \geq \frac{1}{8g^2}.
	\]
	Putting the two case together, and setting $\epsilon < f^{-1}(8g)$ establishes the claim.
\end{proof}

\section{Proof of Theorem \ref{thm:hsse-hyper}}

The algorithm for Theorem \ref{thm:hsse-hyper} is almost identical to that of Theorem \ref{thm:hsse-main}; the only change required is that the cardinality/expected weight constraint now compares with the total weight of edges, in order to accommodate the definition of hypergraph expansion. For completeness we state the algorithm below. 

\begin{algorithm}[ht!]
	\SetAlgoLined
	\KwIn{A hypergraph $H = (V,E)$ of arity $d$, volume parameter $\delta \in (0,1/2]$, contant $C^*_0$, and an error parameter $t$.} 
	Solve the following $R \defeq (t+2)$-round Lasserre relaxation of 
	\begin{align*}
		\min & \frac{1}{\delta W(V)} \Ex_{S \subseteq_t V_G} \Ex_{X_S \sim \mu_S} \sum_{e \in E}w(e) \max_{i,j \in e}\Pr_{(X_i,X_j) \sim \mu_{ij}|X_S}\left[X_i \neq X_j\right] \\
		& \E_{i \sim V}W(i) \Pr_{X_i \sim \mu_i | X_S \gets \alpha} \Big[X_i = 1\Big] \leq 2  \delta d \sum_{e \in E} w(e) & \forall (R-1)\textnormal{-admissible} (S,\alpha) \\
		& |\mu_{i,1} - \mu_{j,1}| \leq \Pr_{\mu_{i,j | X_S \gets \alpha}}\left[X_i \neq X_j\right] & \forall i,j \in e, \forall e \in E 
	\end{align*}
	Let $\{u_{S,\alpha}\}$ be the optimal vector solution to the above, and let $r \leq n^{O(R)}$ denote the dimension of these vectors\;
	Sample a uniformly random subset $S \subseteq V$ of size $t$ and an assignment $X_S \gets \alpha \sim \mu_S$\;
	Let $\{u_i\}_{i \in V}$ be the denote the set of vectors corresponding to the degree-$2$ SoS psuedodistribution conditioned on $X_S \gets \alpha$ 
	\label{step:condition1} \;
	Let $\nu_e = \max_{i,j \in e} \Pr_{\mu_{ij}| X_S \gets \alpha} \left[X_i \neq X_j\right]$ be the contribution from $e \in E$,
	and let $\nu^*\sum_{e \in E} w(e) = \sum_e w(e) \nu_e$ denote the SDP objective value w.r.t. distribution conditioned on $X_S \gets \alpha$\;
	Delete all the hyperedges with $\nu_e \geq 1/10$ and consider them as cut \label{step:e-del1} \linebreak
	{\bf New Vector Solution}. \\
	Construct a new vector solution $\{v'_i\}_{i \in V}$ as follows. 
	Let $u_i = \mu_i \uphi + z_i$ where $\uphi$ is the ``one vector'' and $\langle \uphi, z_i \rangle = 0$.
	For every $i \in V$, let $v_i \defeq (1 - 2\mu_i) \uphi - 2z_i$ \label{step:newvec1}\;
	Let $\hat{z}$ be a unit vector orthogonal to $\uphi, \{z_i\}_{i \in V}$,
	and let $\theta \gets \delta^{12}$.
	For every $i \in V$, define \label{step:vi'1} 
	\[
	v'_i \defeq \frac{(1 - 2\mu_i) \uphi - 2z_i - \theta \hat{z}}{\sqrt{1 + \theta^2}}
	\]
	{\bf Shifted Hyperplane Rounding} \label{step:gaussproj1} \\ 
	Write $v'_i = (1 - 2\mu'_i) \uphi - 2z'_i$ where $\langle \uphi,z'_i \rangle = 0$\;
	Sample Gaussian $g \sim N(0,1)^r$ and for every $i \in V$, define $\zeta_i = \langle g, \bar{z}'_i \rangle$\;
	For every $i \in T$, assign $y_i = \mathbbm{1}\left(\zeta \leq \Phi^{-1}(\mu'_i)\right)$\;
	Output subset $S = {\rm supp}(y)$.
	\caption{Approximation algorithm for \hsse} 
	\label{alg:hyper-sse}
\end{algorithm}

\subsection{Analysis of Algorithm \ref{alg:hyper-sse}}

Most of the claims and observations from the analysis of Algorithm \ref{alg:approx-pred} apply as is to this setting as well. We highlight the key conclusions that hold as is from proof of Theorem \ref{thm:hsse-main}.

To begin with, we bound the shift in the total weight due to the shifted biases. Let $\widehat{W}: V \to [0,1]$ denote the normalized weights defined as $\hat{W}(i) = W(i)/d \sum_{e \in E} w(e)$. Then, we can again use Claim \ref{cl:mu-bounds} to bound the shift in the total bias as
	\[
	\left|`\Ex_{i \sim \widehat{W}} \mu_i - \Ex_{i \sim \widehat{W}} \mu'_i \right| \leq \delta^{C^*_0}  
	\] 
	which implies that 
	\[
	0.99 \delta d \sum_{e \in E} w(e) \leq \sum_{e \in E} W(i) \mu'_i \leq 1.01 \delta d \sum_{e \in E} w(e)
	\]
	whenever $\delta$ is smaller than an absolute constant. 
	
Now we analyze the conditioning step. As before, using Lemma \ref{lem:conc}, with probability at least $0.9$ we have
	\[
	0.9 \delta d \sum_{e \in E} w(e) \leq \sum_{i \in S} W(i) \leq 1.1 \delta d \sum_{e \in E} w(e)
	\]
Again, for a random conditioning, using Markov's inequality, we have that with probability at least $0.9$, 
\[
\sum_{e \in E} w(e) \max_{i,j \in e} \|v_i - v_j\|^2 \leq 10 \cdot {\sf Sdp} 
\]
and consequently, using Claim \ref{cl:vprime-equal},
\[
\sum_{e \in E} w(e) \max_{i,j \in e} \|v'_i - v'_j\|^2 \leq 10 \cdot {\sf Sdp} .
\]

In particular, the above conclusions hold simultaneously with probability at least $0.8$.

\paragraph{Bounding the Expansion}

The key technical tool, Lemma \ref{lem:edge-cut}, carries over as is. For any hyperedge $e \in E$, let $\hat{\alpha}_e \defeq \min_{i \in e} \mu_i$. Then, analogous to Observation \ref{obs:alpha-bd}, using the $\ell_1$-constraints we have $\hat{\alpha}_e \leq \alpha_e + \nu_e$ and therefore, repeating the steps \eqref{eqn:step-1}-\eqref{eqn:step-2} from the previous analysis we can bound the probability of an edge getting cut in the rounding step as 
\[
\Pr_{{\sf Alg}} \Big[e \mbox{ is cut}\Big] \lesssim \left(C^*_0 \log(1/\delta)\right)\sqrt{\hat{\alpha}_e\nu_e \log(d)} + \tilde{O}(d) \nu_e(C^*_0\log(1/\delta)).
\]
We can then bound the expansion of the set $S$ as 
\begin{align*}
& \frac{\sum_{e \in E}\sqrt{\hat{\alpha}_e\nu_e \log(1/\delta)\log(d)}}{\sum_{v \in S} W(v)} + \tilde{O}(d) (\log(1/\delta))^2 \frac{\sum_{e \in E} w(e) \nu_e}{\sum_{v \in S}W(v)}  \\
& \lesssim \frac{\sum_{e \in E}\sqrt{\hat{\alpha}_e\nu_e \log(1/\delta )\log(d)}}{\sum_{v \in S} W(v)} + \tilde{O}(d) \phi^* (\log(1/\delta))^2 
\end{align*}

Define $\eta^H_{\rm max}$ as follows. For every hyperedge $e \in E$, let  $e^{\circ} \subseteq e$ be a choice of a non-empty subset of vertices. Then we have 
\[
\eta^H_{\rm max} = \min_{\{e^{\circ}\}} \max_{u \in V} \sum_{e: u \in e^{\circ}} \frac{\log|e^{\circ}|}{|e^{\circ}|}
\]
Then towards bounding the first term, we can use the following observation implicit in the proof of Theorem 6.3~\cite{lm16}

\begin{lemma}
We can bound	
	\[
	\frac{\sum_{e \in E}\sqrt{\hat{\alpha}_e\nu_e \log(d)}}{\sum_{v \in V} W(v)} \leq \sqrt{\eta^H_{\rm max} \phi^* \log(d)}
	\]
In particular, if $H$ is a $d$-uniform hypergraph with maximum degree $r$, then $\eta^{H}_{\rm max} = (r \log_2 (d))/d$.
\end{lemma}

Plugging in the bound from above lemma in place of the first term concludes the proof.

\subsection*{Acknowledgments}

AL was supported in part by SERB Award ECR/2017/003296, a Pratiksha
Trust Young Investigator Award, and an IUSSTF virtual center on “Polynomials as an Algorithmic
Paradigm”.

\bibliographystyle{alpha}
\bibliography{main}

\appendix
\section{Facts from Information Theory}				\label{sec:inf-theory}

As in \cite{RT12}, the proof of Lemma \ref{lem:conc} uses elementary information theory. For completeness, we introduce the (specific instantiations) of quantities and related properties used in the proof. We again point out that all the logarithms used in this section will be base $2$. Given a $0/1$ random variable $X$ with bias $\mu$, the entropy of $X$, denoted by $H(X)$, is defined as $H(X) = \mu\log(1/\mu) + (1 - \mu) \log(1/(1- \mu)) $. For a pair of random variables $X,Y$, the mutual information between $X$ and $Y$ is denoted by $I(X;Y)$ and is defined as $I(X;Y) = H(X) - H(X|Y)$ where $H(X|Y)$ is the conditional entropy of $X$ given $Y$. We now list the facts used by the proof of Lemma \ref{lem:conc}.

\begin{fact}		\label{fact:h-bound}
	For any $X \sim {\sf Bernoulli}(\mu)$ such that $\mu \leq 1/2$ we have  
	\[
	H(X) \leq 2 \mu \log\frac{1}{\mu} .
	\]
\end{fact}
\begin{proof}
	Using $\log(1 + \mu) \leq \mu$ we get that 
	\begin{eqnarray*}
	H(X)= \mu\log\frac1\mu + (1 - \mu)\log\frac{1}{1 - \mu}
	&\leq& \mu\log\frac1\mu + (1 - \mu)\log\left(1 + \frac{\mu}{1 - \mu}\right) \\
	&\leq& \mu\log\frac1\mu + \mu \leq  2\mu\log\frac1\mu .
	\end{eqnarray*}
\end{proof}

\begin{fact}			\label{fact:I-bound}			
	For any pair of boolean random variables $X,Y$ we have $I(X;Y) \leq 1$.
\end{fact}
\begin{proof}
	Using the fact that the entropy of boolean random variable lies between $[0,1]$ (for e.g., see Example 2.1.1~\cite{cover-thomas}), we have $I(X;Y) = H(X) - H(X|Y) \leq H(X) \leq 1$.
\end{proof}

\begin{fact}			\label{fact:gauss-mi}
	For any pair of Gaussian random variables $(g_i,g_j) \sim N(0,\Sigma)$ we have $I(g_i;g_j)  = - \frac12\log |\Sigma|$.
\end{fact}

\begin{lemma}[Data Processing Inequality]			\label{lem:data-proc}
	Let $(X,Y)$ and $(X',Y")$ be pairs of random variables such that $X$ is completely determined by $X'$ and $Y$ is completely determined by $Y'$. Then $I(X;Y) \leq I(X';Y')$
\end{lemma}

\section{Integrality Gap}		\label{sec:int-gap}

Here we show that the integralilty gap of the lifted SDP is still $\Omega(d)$ for up to near-linear levels of SoS, as stated formally in the following proposition.

\begin{restatable}{reprop}{intgap}					\label{lem:int-gap}
	For every $d \geq 2$, there exists a $\delta = \delta(d)$ such that for $R = \Omega(n/\log(n))$, the $R$-level SoS lifting of the basic SDP for the $\delta$-HSSE problem has integrality gap $\Omega(d)$.
\end{restatable}

\begin{proof}
The gap instance is simple: it is just a $d$-uniform random hypegraph with a linear number of hyperedges. Formally, our construction is as follows. 

{\bf Construction}. Given $d,n \geq 2$, set $\delta = 1/d$, $k = \delta n$ and $r = C n \log(1/\delta)$ for some large constant $C$.
Let $H = (V,E)$ be a random hypergraph on $n$-vertices where $E = \set{e_1,\ldots,e_r}$ such that each $e_i $ is a uniformly random $d$-sized subset of $V$.

\paragraph{Soundness.} 
We first note that for any constant $c>0$, we have $(\delta n - c)/(n-c) \leq \delta$.
Fix a subset $S \subseteq V$ of size $k$. Then for any fixed $i \in [r]$,
\[
\Pr_H\left[e_i \in \partial_H(S)\right] \geq 1 - \delta^d - (1 - \delta)^d \geq 1 - e^{-1} - \delta^d \geq 1/4
\]
where the first inequality is using $d = 1/\delta$. Therefore, $\Ex |\partial_H(S)| \geq r/4$. Therefore, using Chernoff bound over the randomness of the choice of the hyperedges, we have that  
\[
\Pr_H\left[ |\partial_H(S)| \geq \frac{r}{8}\right] \geq 1 - e^{-r/64}
\]
and hence taking a union bound over all $\binom{n}{\delta n} \leq e^{n \log (1/\delta)}$ subsets of size $k$ we have that 
\[
\Pr_H\left[\forall \ S \in \binom{[n]}{k} : |\partial_H(S)| \geq \frac{r}{8}\right] \geq 
1 - \exp\left(n \log (1/\delta ) - r/64\right) \geq 1 - e^{-Cr/2} 
\]
for large constant $C$.

\paragraph{Completeness.} Consider the distribution over subsets described by the following process
\begin{itemize}
	\item {\em Choose a random $k = \delta n = n/d$ sized subset $S$ and set $X_i = \mathbbm{1}(i \in S)$ for every $i \in V$.}
\end{itemize}

The above is a mixture over integer solutions, and therefore admits a valid degree $R$-pseudo-distribution. Fix a $R$-sized subset $S$, and let $E_S$ denote the set of hyperedges incident on $S$ i.e, 
\[
E_S := \left\{e \in E | e \cap S \neq \emptyset \right\}
\]
Now it is easy to see that 
\[
\Pr_{e = (i_1,\ldots,i_d)\sim \binom{[n]}{d}} \left[e \in E_S\right] \leq \sum_{j \in [d]} \Pr \left[i_j \in S\right] \leq \frac{d R}{n}
\]
and hence the expected number of hyperedges incident on $S$ is at most $d R r/n \lesssim  d R \log(1/\delta)$ (since $r = O(n \log(1/\delta))$ by our choice of parameters). Hence using Chernoff bound we have that 
\[
\Pr \left[|E_S| \geq  dR \log(1/\delta)\sqrt{\log n}\right]
\leq \exp\left(-\frac{\log(n)\log(1/\delta) dR}{4}\right). 
\]
Taking a union bound over all $\binom{n}{R} \leq e^{R \log n}$ choices of $R$-sized subsets we get that 
\begin{eqnarray*}
	\Pr \left[\forall \ S \in \binom{n}{R} : |E_S| \leq d R \log(1/\delta)\sqrt{\log(n)} \right]
	&\leq& \exp\left(R \log n- \frac{\log(n)\log(1/\delta) dR}{4}\right)  \\
	&\leq& \exp\left(- \frac{\log(n)\log(1/\delta) dR}{8}\right)
\end{eqnarray*} 
Now choose an  $R'$-sized subset $T$ with $R' \leq R-1$ and and a conditioning $X_T \gets \alpha$. Suppose the conditioning fixes $r'$ out of $R'$ variables to $1$. Then then the degree $R- R'$ distribution on the variables $V \setminus T$ conditioned on $X_T \gets \alpha$ corresponds to the following.
\begin{itemize}
	\item {\em Choose a uniformly random $(n/d - r')$-sized subset $S$ from $V \setminus T$ and set $X_i = \mathbbm{1}(i \in S)$}.
\end{itemize}  
In particular, for any $i \in V \setminus T$ we have that 
\[
\Pr_{X_i \sim \mu_i | X_T \gets \alpha}[X_i = 1] = \frac{\frac{n}{d} - r'}{n - R'} \leq \frac{2}{d}
\]
as long as $R' \leq n/2$. From the above, we can conclude that for any pair of variables $i,j \in V \setminus T$
\begin{equation}			\label{eqn:cross}
\Pr_{\mu_{ij}|X_T \gets \alpha}\left[X_i \neq X_j\right] \leq \Pr_{\mu_i | X_T \gets \alpha}\left[X_i = 1\right] 
+ \Pr_{\mu_j | X_T \gets \alpha}\left[X_j = 1\right] = \frac{4}{d}
\end{equation}
Therefore, in summary, for any $R$-sized subset $T$, and a conditioning of the variables in $T$ we have the following.  
\begin{enumerate}
	\item The number of hyperedges incident on $T$ is at most  $d R \log(1/\delta)\sqrt{\log (n)}$ for which the contribution to the objective value is at most $1$.
	\item For any hyperedge not intersecting with $T$, from \eqref{eqn:cross} we know that the contribution to the objective is at most $4/d$.
\end{enumerate}
Hence, the total value of the objective for any conditioning on any $R$-sized subset is at most
\[
d R \log(1/\delta) \sqrt{\log (n)} + \frac{4r}{d} \leq \frac{8r}{d}
\]
which completes the soundness analysis.
\end{proof}

\section{Reduction from \smallsetvertexexpansion~to \hypersse}		\label{sec:redn}

We introduce some additional notation introduced in this section. For a graph $G = (V,E)$, and a subset of vertices $ S \subset V$, we define the symmetric vertex boundary as $\partial^{\rm sym}_G(S) = \partial^V_G(S) \cup \partial^V_G(S^c)$. Furthermore, the internal boundary of the set $S$ is defined as $\partial^V_{\rm int}(S) = \partial^V_G(S^c)$. The symmetric vertex expansion of $S$, denoted by $\Phi^{V}_G$ is defined as 
\[
\Phi^V_G(S) = \frac{w(\partial^{\rm sym}_G(S))}{w(S)}
\]

\begin{lemma}				\label{lem:red1}
	Given a vertex weighted graph $G = (V,E,w)$ with vertex weights $w:V \to \mathbbm{R}_{\geq 0}$, there exists an efficient procedure to construct a vertex weighted graph $G' = (V',E',w')$ for which the following properties holds:
	\begin{itemize}
		\item {\bf Completeness}: If there exists a subset $S \subset V$ such that $w(S) = \delta$ and $\phi^V_G(S) \leq \epsilon$, then there exist a set $S' \subset V'$ such that $w'(S') = \delta$ and $\Phi^{V'}_{G'}(S') \leq 2\epsilon$.
		\item {\bf Soundness}: If there exists a subset $S' \subset V$ such that and $w'(S') = \delta$ and $\Phi^{\rm sym}_{G'}(S) \leq \epsilon$, then there exists a set $S \subset V$ such that $(1 - \epsilon)\delta  \leq w(S) \leq \delta$ and $\phi^V_G(S) \leq \epsilon/(1 - \epsilon)$.
	\end{itemize}
	Furthermore, the maximum degree of the graph $G'$ is exactly the maximum degree of graph $G$.
\end{lemma}

\begin{proof}
	Given the graph $G = (V,E,w)$ we construct a new bipartite graph $G' = (V_L,V_R,w')$ as follows. Here we assign $V_L = V$ and $V_R = E$ i.e., the left vertex set is the set of vertices in $G$ and the right vertex set is the edge-set in $E$. Furthermore, for every vertex set $u \in V$ and every edge $e = (u,v) \in E$ we add an edge $(u,e) \in E$. For every $u \in V_L$ we assign $w'(u) = w(u)$ and for every vertex $e \in V_R$ we assign $w'(0)$. This completes the description of the graph $G'$. For brevity, we shall denote $V' = V_L \cup V_R$. Now we analyze the completeness and soundness of the reduction.
	
	{\bf Completeness}. Suppose there exists a subset $S \subset V$ such that $\phi^V_G(S) \leq \epsilon$ and $w(S) = \delta$. Construct the set $S'$ by including all the vertices in $S$ and all the edges incident on $S$ i.e., 
	\[
	S' = S \cup \Big\{(u,v) \in V_R \Big| u \in S \vee v \in S \}
	\]
	Since the edge vertices $e \in V_R$ have weight $0$, we have $w'(S) = w(S) = \delta$. Now we bound the symmetric vertex expansion of $S'$. Fix a vertex $v \in \partial^{\rm sym}_{G'}(S)$. Since if $v \in V_R$, we have $w'(v) = 0$ so it does not contribute anything to the weight of the symmetric boundary. Therefore we may assume that $v \in V_L = V$. Then we claim that $v \notin S'$. Otherwise if $v \in S'$, we must have $v \in S$ (since $S' \cap V = S$). Furthermore, if $v \in S$, then for every $e = (v,v') \in E$ we have $e \in S'$ (from the construction of $G'$) and hence $N_{G'}(v) \subset S'$ which contradicts $v \in \partial^{\rm sym}_{V'}(S)$. Therefore we must have $v' \notin S$. But then there exists edge-vertex $e = (u,v) \in S'$ for some $u \in S'$. This implies that $v \in \partial^V_G(S)$. Hence, we have 
	\[
	\Phi^{V'}_{G'}(S') = \frac{w'\left(\partial^{\rm sym}_{G'}(S')\right)}{w'(S')} 
	= \frac{w\left(\partial^{\rm sym}_{G'}(S' \cap V)\right)}{w(S)} 
	\leq \frac{w\left(\partial^{V}_{G}(S)\right)}{w(S)} = \phi^V_G(S) \leq \epsilon.
	\]
	
	{\bf Soundness}. Fix a subset $S' \subset V'$ such that $w'(S') = \delta$ and $\Phi^{V'}_{G'}(S') \leq \epsilon$. Let $S'' \defeq S' \setminus \partial^{\rm int}_{G'}(S')$, and define $S = S'' \cap V$. We claim that $\partial^V(S) \subset \partial^{\rm sym}_{G'}(S')$. To see this, fix a vertex $v \in \partial^V_G(S)$. Then there exists $u \in S \cap N_G(v)$. Now since $u \in S$ we also have $u \in S^c$. Now we consider the following cases. Now for contradiction, suppose $v \notin \partial^{\rm sym}_{G'}(S')$, then this combined with the fact that $v \notin S$ implies that $v \notin S'$. Now we consider two cases:
	
	{\bf Case (i)}. Suppose $(u,v) \notin S'$. Then $u \in \partial^{\rm int}_{G'}(S')$ which contradicts the fact that $u \in S$.
	
	{\bf Case (ii)}. Suppose $(u,v) \in S'$. Then $v \in \partial^{\rm sym}_{G'}(S')$, contradicting that $v \notin \partial^{\rm sym}_{G'}(S')$.
	
	Therefore, the above establishes that $\partial^V_{G}(S) \subseteq \partial^{\sym}_{G'} \cap V$. Furthermore, note that we have 
	\[
	w(S) = w'(S'') = w'(S') \setminus w\left(\partial^{\rm int}_{G'}(S')\right) \geq w'(S') - \epsilon w(S') \geq (1 - \epsilon)w'(S') = (1 - \epsilon)w(S') 
	\]
	and since $S \subset S' \cap V$ we have $w(S) \cap w(S' \cap V)$. Therefore, combining the above observations we get that  
	\[
	\phi^V_G(S) = \frac{w\left(\partial^V_G(S)\right)}{w(S)} \leq \frac{w\left(\partial^{\rm sym}_{G'}(S')\right)}{(1- \epsilon) w'(S')} = \frac{\Phi^{\rm sym}_G(S)}{1 - \epsilon}.
	\]
\end{proof}

\begin{lemma}				\label{lem:red2}
	Given a vertex weighted graph $G = (V,E,w)$ of maximum degree $d$, there exists an efficient procedure to construct a weighted hypergraph $H = (V,E',w',W)$ on the same vertex set with edge weights given by $w':E' \to \mathbbm{R}_{\geq 0}$ and vertex weights given by $W':V \to \mathbbm{R}_{\geq 0}$ such that the following holds. For every subset $S \subset V$ we have $w(S) = W(S)$ and $w'\left(\partial^{E'}_{H}(S)\right) = w\left(\partial^{\rm sym}_G(S)\right)$. Consequently, we have $\phi^E_H(S) = \Phi^{V}_G(S)$ for every subset $S \subset V$.
\end{lemma}

\begin{proof}
	Given a graph $G = (V,E,w)$, we construct the hypergraph $H = (V,E,W,w')$ on the vertex set $V$ as follows. For every vertex $v \in V$ we introduce a hyperedge $\cE(v) = \{v\} \cup N_G(v)$. Furthermore, for every hyperedge $\cE(v)$ we assign it weight $w'(\cE(v)) = w(v)$, and assign the vertex weight $W(v) = w(v)$. This concludes the construction of the hypergraph. From the construction, it follows that the max-arity of the hypergraph $H$ is the same as the maximum degree of $G$. Furthermore, for every $v \in V$ we have $w(v) = W(v)$ and hence $w(S) = W(S)$ for every subset $S \subset V$. Finally, we observe the following. Fix a vertex $v \in V$. Then,
	\begin{align*}
	v \in \partial^{\rm sym}_G (S) &\Leftrightarrow \exists \ u \in N_G(u) \mbox{ s.t. } \big\{u \in S, v \in S^c \} \vee \{u \in S^c , v\in S \} \\
								   &\Leftrightarrow \cE(v) \in \partial^{E'}_H(S)	
	\end{align*}	
	The above one-to-one correspondence directly implies for any subset $S \subset V$ we have 
	\[
	w'\left(\partial^{E'}_H(S)\right) = \sum_{v : \cE(v) \in \partial^{E'}_H(S)} w'(e) = \sum_{v : v \in \partial^V_G(S)} w(e) = w\left(\partial^V_G(S)\right)
	\]
	The above claims taken together immediately imply that $\phi^{E'}_H(S) = \Phi^{V}_G(S)$ for every subset $S \subset V$.
\end{proof}

\end{document}